\def\oakland{0}
\def\usenix{0}
\def\notes{0}
\def\tpdp{0}
\def\submission{0}

\ifnum\oakland=1
    \newcommand{\noakland}[1]{}
    \newcommand{\yoakland}[1]{#1}
\else
    \newcommand{\noakland}[1]{#1}
    \newcommand{\yoakland}[1]{}
\fi

\ifnum\oakland=0
\documentclass[11pt]{article}
\else
\documentclass[conference,compsoc]{IEEEtran}
\fi

\usepackage[utf8]{inputenc}
\usepackage{notes}
\usepackage{amssymb}

\usepackage{hyperref}
\usepackage[utf8]{inputenc}
\usepackage[T1]{fontenc}
\ifnum\tpdp=0
\usepackage[letterpaper,margin=1in]{geometry}
\else
\usepackage[letterpaper,margin=0.95in]{geometry}
\fi

\usepackage{amsmath}
\usepackage{amsthm}
\usepackage{amssymb}
\usepackage{amsfonts}
\usepackage[normalem]{ulem}
\ifnum\usenix=0
\usepackage[numbers]{natbib}
\else
  \usepackage[nocompress]{cite}
\fi

\ifnum\oakland=1
\usepackage{array}
\fi

\usepackage{xspace}
\usepackage{hyphenat}
\usepackage{graphicx}
\usepackage[dvipsnames]{xcolor}
\usepackage{hyperref, bbm}
\usepackage[inline,shortlabels]{enumitem}
\newcommand{\msf}{\mathsf}
\usepackage{cmll}
\usepackage{mathtools}
\usepackage{algorithm}          
\usepackage{algpseudocodex}     
\usepackage[inline,nomargin,multiuser]{fixme}
\usepackage{tikz}
\usepackage{cleveref}
\renewcommand{\cref}{\Cref}

\usepackage{thm-restate} 
\usepackage{booktabs}

\newcommand{\examplenameone}{Asahi\xspace} 
\newcommand{\examplenametwo}{Blair\xspace} 
\newtheorem{lemma}{Lemma}
\newtheorem{theorem}{Theorem}

\newtheorem{claim}{Claim}

\newtheorem{definition}{Definition}

\newcommand{\zo}{\{0,1\}}
\newcommand{\N}{\mathbb{N}}

\newcommand{\E}{\mathbb{E}}
\newcommand{\setrandomly}{\ensuremath{\xleftarrow{\$}}}

\DeclarePairedDelimiter\abs{\vert}{\rvert}

\newlist{todolist}{itemize}{2}
\setlist[todolist]{label=$\square$}

\newcommand{\cL}{\mathcal{L}}

\newcommand{\cA}{\mathcal{A}}
\newcommand{\cD}{\mathcal{D}}
\newcommand{\cC}{\mathcal{C}}
\newcommand{\cF}{\mathcal{F}}
\newcommand{\cG}{\mathcal{G}}

\newcommand{\cE}{\mathcal{E}}

\newcommand{\cO}{\mathcal{O}}

\newcommand{\cX}{\mathcal{X}}
\newcommand{\cY}{\mathcal{Y}}
\newcommand{\alg}{\cA}
\newcommand{\wdist}{\dist_{\mathcal{W}_1}}

\newcommand{\eps}{\varepsilon}

\newcommand{\selection}{\rho }

\newcommand{\univ}{\mathcal{X}}

\newcommand{\coherenceExp}{\mathsf{DemCoh}}

\newcommand{\dist}{\msf{dist}}

\definecolor{azure}{rgb}{0,0.49,1}

\ifnum\notes=1
\newcommand{\gsk}[1]{\textsf{{\textcolor{azure}{[[GPK: #1]]}}}}
\newcommand{\gsknote}[1]{{\color{azure}\footnote{{\color{azure} {\bf GPK:} #1}}}}
\else
\newcommand{\gsknote}[1]{}
\newcommand{\gsk}[1]{}
\fi

\definecolor{mypurple}{rgb}{0.58,0.23,0.94}

\ifnum\notes=1
\newcommand{\pj}[1]
{{\color{mypurple}#1}}
\newcommand{\pjinline}[1]{{\color{mypurple} #1}}
\newcommand{\pjnote}[1]{{\color{mypurple}\footnote{{\color{mypurple} {\bf P:} #1}}}}
\newcommand{\pjtodo}[2]{{\color{mypurple}\sf
{\bf[[P:} #1{\bf]]} #2}
\footnote{{\color{mypurple} {[[\bf P:} #1 {\bf]]}}}
}
\else
\newcommand{\pj}[1]{{\color{black}#1}}
\newcommand{\pjinline}[1]{}
\newcommand{\pjnote}[1]{}
\newcommand{\pjtodo}[2]{}
\fi

\ifnum\notes=1
\newcommand{\mc}[1]
{{\color{orange}\sf [[MC: #1]]}}
\newcommand{\mcinline}[1]{{\color{orange} #1}}
\newcommand{\mcnote}[1]{{\color{orange}\footnote{{\color{orange} {\bf MC:} #1}}}}
\else
\newcommand{\mc}[1]{}
\newcommand{\mcinline}[1]{}
\newcommand{\mcnote}[1]{}
\fi

\ifnum\notes=1
\newcommand{\mb}[1]
{{\color{red} {\bf MB:}  #1}}
\newcommand{\mbnote}[1]{{\color{red}\footnote{{\color{red} {\bf MB:} #1}}}}
\else
\newcommand{\mb}{\mbnote}
\newcommand{\mbnote}[1]{}
\fi

\ifnum\notes=1
\newcommand{\sstext}[1]
{{\color{purple}#1}}
\newcommand{\ssnote}[1]{{\color{purple}\footnote{{\color{purple} {\bf SS:} #1}}}}
\else
\newcommand{\sstext}[1]{{{\color{black}#1}}}
\newcommand{\ssnote}[1]{}
\fi

\ifnum\notes=0
\renewcommand{\sout}[1]{}
\fi

\newcommand{\eg}[0]{e.g.,\xspace}
\newcommand{\ie}[0]{i.e.,\xspace}

 \newcommand{\cdf}{\msf{cdf}}
    \newcommand{\distw}{\dist_{\mathcal{W}_1}}
    \newcommand{\wass}{\distw}

\DeclareMathAlphabet{\mathbbold}{U}{bbold}{m}{n}

\newcommand{\coherent}{\xspace demographically coherent\xspace}

\newcommand{\restrict}[2]{\raisebox{0.17ex}{$\mathop{\boldsymbol{\large\pi}_{#2}}$}\left(#1\right)}

\newcommand{\defeq}{\stackrel{\text{def}}{=}}

\newcommand{\ignore}[1]{{ }}
\newcommand{\bigmid}{\,\,\,\Big\vert\,}

\DeclareMathOperator*{\Expectation}{\mathbb{E}}
\newcommand{\Ex}[2]{\Expectation_{#1}\left[#2\right]}
\DeclareMathOperator*{\Probability}{\mathrm{Pr}}
\newcommand{\prob}[1]{\mathrm{Pr}\left[#1\right]}
\newcommand{\Prob}[2]{\Probability_{#1}\left[#2\right]}

\newcommand{\bbx}{\mathbf{x}}
\newcommand{\bX}{\mathbf{X}}

\newcommand\numberthis{\addtocounter{equation}{1}\tag{\theequation}}

\newcommand\allbold[1]{{\boldmath\textbf{#1}}}

\newcommand{\mybox}[4]{
\ifnum\oakland=1
\begin{figure}[t]
\else
\begin{figure}[H]
\fi
\begin{center}
\fcolorbox{black}{mypurple!3}
{
\small
\hbox{\quad
\ifnum\oakland=1
\begin{minipage}{.8\columnwidth}
\else
\begin{minipage}{6in}
\fi
\vspace{0.3em}
\begin{center}
{\allbold{#1}}
\end{center}
#4
\vspace{0.2em}
\end{minipage}
}
}
\caption{\label{#3} #2}
\end{center}
\end{figure}
\vspace{-1em}
}

\ifnum\oakland=0
\ifnum\notes=1
\usepackage{showlabels}

\fi
\fi

\fxusetheme{signature}
\FXRegisterAuthor{mlc}{amlc}{Marco}

\title{Enforcing Demographic Coherence: A Harms Aware Framework for Reasoning about Private Data Release}

\ifnum\oakland=0
\author{Mark Bun, Marco Carmosino, Palak Jain, Gabe Kaptchuk, Satchit Sivakumar}
\else
\author{
{\rm Anonymous Author(s)}\\
Affiliation
}
\fi

\ifnum\submission=0
 \author{Mark Bun\thanks{Department of Computer Science, Boston University. \texttt{\{mbun,palakj,satchit\}@bu.edu}.} 
    \and Marco Carmosino \thanks{IBM Yorktown. \texttt{marco@ntime.org}} 
    \and Palak Jain\footnotemark[1]
    \and Gabe Kaptchuk \thanks{Department of Computer Science, University of Maryland, College Park. \texttt{kaptchuk@umd.edu}.}  
    \and Satchit Sivakumar\footnotemark[1]}
\else
\author{
{\rm Anonymous Author(s)}\\
Affiliation
}
\fi

\date{ }

\usepackage{etoolbox}
\AtBeginEnvironment{proof}{\setlength{\parskip}{6pt} \setlength{\parindent}{0pt}}
\allowdisplaybreaks
\begin{document}

\maketitle

\begin{abstract} 


The technical literature about data privacy largely consists of two complementary approaches: formal definitions of conditions \pj{``sufficient''} for privacy preservation and attacks that demonstrate privacy breaches. Differential privacy is an accepted standard in the former sphere.  However, differential privacy's powerful adversarial model and worst-case guarantees may make it too stringent in some situations, especially when achieving it comes at a significant cost to data utility\gsk{is this really what we want to say? seems pretty anti-dp I think?}. 
Meanwhile, privacy attacks aim to expose real and worrying privacy risks associated with existing data release processes but often face criticism for being unrealistic. Moreover, the literature on attacks generally does not identify what properties are necessary to defend against them.

We address the gap between these approaches by introducing \emph{demographic coherence}, a condition inspired by privacy attacks that we argue is necessary for data privacy. This condition captures privacy violations arising from inferences about individuals that are incoherent with respect to the demographic patterns in the data. Our framework focuses on  confidence rated predictors, which can in turn be distilled from almost any data-informed process. Thus, 
\pj{it captures}
privacy threats that exist even when no attack is explicitly being carried out. \pj{This}
framework not only provides a condition with respect to which data release algorithms can be analyzed but suggests natural experimental evaluation methodologies that could be used to build practical intuition and make tangible assessment of risks. Finally, we argue that demographic coherence enforcement is \pj{less stringent}
than differential privacy: we prove that every differentially private data release is also demographically coherent, and we show that there are demographically coherent algorithms which are not differentially private. 
\end{abstract}

\ifnum\oakland=0
\newpage  
\tableofcontents
\newpage
\fi


\section{Introduction}\label{sec:intro}
\ssnote{Mention the appendix on non DP algorithms enforcing coherence somewhere in the intro. Perhaps in technical overview?}

\footnote{The code used in this paper can be found in \href{https://anonymous.4open.science/r/Enforcing-Demographic-Coherence-A-Harms-Aware-Framework-for-Reasoning-about-Private-Data-Release-C7E0/README.md}{this repository}.}
Inspired by high-profile demonstrations of the limitations of heuristic approaches to data privacy \cite{sweeney,transyouthintexas,censusreconstruction,censusreconstruction2,USENIX:Cohen22}, there is an emerging consensus among researchers and practitioners that formal data privacy approaches are a necessary component of responsible data release pipelines.  Differential privacy (DP) \cite{DworkMNS16j} in particular has become the accepted standard guarantee, seeing deployment in both industry \cite{CCS:ErlPihKor14,apple2017,ding2017collecting,tezapsidis2017uber} and government \cite{abowd2018us}.  Formal approaches are often treated as \emph{sufficient} for ensuring data subject's privacy, 
in that they are ``one-size-fits-all,'' \ie data curators can apply best practice protections without needing to consider the intricacies of each deployment.

In practice, however, there are significant barriers to making most formal privacy notions an easy ``apply and forget'' process, stemming from the fact that, unlike encryption, the privacy guarantees are not binary and require striking a delicate balance between the benefits of privacy preservation and the cost to accuracy~\cite{amin2024practicalconsiderationsdifferentialprivacy}. Within the context of DP in particular, this tension  is represented by the selection of parameters $\eps$ and $\delta$ which modulate the accuracy-privacy tradeoff. Reasoning about this tradeoff is a central challenge in the deployment of DP, in no small part due to the fact that it can be unclear what the necessary level of privacy is and what a failure of DP would even mean.  Moreover, the generality of DP means that it is naturally very abstract, making it easy to lose sight of the concrete privacy harms against which DP is intended to protect~\cite{CummingsHSS24}.
This confluence of factors can give rise to a common critique of DP: the lack of clear guidance \sstext{\sout{during parameter selection} about how privacy parameters tie to intuitive privacy risks} makes it unclear what ``good'' parameter choices are, leading some practitioners to make choices that are seen as overly conservative while others make overly permissive choices that fail to concretely protect against concerning attacks.


To more clearly illustrate the problem, consider a data curator charged with managing a privatized data release.  Having determined that DP is the ``right'' guarantee for their setting, they must now pick parameters for a data release that maximize accuracy without risking harm to data subjects.  Looking at existing deployments, they will find a wide variety of parameter choices in existing deployments, offering little direct guidance; e.g., Khavkin and Toch \cite{KhavkinToch25} document industry deployments with $\eps$ choices between 0.1 and >10 (recall that $\eps$ is in the exponent of the privacy-accuracy tradeoff).  In the hopes of making the harms against which they must protect more concrete, the curator might study the known attack methodologies \cite{sweeney,transyouthintexas,censusreconstruction,censusreconstruction2,USENIX:Cohen22} and the results of testing
concrete DP mechanisms against state-of-the-art attack implementations \cite{JayaramanE19, JagielskiUO20, NasrH0BTJCT23, JagielskyNS23}. \pj{While such \sstext{evaluations \sout{this highlights the risks of data sharing in general, and testing results} may be used to justify the deployment of formally private mechanisms., they does not indicate what privacy parameters are \textit{necessary} to to defend against a given class of attacks: the failure of standard attack implementations does not preclude more sophisticated variants from succeeding.}}
\gsk{happy with this, but I think we need to make the anticedents more clear. ``what is necessary'' -- does this refer to parameters? Lets make sure this is very very concrete and specific \pj{yeah, in the context of DP specifically, we're talking about the  parameters necessary}}
%
Moreover, recontextualizating these attacks to the curator's context can be challenging (e.g., \emph{What auxilliary information do adversaries have? What are the harms associated with failure?}). \sstext{\sout{and there is no promise that attacks will not improve with time.}}
As such, it is not clear if the results of these experiments provide actionable guidance within the curator's context or how to operationalize these results \sstext{\sout{if the data contained within a data release will remain meaningfully sensitive for a long time} given that future advances in attack methodology may compromise information that remains sensitive.}
\gsk{`if the data will remain meaningfully sensitive for a long time' was supposed to point to the attacks improving with time, so over a long time there may be a significant improvement in attacks.  Could be more clear.}
Left without any other support, the curator will inevitably be reduced to making a heuristic decision, possibly choosing to be highly conservative in order to account for uncertainty (sacrificing accuracy) or defaulting to loose parameters that maximize accuracy.   Note that while we described this process for DP, these challenges will inevitably arise for any privatization method with a privacy-accuracy tradeoff.


\smallskip
\noindent
\textbf{Necessary privacy conditions.}  Caught between abstract
privacy conditions and the heuristic nature of existing attack work, we explore an intermediary design philosophy:  \emph{necessary} conditions.  
Within this approach, we can formally define a set of properties that any private data release should guarantee without needing to provide a single, unifying, sufficient condition. These necessary conditions can be seen as giving formal procedures for recognizing when an attacker has inflicted a particular type of harm. Critically, because the necessary conditions need not be all-encompassing, they can be more easily contextualized allowing for more clear application of domain experts' intuition about acceptable levels of risk.  At the same time, specifying necessary conditions with mathematical rigor (and universal quantifiers, in particular) embraces the formality of sufficient conditions and supports rigorous analysis.

Of course, thinking in terms of necessary conditions has always been implicit in the practice of parameter selection, albeit, usually informally.  In the process described above, the data curator might have leveraged the literature on attacks to set an internal understanding of some necessary conditions.  For example, the curator might identify membership inference \cite{Homer08} as an attack against which they would want to protect. Published attacks (e.g., \cite{JagielskyNS23, Homer08, DworkSSUV15, DworkSSU17}), however, do not provide much support as the curator translates this implicit intuition into decisions. \sstext{In particular, it is unclear how parameters should be chosen to prevent these attacks, and successfully defending against canonical membership-inference attacks does not guarantee protection against novel membership inference attacks developed in the future.}

In contrast, by giving the curator a set of explicit and formal necessary conditions, \sstext{and establishing connections between these conditions and the parameters appearing in formal privacy definitions, such as differential privacy,}
we can hope to generate a concrete methodology for justifying parameter choices by identifying parameter regimes that could enable specific harms. 


\vspace{0.5em}
\noindent
\textbf{A new necessary condition: Demographic Coherence.} 
We take a step towards this reality by developing a new necessary condition. Our approach is rooted in three key insights: 
(1) privacy harms are increasingly going to come in the form of inferences at the hands of predictive algorithms.\footnote{In this work we intentionally use the term ``algorithm'' broadly to capture, \eg informal decision-making process made by humans that might not be explicitly codified as traditional algorithms.} That is, we should be interested in the predictions that these algorithms make about people---and the decisions organizations may make based on these predictions---even when predictive algorithms are not intentionally designed with causing harm in mind; (2) we should consider the confidence with which an algorithm can make predictions, because simply increasing the \emph{confidence} that an individual or a group has a certain attribute may be enough to result in harm;
and (3) the harms associated with breaches of privacy are not experienced uniformly among members of a population. This means that, if not defined carefully, an aggregate measure across an entire population could easily ``hide'' effects on vulnerable subgroups by averaging them away.

The resulting notion, which we call \emph{demographic coherence}, is intentionally designed to be \emph{ergonomic}\footnote{We use ergonomic in this context to mean ease of use by many different stakeholders.  We intentionally move away from the term ``usable,'' as this typically focuses only on end-users and we are interested in ease of use from a more diverse set of communities.} in many different contexts. For example, we provide sufficient formalism to enable rigorous analysis and provable realization, all while keeping the specific harms against which demographic coherence protects compellingly  salient.  Additionally, we provide a vision as to how demographic coherence can support the type of intuition building required to set real-world parameters. Critically, as a necessary condition, we imagine demographic coherence as one of a set of conditions which a data curator might want to jointly consider when making deployment decisions\sstext{: demographic coherence protects against one type of privacy harm, and other necessary conditions might protect against other harms}. 

\ifnum\tpdp=1
A discussion of related works can be found in \Cref{sec:related}.
\fi

\subsection{Our Contributions}

\noindent
In this work we make the following contributions:

\smallskip
\noindent
-- \textbf{New framework of demographic coherence enforcement.} We introduce 
\emph{demographic coherence enforcement}, 
an analytical framework for reasoning about the privacy provided by data release algorithms. It has the following qualities: \gsk{Is it worth saying that our notion is best understood as a group privacy notion?  So folks don't get freaked out when they encounter this fact? \sstext{\textbf{Satchit:} I think it might be worth mentioning it, yes.} \pj{PJ: which gives us an excuse to hint that this necessarily takes the perspective of a data processor preventing harms and concretely moves away from assuming that individuals can choose whether or not to participate}}
\begin{itemize}
\item[--] \emph{Captures predictive harms.} Demographic coherence \sstext{captures a broad class of privacy harms associated with data releases. It considers predictors built using these releases and compares the behavior of a predictor on the sample from which it was derived with its behavior on held out samples from the population.} It builds on conceptual tools from generalization ~\cite{vapnik1971, DworkFHPRR15, CummingsLNRW16, ImpLPS22, BunGHILPSS23} and multicalibration~\cite{HebertJohnsonKRR18, KimGZ19}. \sstext{Two features are particularly notable: (1) demographic coherence evaluates the behavior of predictors without requiring accuracy to be measured against a ground truth, which may be unknown in many settings (for example, in a semi-supervised learning setup); and (2) it evaluates the risk of predictive harms within subgroups of a population, which is crucial for identifying harms concentrated among vulnerable groups. \sout{
(1) evaluate the risk of predictive harms distributionally without measuring accuracy with respect to an unknown (and possibly unknowable) ground truth and (2) evaluate the risk of predictive harms local to the different subgroups within a population. Evaluating risks distributionally allows the framework to remain applicable even when ground truth is unavailable, and evaluating risks for different subgroups allows the framework to identify effects specific to vulnerable subgroups.}}
\item[--] \emph{Lends itself to experimental auditing.} Demographic coherence has a natural translation to an experimental setup for comparing the effects of various algorithms for privacy preserving data release. 
In addition, demographic coherence is measured by computing a distance metric over two distributions, which facilitates quantification of the concrete risk. We focus on demographic coherence when instantiated with Wasserstein distance.

\item[--] \emph{Lends itself to analytical arguments.} Finally, the formalism we build supports rigorous analytical arguments about algorithms. For example, 
we show that all algorithms with bounded max-information, and therefore all differentially private algorithms, are also demographic coherence-enforcing.
\end{itemize}


\smallskip
\noindent
-- \textbf{Demographic coherence enforcement is achievable.} 

 In \Cref{sec:acheivingdc} and \Cref{sec:nondp}, we describe several ways to achieve demographic coherence:
 \begin{itemize}
    \item[--] In \Cref{sec:acheivingdc}, we show that differentially private algorithms automatically provide a degree of demographic coherence. Specifically, we give parameter conversions under which every pure differentially private (pure-DP) algorithm and every approximate differentially private (approx-DP) algorithm enforce demographic coherence. We describe a methodology to use these results to select parameters for differential privacy in \Cref{sec:parametersetting}.
    \sstext{\item[--] In \Cref{sec:nondp}, we describe a mechanism that is not differentially private ($\epsilon=\infty$) but nevertheless achieves strong demographic coherence parameters. This result suggests that, in some settings, demographic coherence may be provably easier to achieve than differential privacy, which constitutes an intriguing direction for future exploration.}
\end{itemize}

\noindent
\sstext{
-- \textbf{Demographic coherence can expose privacy vulnerabilities:} 
In \Cref{sec:reallifevul}, we give a simple example illustrating how the demographic coherence framework can capture concrete privacy vulnerabilities. We consider a setting in which a government releases data to facilitate the study of disability while seeking to prevent the identification of particular individuals. The government anonymizes the data by removing names and locations but retains several other attributes, including cognitive-disability status; we model this setting using a real-world dataset.

    We consider a natural algorithm that learns a predictor on the released dataset. Under this predictor, an individual with a cognitive disability is substantially more likely to be classified as having a cognitive disability when they are included in the released dataset than when they are not. Inclusion in the release may therefore expose such individuals to harms, including employment discrimination -- even when an explicit privacy attack is not being carried about! As discussed further in \Cref{sec:reallifevul}, demographic coherence captures this vulnerability: this data-release mechanism has very bad demographic coherence parameters. This example provides evidence that demographic coherence can identify meaningful new privacy vulnerabilities arising in practical scenarios.}

\section{Overview of Framework}

In this section, we give a brief overview of our framework. A more detailed and nuanced discussion of the choices underlying the framework, together with a formalization, can be found in \Cref{sec:walkthrough}.

\paragraph{Capturing the privacy harm.}

We study privacy in the context of large-scale public data releases, such as those produced by the Census Bureau. In this setting, a data curator has collected rich information about a population and seeks to release it in a form that supports population-level analysis. Broad access to this data is socially valuable: it can help governments allocate resources, enable businesses and other organizations to better understand the communities they serve, and allow researchers to generate evidence that informs public policy in areas such as health and education. At the same time, the richness that makes these data useful can create privacy risks for the individuals represented in them.

One important downstream use of large data releases is as input to learning algorithms that produce predictors used to inform decisions about individuals, an increasingly common practice as algorithmic tools play a growing role in decision-making. Local governments already use predictors in administering social services; for example, Los Angeles County uses predictive modeling to identify individuals at elevated risk of homelessness and connect them with case management, financial assistance, and other support services \cite{BlackwellEtAl2024}. Banks may similarly use predictive models to inform lending decisions. For example, Bancamía, a Colombian bank focused on small-business lending, has used predictive models to construct credit-risk scores that are used in deciding whether to approve loans and on what terms \cite{ParavisiniSchoar2015}.

In this setting, we identify a class of \emph{predictive privacy harms}, motivated by the following principle: an individual's participation in a study should not, by itself, materially change how that individual is treated. When participation causes an individual to experience harm, that harm is naturally understood as a privacy harm. One way such a harm can arise is when a predictor behaves systematically differently on individuals who participated in the dataset than on otherwise comparable individuals who did not. Suppose, for example, that an employer uses a predictor learned from the released data to identify some attribute it regards as undesirable. If the predictor is substantially more likely to make such an inference about individuals in the dataset than about individuals outside it, then participation in the study can itself increase an individual's likelihood of receiving an adverse employment decision. Our framework seeks to rule out precisely this kind of participation-induced predictive harm.

The difficulty in formalizing this principle is that we cannot observe the relevant counterfactual: for an individual who participates in a data release, we cannot observe what would have happened had they been absent. We therefore consider the following thought experiment, which uses a group-level comparison as a proxy for this individual-level counterfactual and is designed to detect when inclusion in the data systematically increases individuals' likelihood of receiving adverse downstream treatment. \gsk{this jumps from the individual to groups.  It might be worth the extra sentence to tie it back to the individual?  Maybe also say something about sub-sampling as an intuitive, heuristic privacy protection? \sstext{attempted to create a better transition into groups.}} Randomly partition the dataset into two halves, and produce a predictor using only one of them. We then compare the distribution of its predictions on the training half with the distribution of its predictions on the held-out half. Because membership in the two halves is determined at random, the halves provide two comparable groups of individuals; the salient difference between them is that one group participated in training the predictor while the other did not. Thus, if participation in the training data does not materially affect how individuals are treated by the resulting predictor, its predictions should have similar distributions on the two halves.

As a useful benchmark, consider a predictor produced independently of the data split and subsequently evaluated on it as above. Such a predictor exposes no information about the data split and hence does not exhibit the privacy risk identified in the preceding paragraph. Indeed, when evaluated on a uniformly random partition, its empirical prediction distributions on the two halves are close with high probability; we formalize this in the proofs of \Cref{sec:acheivingdc}. This thought experiment motivates the following definition, which asks that a learning algorithm trained on one randomly selected half of a dataset produce similar distributions of predictions on the individuals whose data was used for training and those whose data was not used for training.

\begin{definition}[Coherent learning algorithms; Informal Version of Definitions~\ref{def:coherence}]\label{informaldef:coherence}
Consider a data universe $\univ$, a class of predictors $H \subseteq {h:\univ\to[-1,1]}$, and a learning algorithm $\cL:\univ^*\to H$ that maps datasets to predictors. We say that $\cL$ produces $(\alpha,\beta)$-demographically coherent predictions if, for every dataset $X$,
$$
\Pr_{\scriptsize \begin{gathered}X_a,X_b \setrandomly X, h \leftarrow \cL(X_a)\end{gathered}}[\dist(h(X_a), h(X_b)) \geq \alpha] \leq \beta,$$
where $X_a$ and $X_b$ form a uniformly random partition of $X$ into two equally sized halves, $h(X_a)$ and $h(X_b)$ denote the empirical distributions of the predictions made by $h$ on the respective halves, and $\dist(\cdot,\cdot)$ is a metric on distributions. 
\end{definition}

We note that we use the Wasserstein distance to instantiate the distance metric above (See \Cref{def:wasserstein-distance}). Detailed reasons for this choice are given in \Cref{sec:walkthrough}.

\paragraph{Defining an associated privacy condition for data releases.}

The discussion above identifies a particular class of privacy harms that can arise from training a predictor, and demographic coherence provides a condition for ruling out such harms for a given learning algorithm. We now turn to the perspective of a data curator, whose goal is to release useful information while preserving the privacy of the individuals represented in the data. We lift the learner-level condition of demographic coherence to a privacy requirement for large data releases that we term \emph{coherence-enforcement}: rather than fixing a particular predictor or learning algorithm, we ask whether the release itself prevents these harms under arbitrary downstream predictors derived from the data release.

To motivate this requirement, note that large-scale data releases may be used in many different ways, and the organizations releasing them cannot fully control how the data are used once they are made public. Indeed, the possibility of unanticipated and creative downstream uses is part of what makes such releases valuable in the first place. Consequently, a privacy definition for such releases cannot reasonably assume that the curator knows in advance which attributes will be predicted or which learning algorithms will be applied to the released data. We therefore require the release to prevent \emph{any} downstream learning algorithm from producing incoherent predictions.

Concretely, suppose that we randomly partition a dataset into two halves $X_a$ and $X_b$, run the data-curation algorithm on $X_a$ to produce a report $R$, and then allow an arbitrary learning algorithm to process $R$ into a predictor. We require that the resulting predictor have similar distributions of predictions on $X_a$ and $X_b$. Thus, while demographic coherence is a property of a particular learning algorithm, coherence enforcement is a property of the data-curation algorithm that requires demographic coherence to hold under every downstream learning algorithm that uses the release. \mb{\sstext{mark's old comment here: I changed it so this should be fine now. }I get that the point of this discussion is to `lift'' the idea of demographically coherent predictions to that of coherence enforcing curators, but the wording may make it seem circular to the casual reader. I.e., motivated by demographic coherence, we define demographic coherence as ... I think the contrast between demographic coherence'' and coherence-enforcing algorithm'' should be more prominent before the definition environment}

\begin{definition}[Coherence enforcement; Informal Version of Definition~\ref{def:coherenceEnforcing}]\label{informaldef:coherenceEnforcing}
Consider a data universe $\univ$ and a data-curation algorithm
$\alg:\univ^*\to\cY$. We say that $\alg$ enforces $(\alpha,\beta)$-demographic coherence if, for every learning algorithm
$\cL:\cY\to(\univ\to[-1,1])$ mapping the curator's report to a confidence-rated predictor, the composed algorithm $\cL\circ\alg$ produces $(\alpha,\beta)$-demographically coherent predictions.
\end{definition}

The informal definition above suppresses several additional features of our formal framework. We mention two here. First, comparing prediction distributions over the population as a whole can obscure harms concentrated within particular groups: differences affecting one subgroup may be diluted, or even canceled, when averaged over the entire population. We therefore strengthen the definition to require coherence within a specified collection of subgroups as well.

Second, even a predictor chosen independently of the data will generally have slightly different empirical prediction distributions on two finite random halves simply because of sampling error. This issue is particularly important when restricting attention to subgroups, which may be divided very unevenly between the two halves of the random partition. Our formal definition therefore incorporates appropriate size constraints, requiring coherence only for subgroups larger than a specified minimum size.  \mb{It feels easier to be honest to the technical definition by interchanging the order of these two modifications. That way, we can say that the ``appropriate size constraint'' applies to subgroups rather than just to datasets. \sstext{done now.}} We defer the precise formulation of these and other nuances, together with the considerations motivating our particular choices, to \Cref{sec:walkthrough} and \Cref{sec:formaldef}.

\paragraph{Comparison to generalization.} 
The reader with a machine learning background might observe that coherence enforcement is based on an experiment comparing the behavior of predictors on training and held-out data, and hence wonder if it is a variant of standard generalization notions. There are, however, several important differences:
\begin{itemize}
    \item \textbf{Type of object.} Standard generalization is typically a property of a learning algorithm: it asks whether a predictor learned from a sample behaves similarly on the sample and the underlying population. In contrast, coherence enforcement is a property of a \emph{data-curation algorithm}: it asks whether the released report can be processed by any downstream learning algorithm to produce incoherent predictions. 

    \item \textbf{No reference to accuracy.} Unlike standard generalization guarantees, which often compare a scalar quantity such as empirical loss or error between a sample and the population, coherence directly compares the \emph{distributions of predictions} made on participants and nonparticipants. 

    \item \textbf{Robustness to post-processing.}  Coherence enforcement quantifies over \emph{all} downstream learning algorithms applied to the released report. This is stronger than a standard generalization guarantee, which need not be preserved under post-processing: a learning algorithm that generalizes can have its output post-processed to obtain another learning algorithm that does not generalize \cite{CummingsLNRW16}. In this respect, coherence enforcement is closer in spirit to notions such as perfect generalization \cite{CummingsLNRW16}, which are designed to remain meaningful under arbitrary post-processing of an algorithm's output. These two definitions, however, compare different things: perfect generalization compares the distributions of the curator $\alg$'s reports when run on $X_a$ and when run on $X_b$, whereas demographic coherence involves running the algorithm only once on $X_a$ and compares a derived predictor's predictions on $X_a$ and $X_b$. We refer to \Cref{sec:comparePG} for more details.
\end{itemize}

\section{Overview of Technical Results}\label{sec:overview-of-technical-results}

We next summarize our main technical results and give an overview of their proofs. We first prove that differential privacy implies coherence enforcement, via an intermediate connection to max-information. We then give a separation in the opposite direction, exhibiting a natural data-release mechanism that is not differentially private but nevertheless satisfies strong coherence guarantees.

\paragraph{Differentially private algorithms achieve demographic coherence.}
In \Cref{sec:acheivingdc}, we give parameter conversions under which differentially private algorithms enforce demographic coherence. We first state informally our result for pure differential privacy; an analogous result for approximate differential privacy appears as \Cref{thm:approx-dp-implies-coherence-enforcement}.

\begin{theorem}[Differential Privacy $\implies$ Demographic Coherence, Informal Version of Theorem~\ref{thm:pure-dp-implies-coherence-enforcement}]\label{introthm:pure-DP-implies-demographic-coherence}
Let $n\in\N$, $\beta,\eps \in (0,1)$, and $\alpha \in (0,2]$. Consider an $\eps$-differentially private data-curation algorithm ${\alg:\univ^{n/2}\to \cY}$. Then $\alg$ enforces $(\alpha,\beta)$-demographic coherence provided that
\(        \eps \leq k \cdot \frac{\alpha}{\ln(1/\beta)}
   \)
for some constant $k$.
\end{theorem}

One way to interpret this theorem is as a tool for parameter selection. A data curator might first identify, possibly experimentally, values of $\alpha$ and $\beta$ corresponding to levels of demographic incoherence that they consider acceptable for a given release. The theorem then gives a sufficient choice of $\eps$ under which the use of differential privacy guarantees the desired coherence parameters. We discuss this interpretation, and parameter selection more generally, in \Cref{sec:parametersetting}.

The proof goes by way of \emph{max-information} \cite{DworkFHPRR15}, which controls how different the joint distribution of an algorithm's input and output can be from the distribution obtained by sampling them independently from their respective marginals. We first show that any data-curation algorithm with bounded max-information enforces demographic coherence.

\begin{theorem}[Informal Version of Theorem~\ref{thm:max-info-implies-demographic-coherence}]\label{introthm:max-info-implies-demographic-coherence}
Let $n\in\N$, $\zeta>0$, $\beta \in (0,1)$, and $\alpha \in (0,2]$. Consider a data-curation algorithm ${\alg:\univ^{n/2}\to \cY}$ with bounded max-information
\(        I^{\beta/2}_{\infty}(\alg,n/2) < \zeta.
   \)
Then $\alg$ enforces $(\alpha,\beta)$-demographic coherence provided that
\(        n \geq k\cdot\frac{\zeta + \ln(1/\beta)}{\alpha^2}
   \)
for some constant $k$.
\end{theorem}

The main idea behind this result is easiest to see in a hypothetical experiment in which the predictor $h$ is fixed \emph{before} the dataset is randomly partitioned into $X_a$ and $X_b$. In this case, $h$ is independent of the random partition. Recall that the Wasserstein distance between two distributions can be expressed as the integral of the absolute difference between their cumulative distribution functions. For any fixed threshold, the number of points assigned to $X_a$ whose predictions lie below that threshold is hypergeometrically distributed, and similarly for $X_b$. Concentration bounds for hypergeometric random variables therefore imply that the two empirical CDFs, and consequently their Wasserstein distance, are close with high probability.

In the actual coherence experiment, however, the predictor is produced from the release $\alg(X_a)$ and is therefore dependent on the choice of $X_a$. Bounded max-information allows us to control this dependence: up to an appropriate loss in parameters, events that are unlikely in the hypothetical independent experiment remain unlikely when the predictor is produced from $\alg(X_a)$. \mb{This intuition about MI isn't obvious from what has been described about it so far. A reader who is not familiar with MI might be confused about not having been able to conclude it from themselves. Could either assert just assert that this is an interpretation of/consequence of the definition, or give the reader a bit more intuition to see why it's true. \sstext{added in a more precise informal def of max info earlier that should make it easier.}} This allows us to transfer the concentration argument above to the setting of coherence enforcement.

The argument is somewhat more involved when coherence is required within subgroups. The number of members of a subgroup assigned to each of $X_a$ and $X_b$ is itself random and hypergeometrically distributed. The empirical CDF within a subgroup therefore involves a ratio of dependent hypergeometric quantities, requiring additional concentration arguments; we give the full details in \Cref{sec:acheivingdc}.

Finally, we combine \Cref{introthm:max-info-implies-demographic-coherence} with the connection between differential privacy and max-information to obtain Theorem~\ref{introthm:pure-DP-implies-demographic-coherence}. The relevant DP-to-max-information result was previously established for i.i.d.\ samples \cite{DworkFHPRR15, RogersRST16}; we adapt the argument to the sampling-without-replacement setting induced by our random-partition experiment.

The parameter conversion in \Cref{thm:pure-dp-implies-coherence-enforcement} would likely require a value of $\eps$ that is too small for many practical applications. We expect this loss to be inherent, at least in part, to a black-box conversion from differential privacy to coherence enforcement. An intriguing open question is whether particular differentially private data-release algorithms can be analyzed directly to obtain substantially stronger coherence guarantees than those provided by this generic conversion.

\paragraph{Example of a non-DP algorithm that is coherence enforcing.}
Given the preceding result, it is natural to ask whether differential privacy and demographic coherence enforcement are equivalent. They are not. In \Cref{sec:nondp}, we analyze an algorithm that is not differentially private (indeed, it has $\eps=\infty$), but nevertheless satisfies strong demographic coherence guarantees. \mb{I think we should offer an explicit opinion that this algorithm still satisfies an intuitive notion of privacy. That way, the reader won't view its existence as evidence that demographic coherence is obviously too weak. \sstext{done below.}} This shows that demographic coherence enforcement is provably weaker than differential privacy.

We consider an algorithm that releases a single binary column of a dataset \emph{in the clear}, but first randomly shuffles its entries. Equivalently, the release reveals exactly how many $1$s appear in the column, but removes any association between these values and the other attributes of individual records.  Such a release can be viewed, for example, as publishing the exact outcome of a binary vote. Because no noise is added, the resulting statistic can be released with a level of accuracy that would be impossible under differential privacy. 

We note that, despite failing to satisfy differential privacy, this mechanism still provides an intuitive form of privacy protection: an analyst learns exactly how many records have the binary attribute, but not which individuals possess it or how it is associated with other attributes. We formalize this intuition by showing that the mechanism is demographically coherent.

\begin{theorem}[Informal Version of Theorem~\ref{thm:nondpcoherencefulllens}]\label{introthm:nondp}
Fix any $\beta \in (0,1]$, $\alpha \in (0,2]$, and sufficiently large $n \in \mathbb{N}$. Let $\univ$ be an arbitrary space, and consider an algorithm $\alg:(\univ\times\zo)^{n/2}\to\mathcal{Y}$ that outputs a uniformly random permutation of the binary column in its input. Then $\alg$ enforces $(\alpha,\beta)$-Wasserstein coherence provided that

$$
n \geq \frac{k\log(1/\alpha) \log(1/\beta)}{\alpha^2}
$$

for some constant $k$.
\end{theorem}

Let attribute $i$ be the one that is shuffled and released in the clear. The intuition is that, although the release reveals the exact number of records with bit $1$ in attribute $i$ of $X_a$, it reveals nothing about \emph{which} records with bit $1$ in attribute $i$ were selected into $X_a$. This residual uncertainty is enough to prevent a downstream learner from systematically producing different predictions on $X_a$ and $X_b$.

More concretely, let $N_a$ denote the number of records in $X_a$ whose released bit is $1$. The report determines $N_a$, so we condition on its value (and, more generally, on the entire shuffled report). The distribution of $X_a$ conditional on this information is identical to the dataset produced by the following sampling process: among all records in $X$ whose $i^{th}$ attribute is $1$, exactly $N_a$ are selected uniformly at random into $X_a$, with the remainder placed in $X_b$; an analogous statement holds for records whose $i^{th}$ attribute is $0$. Moreover, after conditioning on the report (and the randomness used by the learning algorithm), the predictor produced by any downstream learning algorithm is fixed. Thus, within each of the two bit classes (records with bit $1$ and $0$ in attribute $i)$, the predictor is evaluated on a uniformly random sample without replacement.

We can then analyze the Wasserstein distance through the cumulative distribution functions of the predictions, much as in the proof for differentially private algorithms. Fix a threshold $t$. Within the records whose $i^{th}$ attribute equals $1$, consider those for which the predictor outputs a value at most $t$. The number of these records that land in $X_a$, conditional on $N_a$, is hypergeometrically distributed. The same is true within the records whose binary attribute equals $0$. Hypergeometric concentration therefore implies that, within each bit class, the fraction of records falling below the threshold is similar in $X_a$ and $X_b$.

There is one additional source of discrepancy: the proportions of records with $0$'s and $1$'s in bit $i$ will themselves differ between $X_a$ and $X_b$. But $N_a$ is also generated by the original random partition and is therefore concentrated around its expectation. Combining these two forms of concentration shows that, for each threshold, the empirical CDFs of the predictor on $X_a$ and $X_b$ are close. Integrating this difference over thresholds yields a bound on the corresponding Wasserstein distance. We give the full argument in \Cref{sec:nondp}. Once again, the argument becomes more nuanced when we require coherence simultaneously across multiple categories.

Finally, this positive result is specific to releasing a single binary coordinate. Releasing multiple coordinates jointly can reveal correlations between attributes that are destroyed in the single-coordinate setting; indeed, our experiments in \Cref{sec:reallifevul} demonstrate that such correlations can be exploited to produce substantially less coherent predictions. \mb{Opportunity to say something about how they suggest a smooth tradeoff between the number of coordinates released and the level of coherence achievable?}

\section{Related Work}
\pjnote{need to mention something about pufferfish in this section too}

The study of privacy-preserving data release broadly falls into two categories: demonstration of potential harms via concrete attacks, and the development of formal methodologies that provide robust guarantees. These two approaches provide complementary insights. Formal approaches provide a concrete path to implementing privacy-protections, and the motivation for their use is derived from attacks. In particular differential privacy provably protects against membership inference (\eg \cite{Homer08, DworkSSUV15, ShokriSSS17,YeomGFJ18}), reconstruction (\eg \cite{DinurN03, CohenN20a, C:HaiNukYog22, carlini2021private}), and reidentification (\eg \cite{sweeney, NarayananS08}), as shown by Dwork et. al.~\cite{DworkSSU17}. In practice, however, there are fundamental challenges in using attacks to guide the many choices one must make when implementing privacy protections. These challenges arise from (1) identifying successful attacks, (2) identifying realistic attacks, and  (3) determining the privacy protections \emph{necessary} to prevent the attacks being considered.

\smallskip\noindent
\textbf{Evaluating the success of an attack.} In using attacks to motivate formal methodologies, one must start by demonstrating the extent of potential vulnerabilities. For example, membership inference is an attack that relates directly to the definition of differential privacy---however, the potential to infer membership in a dataset isn't a convincing vulnerability in the case of large data collection efforts like the US decennial Census. \gsk{Hm.  I'm wondering if we want to split this into two steps: (1) its not clear how to do measurement of sucssess, and (2) its not clear that any given choice of measurement tool is actually that convincing.  These two are intermingled here i think?} Therefore, differential privacy frequently derives its motivations from re-identification and reconstruction attacks. Still, the success of these attacks is difficult to evaluate.\footnote{For example, reconstruction of features like gender can be carried out simply from knowing population statistics rather than breaking anonymity \cite{ruggles2022role}.}
In recent work, Dick et al. implemented a reconstruction attack \cite{pnas_reconstruction} along with robust evaluations of its success. Their work has since been cited by the US Census Bureau’s chief scientist as evidence that “database reconstruction does compromise confidentiality”~\cite{KellerA24}. The key insight in their evaluation comparing the results of the reconstruction to a baseline in which reconstruction is conducted with complete access to the distribution underlying the data. 
While the intuition behind this work---that an attack is much more concerning if it reveals more than what could be learned from a detailed knowledge of the distributional properties---applies to many attack paradigms, 
the baseline considered in their work is specific to reconstruction attacks.  Reconstruction attacks are not always possible and conducting a reconstruction attack assumes malicious intent. 

Another place where the efficacy of specific attacks is measured via comparison to baselines is the literature on auditing differentially private algorithms (\eg \cite{JayaramanE19, JagielskiUO20, NasrH0BTJCT23, JagielskyNS23}). Here, attacks are carried out on existing systems, and the efficacy of the attack is used to measure the maximum level of ``effective privacy'' that the system confers.

\smallskip\noindent
\textbf{Identifying realistic attacks.} Research into conducting privacy attacks makes a variety of assumptions about the setting in which those attacks could be conducted, including  the goal of the attacker, the power of the attacker, the type of system attacked, etc\ldots  
These assumptions can radically change the extent to which an attack should be considered a realistic threat against real-world data releases; attacks that require unrealistic assumptions may not be concrete threats. 
The works of Rigaki \& Garcia~\cite{RigakiG24}, Salem et al.~\cite{SalemCEKPSTB23}, and Cummings et al.~\cite{CummingsHSS24} classify existing attack strategies by adversarial resources and goals in order to provide a structure for evaluating privacy risks.
In addition to this, Cohen~\cite{Cohen22} and Giomi et al.~\cite{giomi2022unified} take a different approach, appealing to the law to determine the goals of a realistic attacker.  Specifically, they contextualize the attacks they consider by tying them to existing privacy law. Still, individual attacks, even if successful and realistic, don't provide a clear path forward in terms of designing protections.

\smallskip\noindent
\textbf{Identifying necessary conditions.} Some prior work has started to identify \emph{necessary} conditions for achieving privacy.
Cohen \& Nissim \cite{CohenN20b} introduce a necessary condition, called ``predicate singling out,'' inspired by the GDPR notion of singling out.
Balle et. al.~\cite{BalleCH22} introduce an alternative necessary condition called ``reconstruction robustness,'' which is closely related to reconstruction attacks. Cummings et. al.~\cite{CummingsHSS24} build on the notion of reconstruction robustness, extending it to a weaker adversarial setting. 
Our framework extends this general approach but applies to a much broader class of attacks---namely, any attacks from which a confidence rated predictor could be distilled.

Recent work by Cohen et al.~\cite{CohenKMMNST24} also recognizes the need to bridge the gap between formal privacy guarantees and practical attacks. Building on definitions in prior work \cite{BalleCH22, CohenN20b, CummingsHSS24} they introduce ``narcissus resiliency,'' a framework for establishing precise conditions under which an algorithm prevents various classes of existing attacks, including reconstruction attacks, singling out attacks, and membership inference attacks.
Our definition defines invulnerability against a different type of privacy loss, providing complementary insights in the form of necessary conditions that can be considered alongside their definitions.  Specifically, 
it is important to consider demographic coherence alongside their notion of narcissus singling out; the latter captures an important property that the former does not. (An algorithm that chooses a small subset of the data to leak in the clear does not meet the definition of \emph{singling out security} even though it may be \emph{demographic coherence enforcing}\noakland{ if the subset is small enough}.)
Another key difference between our works is that the narcissus framework does not naturally lend itself to concrete experimental evaluation, whereas demographic coherence is intentionally designed with this use case in mind.

Another framework for analyzing the information leakage that is inherent in useful systems 
is that of quantitative information flow (QIF)~\cite{smith2009foundations}. At a very high level, works on QIF consider some similar problems to those considered in our work e.g., the need for formal guarantees that align with a variety of convincing inferences~\cite{smith2011quantifying,m2012measuring}, adversarial accuracy vs uncertainty~\cite{clarkson2009quantifying}, and 
connections to differential privacy~\cite{alvim2011relation,fernandes2024explaining}. However, this framework measures leakage by comparing an adversary's advantage before and after observing a system's outputs, therefore capturing the system's \textit{overall} epistemic leakage. In contrast, demographic coherence enforcement teases out leakage from a system's outputs that are based on particular details of its inputs rather than the demographic or population level patterns the details represent. 

\smallskip\noindent
Finally, most of the works discussed above measure the success of an attack via its accuracy (\ie is the information extracted about the data subject \emph{true}?).  We observe that harm is not necessarily predicated on accuracy, and we design demographic coherence to be intentionally independent of accuracy.  One impact of this choice is that demographic coherence is a more natural fit for settings in which ground truth is difficult or impossible to measure.

\section{A Walk Through Our Definition}\label{sec:walkthrough}
We now slowly build up our definition in order to clearly explain the choices embedded within it.

\smallskip\noindent
\textbf{Notation and Conventions.}
Assume that the data curator has collected a sample $X$ from the overall population of interest. We make no requirements on the relative sizes of $X$ and the population such that our framework can be used broadly---even when the goal is to sample the entire population, as in censuses. Our ultimate goal is to reason about a data curator $\alg$ who uses $X$ to generate a privacy-preserving release $R$.\footnote{In our formal experiment, we actually suppress the formal object of the report.  Specifically, we reason directly about the composition of some data processing algorithm $\alg$ with an arbitrary algorithm $\cL$, rather than making the report an explicit object that is then passed to $\cL$.}

\subsection{Predictive Harms}\label{sec:incoherent-predictions}

Within our framework, we characterize the adversary as a party interested in making predictions about individuals, \eg if people have some particular stigmatized feature or are going to buy a product if they are targeted with an advertisement. We formalize this conceptual approach by considering an arbitrary algorithm $\cL$ used by the adversary to design the predictor $h$.\footnote{A more complex interpretation could also cover an approach to prediction that takes into account a social decision-making process.  While this interpretation is beyond the technical scope of our work, it may be interesting to consider in future work.} We choose to measure privacy risk in terms of predictive harms for the following reasons:
\begin{enumerate}[--,leftmargin=*]
    \item \emph{Predictors are commonplace:}  The predictions made by machine learning models increasingly have direct impacts on people's daily lives. Diagnostic models are being tested as potential aides for medical experts~\cite{AhsanLS22}, and increasingly complex and opaque models are used to ``match'' job candidates with prospective employees~\cite{indeed,ideal,aihr} in order to increase the odds that an individual ends up with a lucrative job. Even complex infrastructures, like those used in digital advertising, can be seen as predictive models that are attempting to classify individuals into target audiences. The fact that predictive algorithms are increasingly commonplace---and the decisions they make concretely impact our daily lives---makes them a very \emph{believable} source of harm.
    
    \item \emph{Harmful predictors need not be maliciously produced:} By considering the impact of predictors, we free ourselves from needing to see the adversary as \emph{intentionally} trying to cause harm and instead can refocus on the (perhaps accidental) harms that a data release has the potential to cause. 

\end{enumerate}

A conceptual concern about considering predictive harms is whether we are explicitly ruling out particular, important types of adversarial behavior that are attempting to extract information (\eg reconstruction attacks). However, we note that by discussing predictors we are only limiting the input/output behavior of the adversary's product, and not how the predictor is produced. For example, our framework could capture an adversary that runs a known reconstruction algorithm (\eg \cite{DickDKLRVZ22}) and then makes predictions about individuals based on the produced table.  In this way, our approach highlights the ways in which existing approaches could be \emph{used} when applied in decision-making contexts. Looking ahead, in theoretically analyzing our framework we will \emph{universally quantify} over algorithms to preserve generality, which means that reconstruction-based approaches---or other known malicious approaches to data extraction---are naturally captured. 
Still, in choosing to concretize the type of our adversary, we do risk failing to consider a \textit{different} type of attacker with inconsistent goals. Definitions created with this philosophy can be extremely helpful at establishing \emph{necessary} conditions for ensuring privacy, but do not claim to be \textit{sufficient}. On the other hand, our approach helps highlight a specific way in which data is likely to be weaponized in the real world.

\smallskip\noindent
\textbf{Incoherent predictions.} 
Within this work we focus on capturing predictive harms that occur specifically by virtue of a data subject appearing in a dataset. To give a concrete example, consider the now classic case of Narayanan and Shmatikov's re-identification of Netflix users within an anonymized data release using public IMDB data \cite{SP:NarShm08}. In this setting, we might consider an adversary interested in learning a predictor which predicts queerness (\eg imagine the adversary is operating in a regime in which queerness is criminalized or highly stigmatized). Now, imagine two similar\footnote{The notion of similarity is obviously a loaded one, as the ways in which two individuals are similar or different depend on the types of predictions being made about them. We eventually handle this by quantifying over many notions of similarity. For the sake of this motivation, it is enough to assume that the similarity of these individuals is meaningful with respect to the characteristic being predicted about them.} individuals \examplenameone and \examplenametwo; each intentionally avoids being perceived as queer, and in particular does not provide ratings on movies with queer themes on their public IMDB profiles. Assume that based on a random sample, one of them (\eg \examplenameone) has their movie ratings released by Netflix and the other (\eg \examplenametwo) does not. A predictor that is likely to guess that \examplenameone is queer when they are present in the dataset but would not have guessed they are queer otherwise indicates that the predictor was able to extract some information about \examplenameone from the data release.  Given that \examplenameone and \examplenametwo are similar, this would also be true of a predictor that guessed that \examplenameone is queer while \examplenametwo is not. 

Importantly, this is true even if it's not clear exactly what form leakage takes or if the  prediction as to their queerness is inaccurate. We call predictors that act in this way ``demographically incoherent.'' There are two important (if unintuitive) subtleties that immediately emerge from this description of incoherent predictions:
\begin{enumerate}[(1),leftmargin=*]
    \item \emph{Harmful predictors need not be accurate:} 
    Within the example above, it is not important if \examplenameone is actually queer, it is enough that the predictor guesses that \examplenameone is queer because of their presence in the data.
    This is because we envision our predictor being used to make some real-world decision, \eg limiting the opportunities available to \examplenameone due to their perceived queerness. As such, the prediction's accuracy is a secondary concern.
    \item \emph{Measuring confidence is critical:}  When considering the ways in which data releases can be translated into real-world harms, it's important to recognize that enabling an adversary to make high confidence predictions about private attributes is a problem. 
    Importantly, this means that we should not require that the adversary can predict private attributes with 100\% certainty in order for it to be considered harmful. 
    Indeed, there is no particular cut-off threshold for certainty at which point it is natural to consider a harm occurring for all contexts.  In turning to predictors as our adversarial strategy, we naturally arrive in a context within which notions of confidence have been extensively explored. 
    Specifically, our approach considers confidence-rated predictors $h$, 
    allowing us to directly reckon with predictive uncertainty.
\end{enumerate}

We note that there are other pathological predictors which do not indicate privacy loss, and are therefore not considered demographically incoherent. For example, a predictor may make guesses that are entirely random or guess that everyone in the population has some feature, \ie make predictions that do not depend on the characteristics of individuals. The challenge then is to detect demographically incoherent predictors, whose behavior indicates privacy leakage, without depending on accuracy and without accidentally measuring variance in behavior that is not dataset dependent.

\subsection{An Experiment to Detect Demographically Incoherent Predictors} \label{sec:detecting-incoheret-algorithms}

With intuition about our class of ``bad'' predictors in hand, we now turn our attention to designing an experiment for detecting algorithms that produce them. In this discussion, we will defer on the concrete ways in which we will measure the demographic incoherence, and first focus on the experiment itself—that is, first we will decide what values we should measure, and then proceed to deciding how to do that measurement.
 
Because a symptom of incoherent predictors is differing performance on in-sample and out-of-sample individuals, it is clear that a \emph{comparison} is required.  However, it is not immediately obvious  what the ``right'' comparison should actually be.  In fact, some natural approaches fall short of our goals. As such, we walk through two seemingly natural, but flawed, experiments before discussing our final choice. Recall that the data curator has a dataset $X$ and will be releasing a privacy-preserving report $R$.

\begin{enumerate}[(1),leftmargin=*]

    \item \emph{Comparing before and after a data release:}  A very natural approach  would be to compare the performance of a predictor created \emph{before} a data release with one created \emph{after}, \ie comparing the performance (on individuals in $X$) of $h_0$ produced by an algorithm $\cL$ with access to the adversary's pre-existing, auxiliary information $\mathsf{Aux}$ to a predictor $h_1$ produced by $\cL$ with access to both $\mathsf{Aux}$ and the report $R$.\footnote{As this approach sketch is mainly to motivate our final approach below, we gloss over some formalities in this description.  For example, how do we know that $\cL$ does not act differently when provided one input ($\mathsf{Aux}$) and two inputs ($\mathsf{Aux}$ and $R$)?} 
    Such a comparison, intuitively, should isolate exactly the predictive changes associated with releasing $R$.

    Where this approach fails is that it does not recognize that there \emph{should} be a difference between the predictors $h_0$ and $h_1$ over the inputs in $X$.  After all, if there was no difference between $h_0$ and $h_1,$ there would be no value whatsoever in releasing $R$! As such, this comparison is necessarily conflating potential ``bad'' types of predictions that releasing $R$ enables with the ``good'' types of predictions that motivated the release of $R$ in the first place.

    \item \emph{Comparing to the base population:} The next most natural approach would be to compare the behavior of a single predictor $h$ on individuals in $X$ with that on individuals in the rest of the population. For example, by comparing its behavior on another similarly sampled dataset $Y$. This improves on our previous approach because we might expect that a ``good'' learning algorithm uses the dataset to learn about the population at large instead of revealing specifics about individuals.

    While this approach gets to the core of our interests, it has an important flaw. Technically speaking, we cannot assert that a real world sampling procedure has access to the base population distribution,
    \ie one cannot assume that two real world datasets are i.i.d samples from the same distribution. Also, we could conceivably be in a situation where the \emph{entire} population of interest is contained in $X$, leaving no one in $\bar{X}$ against which we could compare.  Therefore, keeping in mind the ergonomics of our definition in a concrete deployment scenario, 
    this approach also falls short of our goals.   
\end{enumerate}

\noindent
\emph{Our approach.} We build off the second approach above by taking control of the randomness used to separate the two comparison populations.  Specifically, we split the dataset $X$ into two uniformly selected halves, $X_a$ and $X_b$. We then use the data curation algorithm to generate the report $R$ using only $X_a$, holding $X_b$ in reserve as our ``test'' data set. We then test the behavior of a predictor $h$, designed based on the report $R$. Specifically, we compare the predictions of $h$ on individuals in $X_a$ and $X_b$. This approach ``fixes'' our second failed attempt by moving the assumptions about the randomness used in sampling $X$---something over which we have no control---into the randomness we use to split $X$ into $X_a$ and $X_b$---something over which we do have control. We say that a data release is \emph{demographically incoherent} with respect to $X$ if its predictions on members of  $X_a$ are noticeably different than the predictions it makes on members of $X_b$ (who necessarily have similar demographic distributions, given the uniform split.) 

\subsection{Measuring the Incoherence of a Predictor}

Finally, we discuss how to compare the behavior of $h$ on $X_a$ and $X_b$ without relying on accuracy.
Formally, we consider real-valued confidence-rated predictors $h: \univ \to [-1,1]$ which predict something about individuals. To capture the fact that these predictions are confidence rated, $h$ outputs values in $[0,1]$ when it predicts the attribute is likely true, and values in $[-1,0]$ when it predicts the attribute is likely false, with a higher absolute value representing higher confidence.  

For any such predictor, we will denote by $h(X_a)$ and $h(X_b)$ the uniform distributions over the predictions of $h$ on $X_a$ and $X_b$ respectively. Comparing these distributions allows us to reason about the general behavior of the predictor $h$ on $X_a$ and $X_b$ without considering accuracy of predictions on individuals.
In order to get a more granular understanding of the behavior of $h$ we further formalize the intuition of making  comparisons over ``similar'' individuals in $X_a$ and $X_b$ below.

\smallskip
\noindent
\textbf{Measuring a difference with respect to ``similar'' individuals.}     
In our motivating discussion of incoherent predictions, our representative individuals \examplenameone and \examplenametwo were assumed to be ``similar'' to one another. 
To formalise this intuition, 
we ask that a predictor is demographically coherent not only on the population as a whole, but also on recognizable subgroups from the population, \eg men, women, college freshmen, middle-school teachers etc\ldots\footnote{We borrow this conceptual approach from \cite{Hebert-JohnsonK18}.} For each of these subgroups of the population, the things that bind them together make them similar, in some particular sense. 
By operationalizing our earlier intuition in this way, we ask that the demographic coherence property holds not only over some particular notion of similarity, but rather over many notions of similarity at the same time. It also has the following technical and social benefits: (1) From a technical perspective, considering only the full population might hide incoherent decisions within sub-populations that effectively ``cancel out.'' That is, there might be a right-shift in one group that masks a left-shift in a different group, each shift effectively ``disappearing'' in the collective distribution over all individuals.  (2) From a social perspective, there may be particularly important groups within the population for whom we want to ensure coherent predictors for normative reasons. For example, if $X$ is a Census-like dataset, we may want to ensure that there are not sub-geographies on which incoherent predictors are allowed. Similarly, we may want to ensure that there aren't legally protected categories (\eg race, sex, religion, etc\ldots) on whom incoherent predictions are allowed. 

\smallskip\noindent\textbf{The \emph{lens} of a predictor.} Consider the adversary using the Netflix dataset to learn a predictor for queerness. We assume that at the time of making predictions, the adversary sees a public user profile (their IMDB ratings) which contains only some of the user information that was contained in the dataset. To formalize this intuition we introduce the \emph{lens} $\selection$ of the predictor, which indicates the attributes contained in the dataset which the predictor ``sees.'' We then compare the behavior of $h$ on members of $X_a$ and $X_b$ as seen through the lens $\rho$.\footnote{The predictor may also have side information about individuals not contained in the dataset, which can be formally included in the description of the adversarial algorithm $\cL$.}

\smallskip\noindent
\textbf{Choosing a metric.} In this work, we recommend instantiating the demographic coherence experiment with the distance metric of \emph{Wasserstein distance} (\Cref{def:wasserstein-distance}),
also known as earth-mover's distance, when measuring demographic coherence (or lack thereof).  Intuitively, this metric measures the minimum amount of work that it requires to deform one probability distribution into another.  If one visualizes a probability distribution as a mound of dirt, Wasserstein distance measures the effort required to move enough dirt to make one mound appear identical to the other (thus, earth-mover's distance). Unlike total variation distance, which disregards the spatial displacement needed to deform one distribution to another, Wasserstein distance is greater with a higher shift in confidence; recall that the importance of measuring confidence is one of the insights we highlight in \Cref{sec:incoherent-predictions}. Another advantage of the Wasserstein distance is that it has been widely studied and used in theoretical and empirical statistics, and so there is a rich mathematical toolkit that one can borrow from when reasoning about it. 

We recognize that there may be other measurement metrics that could be applied to the demographic coherence experiment that might highlight risk in different ways, and encourage this as important follow-up work.

\section{Formal Definition}\label{sec:formaldef}

This section presents the formal definitions corresponding to our framework. For a discussion about the various choices made here, see \cref{sec:walkthrough}.

\Cref{sec:notation} contains a glossary of the notation we use, \Cref{sec:def-incoherent-predictor} formally defines the notion of \emph{incoherence} that we measure, \cref{sec:def-demcoh} defines what it means for an algorithm to be \emph{demographically incoherent}, and \Cref{sec:def-coherence-enforcing} defines what it means for a data curation algorithm to be \emph{coherence enforcing}.

\smallskip
\noindent
\textbf{Notation.} We fix the following notation: \label{sec:notation}

\begin{itemize}[leftmargin=*,itemsep=0pt]
    \item We define a data universe $\univ = (\zo^* \cup \perp)^*$, where each feature of the data can also take the value $\perp$\yoakland{.} \noakland{(a.k.a. null).}
    
    \item We denote a dataset consisting of $n$ records from $\univ$ by dataset $X \in \univ^n$. \footnote{In a real-world scenario, $X$ might be sampled from some underlying population, but our definition does not require this and makes no assumptions about how it might be done.}

    \item We denote a collection $\cC$ of sub-populations $C \subseteq \univ$.

    \item For any dataset $X$ and sub-population $C \in \cC$, we define $X|_C \defeq X \cap C$ to be the restriction of $X$ to the sub-population $C$, \ie the members of the dataset $X$ that belong to sub-population $C$. 

    \item We define a lens $\selection$ as a set of features from $\univ$.
    
    \item For a lens $\selection$, we define $\restrict{X}{\selection}$ to be the data in $X$ restricted to the features in the lens. That is, for every feature represented by $\selection$, the sets $\restrict{X}{\selection}$ and $X$ are exactly the same, and for features not represented by $\selection$, the entries in $\restrict{X}{\selection}$ always have the value $\perp$. 
    
    \item For any fixed predictor $h: \univ \to [-1,1]$, define the distribution $h(X)$ as the uniform distribution over the predictions of $h$ on $X$. That is, the distribution $h(X)$ has cumulative distribution function
    \noakland{$$cdf_{h(X)}(p) = \underset{x \setrandomly X}{\Pr}[h(x) \leq p].$$}
    \yoakland{$$cdf_{h(X)}(p) = {\Pr}_{x \setrandomly X}[h(x) \leq p].$$}
\end{itemize}

\subsection{Measuring Demographic Incoherence of Predictors}\label{sec:def-incoherent-predictor}

In this section we define the notion of \emph{incoherence} that we measure over predictors. This measurement is used informally in the demographic coherence experiment $\coherenceExp$ (\cref{fig:coherenceExp}).

\begin{definition}[Demographically Incoherent Predictor]\label{def:incoherent-predictor}
    Let $\univ$ be a data universe, let $X_a, X_b \in \univ^{n/2}$ be datasets consisting of $n/2$ data entries each, let $\cC$ be a collection of sub-populations $C \subseteq \univ$, let the lens $\selection$ be a set of features from $\univ$, and let $h: \univ \to [-1,1]$ be a confidence rated predictor. Let $d(\cdot,\cdot)$ represent some metric distance between probability distributions.
    
     \smallskip\noindent
     We say that the $h$ has \emph{$\alpha$-incoherent predictions} with respect to $X_a$, $X_b$, $\selection$, and $\cC$, if for some $C \in \cC$, 
     \[ d\Big(h(\restrict{X_a|_C}{\selection}),\, h(\restrict{X_b|_C}{\selection})\Big) > \alpha \] 
    
    \smallskip\noindent
    (\noakland{In contrast,} \yoakland{We say that} $h$ has \emph{$\alpha$-coherent predictions} with respect to $X_a$, $X_b$, $\selection$, and $\cC$, if \yoakland{the converse holds}
    \noakland{for all $C \in \cC$, it satisfies $d\Big(h(\restrict{X_a|_C}{\selection}),\, h(\restrict{X_b|_C}{\selection})\Big) \leq \alpha$.})
\end{definition}

\smallskip\noindent
In other words, the predictions of $h$ are incoherent if there is some sub-population $C$ that witnesses a big distance between its predictions on two sets $X_a$, and $X_b$. Note that this definition distills a notion of ``incoherence'' only once we fix some assumptions on $X_a,X_b$, perhaps that they are drawn from similar distributions. 

\subsection{Demographic Coherence}\label{sec:def-demcoh}

Given a data universe $\univ$, an algorithm ${\cL:\univ^{n/2}}\to{(\univ\to[-1,1])}$ which uses a dataset of size $n/2$ to produce a confidence-rated predictor $h: \univ \to [-1,1]$, a collection $\cC$ of sub-populations $C \subseteq \univ$, a lens $\selection$ (a set of features from $\univ$),  and a distance metric $\dist(\cdot,\cdot)$ between probability distributions, we formally define when the algorithm $\cL$ is \emph{demographically coherent}.

We start by defining the demographic coherence experiment $\coherenceExp$ which checks the demographic coherence of $\cL$ with respect to a specific dataset $X$, collection $\cC$, lens $\rho$, and distance metric $\dist(\cdot,\cdot)$. This experiment works similarly to the description in \cref{sec:detecting-incoheret-algorithms}, except we are testing the coherence of an algorithm that gets the dataset in the clear. In \cref{sec:def-coherence-enforcing}, we will use the notion of demographic coherence to define when a data curator is \emph{coherence enforcing}. 

In the $\coherenceExp$ experiment, the input dataset $X$ is split into sets $X_a, X_b$ where $X_b$ is held in reserve as a ``test'' set. Then the algorithm $\cL$, with input $X_a$, is used to produce a predictor $h$. Finally, the predictor $h$ is checked for \emph{demographic incoherence} (Def~\ref{def:incoherent-predictor}) with respect to $X_a,X_b,\rho$, and $\cC$.

\ifnum\oakland=0
\mybox{$\coherenceExp_{\cL,X,\cC,\selection,\dist(\cdot,\cdot)}(\alpha)$}{Demographic Coherence Experiment}{fig:coherenceExp}{

\newcommand{\expindent}{\hspace{1em}}

\textbf{Input:}\\
\ifnum\usenix=0
\vspace{-2em}
\begin{itemize}[leftmargin=1em,itemsep=0pt,label={}]
\item 
\fi
An algorithm $\cL: \univ^{n/2}\to (\univ \to [-1,1])$, A dataset $X\in\univ^n$, a collection $\cC$ of subgroups $C\subseteq \univ$, a lens $\selection$, and a distance metric $\dist(\cdot,\cdot)$.
\ifnum\usenix=0
\end{itemize}
\fi

\medskip
\textbf{Split Data:}\\

\ifnum\usenix=0
\vspace{-2em}
\begin{itemize}[leftmargin=1em,itemsep=0pt,label={}]
\item 
\fi
$I \setrandomly \{S \subseteq [n] : |S| =  n/2 \}$
\item $X_a = (x_i)_{i \in I}$
\item $X_b = (x_i)_{i \in [n]\setminus I}$
\ifnum\usenix=0
\end{itemize}
\fi

\textbf{Compute Predictor($X_a$):}\\
\ifnum\usenix=0
\vspace{-2em}
\begin{itemize}[leftmargin=1em,itemsep=0pt,label={}]
\item 
\fi
$h\gets \cL(X_a)$
\ifnum\usenix=0
\end{itemize}
\fi

\textbf{Incoherence Condition:}\\ 

\ifnum\usenix=0
\vspace{-2em}
\begin{itemize}[leftmargin=1em,itemsep=0pt,label={}]
\item 
\fi
\pj{Set $b=0$ if there exists $C \in \cC$ such that:}
\ifnum\usenix=0
\end{itemize}
\vspace{-1.4em}
\begin{itemize}[leftmargin=2.2em,label={  },itemsep=0pt]
\item 
\fi
$\dist\Big(h(\restrict{X_a|_C}{\selection}),\, h(\restrict{X_b|_C}{\selection})\Big) > \alpha$
\ifnum\usenix=0
\end{itemize}
\fi
\ifnum\usenix=0
\vspace{-2em}
\begin{itemize}[leftmargin=1em,itemsep=0pt,label={}]
\item 
\fi
\pj{Else set $b=1$}
\ifnum\usenix=0
\end{itemize}
\fi
}


\else


\mybox{$\coherenceExp_{\cL,X,\cC,\selection,\dist(\cdot,\cdot)}(\alpha)$}{Demographic Coherence Experiment 
}{fig:coherenceExp}{

\newcommand{\expindent}{\hspace{1em}}

\textbf{Input:}
\begin{itemize}[leftmargin=1em,itemsep=0pt,label={}]
\item 
An algorithm $\cL: \univ^{n/2}\to (\univ \to [-1,1])$, A dataset $X\in\univ^n$, a collection $\cC$ of subgroups $C\subseteq \univ$, a lens $\selection$, a distance metric $\dist(\cdot,\cdot)$, and a value $\alpha \in (0,2]$.
\end{itemize}

\medskip
\textbf{Split Data:}
\begin{itemize}[leftmargin=1em,itemsep=0pt,label={}]
\item 

$I \setrandomly \{S \subseteq [n] : |S| =  n/2 \}$
\item $X_a = (x_i)_{i \in I}$
\item $X_b = (x_i)_{i \in [n]\setminus I}$

\end{itemize}

\textbf{Compute Predictor($X_a$):}
\begin{itemize}[leftmargin=1em,itemsep=0pt,label={}]
\item 
$h\gets \cL(X_a)$

\end{itemize}

\textbf{Incoherence Condition:}
\begin{itemize}[leftmargin=1em,itemsep=0pt,label={}]
\item Return $\mathbf{b=0}$ if there exists $C \in \cC$ such that: \\
$\dist\Big(h(\restrict{X_a|_C}{\selection}),\, h(\restrict{X_b|_C}{\selection})\Big) > \alpha$
\item Else return $\mathbf{b=1}$.
\end{itemize}

}

\fi

\smallskip\noindent

\vspace{0.8em}
A natural formalization of the intuition we developed in \cref{sec:walkthrough} would say that an algorithm $\cL$ produces $(\alpha,\beta)$-demographically coherent predictions with respect to collections $\cC$ and lens $\selection$ if the following holds for all datasets $X$:
 \[\Pr[\coherenceExp_{\cL,X,\cC,\selection}(\alpha) = 0] \leq \beta.\]
However, this definition cannot be realized with respect to all datasets $X$ and arbitrary categories $C$ because the sampling experiment itself introduces some incoherence. For example, consider the category of men over 60, and a dataset that contains only two such men. There exists a predictor that predicts $-1$ on one of them and $1$ on the other. With probability $1/2$, the men end up in different halves $X_a$, $X_b$, in which case the Wasserstein distance between the predictors' distributions on $X_a|_C$ \& $X_b|_C$ is $2$. Hence, for our definition to be meaningfully achievable, we require the datasets to be sufficiently representative so that the incoherence due to sampling does not dominate. Thus, we define demographic coherence with respect to a size constraint $\gamma$. 

To remove the need for the parameter $\gamma$, one could redefine $X_a|_C$ by ``zeroing out'' members of $X_a$ that are not in $C$ instead of taking the intersection $X_a\cap C$. This would effectively result in asking the predictor $h$ in \cref{def:incoherent-predictor}, and \cref{fig:coherenceExp} to make ``dummy'' predictions, with no information, for every member of $X_a$ and $X_b$ not belonging to $C$. While such a definition may be more mathematically elegant, we believe that the explicit failure point represented by $\gamma$ in our definition is important for interpretability and ease of use---especially by non-experts. For the same reason, it is important to have an explicit collection $\cC$, and lens $\rho$, even though the eventual theorems may be quite general, holding for large ranges of $\gamma$, $\cC$, and $\rho$.

\begin{restatable}{definition}{DefCoherence}\label{def:coherence}
    {\normalfont(Demographic Coherence)\textbf{.}}
    Consider a data universe $\univ$ and an algorithm $\cL: \univ^{n/2}\to (\univ \to [-1,1])$ which uses a dataset of size $n/2$ to produce a fixed confidence-rated predictor $h: \univ \to [-1,1]$. Let $\cC$ be a collection of sub-populations $C \subseteq \univ$ and let lens $\selection$ be a set of features from $\univ$. Let $\dist(\cdot,\cdot)$ represent some metric distance between probability distributions. 
    We say that \emph{$\cL$ produces $(\alpha,\beta)$-demographically coherent predictions} with respect to collection $\cC$, size-constraint $\gamma$, and lens $\selection$ if the following holds:

    \smallskip\noindent 
    For all $X \in \cX^n$, $\cC^*= \{C\in\cC \mid\, |C\cap X| \geq \gamma\}$
    \[\Pr[\coherenceExp_{\cL,X,\cC^*,\selection}(\alpha) = 0] \leq \beta.\]
\end{restatable}

\subsection{Coherence Enforcing Algorithms}\label{sec:def-coherence-enforcing}

We finally define \emph{coherence-enforcing} algorithms by reference to the definition of \emph{demographic coherence} (\cref{def:coherence}). Specifically, a data curator $\alg$ is \emph{coherence enforcing} if, any algorithm $\cL$ can be rendered \emph{\coherent} simply by filtering its inputs through the data curator $\alg$, without any changes to $\cL$.

\begin{restatable}{definition}{DefCoherenceEnforcing}\label{def:coherenceEnforcing}
    {\normalfont(Coherence Enforcing Algorithms)\textbf{.}}
    Consider data universes $\univ, \mathcal{Y}$, a collection $\cC$ of sub-populations $C \subseteq \univ$, a lens $\selection$, and an algorithm $\alg:\univ^{n/2} \to \mathcal{Y}$. Let $\dist(\cdot,\cdot)$ represent some metric distance between probability distributions. 
    \noakland{
    
    \smallskip\noindent}
    We say $\alg$ enforces $(\alpha,\beta)$-demographic coherence with respect to collection $\cC$, size-constraint $\gamma$, and lens $\selection$ if:

    \smallskip\noindent
    For any algorithm $\cL:\mathcal{Y}\to(\cX\to[-1,1])$, the combined algorithm $\cL\circ\alg:\univ^{n/2} \to (\univ \to [-1,1])$ satisfies $(\alpha,\beta)$-demographic coherence with respect to  the collection $\cC$, size-constraint $\gamma$,  and lens $\selection$.
\end{restatable}

\subsection{Instantiating Demographic Coherence with a Metric}

We will instantiate these definitions in the following sections using the Wasserstein-1 metric defined below:

\begin{definition}\label{def:wasserstein-distance}
    Let $P,Q$ represent distributions over a discrete subset $S \subseteq \mathbb{R}$. Then, the $1$-Wasserstein distance between $P,Q$ is defined as    
    $$\wass(P,Q) = \inf_\pi \sum_{i \in S} \sum_{j \in S} \|x_i - x_j \|_1 \pi(x_i,x_j),$$
   where the infimum is over all joint distributions $\pi$ on the product space $S \times S$ with marginals $P$ and $Q$ respectively. 
\end{definition}

If an algorithm is $(\alpha, \beta)$-demographically-coherent as per \Cref{def:coherence} with this Wasserstein metric instantiation, we say that it is \emph{$(\alpha, \beta)$-Wasserstein-demographically-coherent} (or $(\alpha,\beta)$-Wasserstein-coherent for short.) Similarly, if an algorithm enforces $(\alpha, \beta)$-demographic-coherence as per \Cref{def:coherenceEnforcing} with this Wasserstein metric, then we say it \emph{enforces $(\alpha, \beta)$-Wasserstein-demographic-coherence} (or it enforces $(\alpha, \beta)$-Wasserstein-coherence for short.)

\section{Capturing Real-Life Vulnerabilities from Releasing Data in the Clear}\label{sec:reallifevul}


In this section, we construct a simple and realistic experiment to demonstrate how our framework can be used to capture real-life privacy leaks. 

\smallskip
\noindent
\textbf{Setup:} A governmental entity is interested in facilitating academic study on disability with the goal of providing better support, and so releases a dataset 
excluding peoples' names, but including their age, sex, racial features, type of work, and cognitive disability status. The governmental entity is debating whether to release this data in the clear reckoning with the following considerations: 
\begin{itemize}
    \item The privacy harm of releasing such data in the clear is not immediately apparent---disability status is a sensitive feature, but the remaining attributes  may not uniquely identify most individuals and hence may not provide a tool to systematically deduce the disability status of individuals. 
    \item Releasing such data in the clear would maximize the utility associated with the release.
    \end{itemize}

The adversary in our setup is a corporation that is screening job applicants, and due to ableism, wants to avoid hiring employees with cognitive disabilities. However, regulations prevent companies from requiring job applicants to provide this information. Hence, the adversary decides to use the data release to train a classifier predicting cognitive disability status, with the intent of using it to predict disability status of individuals who apply for jobs. We note that machine learning methods are already used to screen resumes\gsk{(add citation here)}, and so this flow captures a realistic decision-making process.

Our experiment aims to demonstrate the utility of our framework in this scenario by exploring whether we can use our framework with the above set up to expose the privacy harms associated with releasing such data.

\smallskip
\noindent
\textbf{Experimental details:} Next, we describe the details of the experiment we run using the above set up.

\emph{Dataset:} We simulate the governmental entity in the setup using the NIST American Community Survey (ACS) Data Excerpts benchmark \cite{TaskBDH2025}, which consists of curated subsets of publicly released datasets containing info about real households in the US. This benchmark is intended to investigate the performance of data synthesizers and other privacy enhancing technologies. We specifically use the national version of the dataset from 2019, which consists of 27254 records with 24 features each drawn from 20 different geographical areas in the United States. We use a subset of 9 features for each record---age, population density of the geographical region, sex, race, Hispanic origin, highest educational attainment, general and specific categories of work, and cognitive disability status (a binary attribute). We note that 1521 people in the dataset (5.5\%) have a cognitive disability, which is a small but significant fraction.

\emph{Algorithm used by the adversary:} Our adversary uses a supervised learning approach: training a random forest classifier with 50 sub-classifiers as the algorithm to generate a predictor for cognitive disabilities. We use the  scikit-learn \cite{scikit} implementation of this algorithm.

\emph{Experimenting with our framework:} Our framework involves randomly splitting the dataset into two halves, training the random forest classifier on one half of the dataset and using the other half as a held-out set for assessing the privacy loss. Our framework requires specifying vulnerable groups; we consider 2 overlapping groups: the entire dataset, and the set of individuals who have a cognitive disability (and hence are at particular risk of not being hired if their disability status is disclosed or inferred). After training the classifier, we compute the prediction distributions on the two halves (restricted to the subgroups of interest), and then measure the Wasserstein distance between these prediction distributions. We repeat this experiment a number of times and consider the average Wasserstein distance. Note once again that cognitive disability is a binary attribute in the dataset; we use $-1$ to represent the presence of a cognitive disability and $1$ to represent the absence of a cognitive disability.

\smallskip
\noindent
\textbf{Results:} Figure~\ref{fig:wass_plt} shows the Wasserstein distance between the training and test set restricted to the cognitive disability subgroup. The corresponding average Wasserstein distance is $1.561$. We note that for the same set of iterations, the average Wasserstein distance when considering the whole training and test sets was a significantly smaller value of $0.0773$.

\begin{figure}[ht]
    \centering \includegraphics[width=0.8\linewidth]{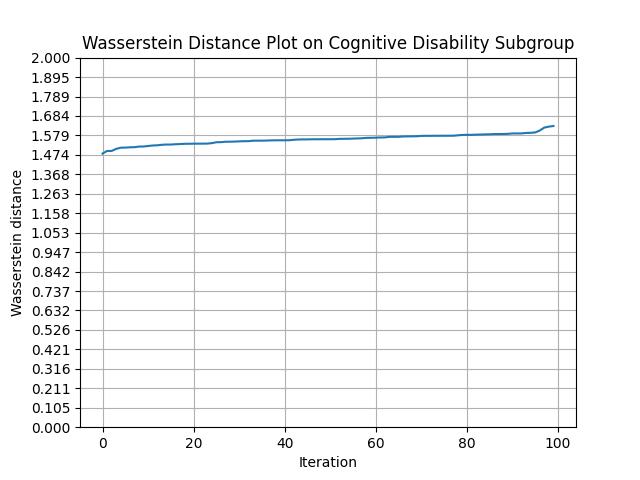}
    \caption{This plot shows the Wasserstein distances in different iterations between the prediction distributions on the training and test set restricted to the cognitive disability subgroup. (Note that the distances have been sorted).}
    \label{fig:wass_plt}
\end{figure}

To interpret the source of the Wasserstein distance difference and explain the associated intuition, we also plot the entire prediction distributions over the training and test sets restricted to the Cognitive Disability subgroup for a single run of the algorithm. Note that the ground truth values in both cases are all $-1$s.

\begin{figure}[ht]
    \centering \includegraphics[width=0.8\linewidth]{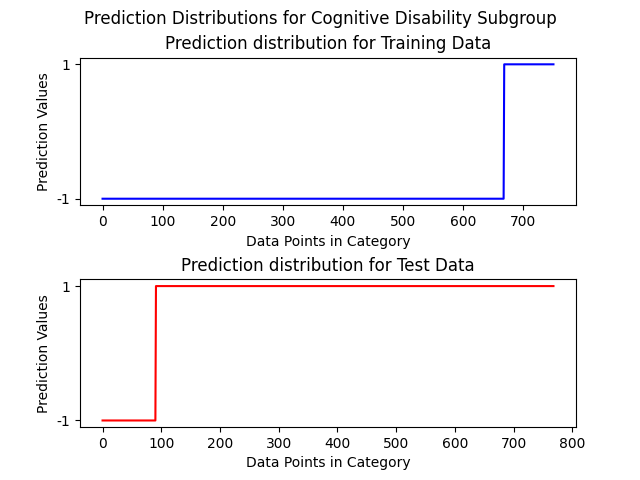}
    \caption{These plots show the predictions (sorted) on the training and test set restricted to the subgroup of people with a cognitive disability for a single iteration of the experiment.}
    \label{fig:singlerun}
\end{figure}

\textbf{Discussion:} Next, we discuss and interpret the results of the experiment.

\textbf{Our framework captures privacy harms of releasing data in the clear in this setup:} From the results of the experiment, we see that for individuals with a cognitive disability, there is a very large Wasserstein distance between the prediction distributions on the training set and the test set ($1.561$, for a value ranging between $0$ and $2$). As seen in Fig~\ref{fig:singlerun}, this is because $+1$ is predicted on a much more significant fraction of this subgroup in the test set than the training set.  In particular, in our setup this would mean that people with cognitive disabilities whose data was disclosed in the release have a much higher chance of being predicted as having a cognitive disability by an adversary than people with cognitive disabilities whose data was not disclosed in the release. As a result, people with cognitive disabilities whose data is in the release face a much higher risk of discrimination (e.g., in hiring). The large Wasserstein distance can be used to flag this privacy harm, supporting its use as a metric in our framework. 

Our explanation of this is as follows: since a small percentage (5.5$\%$) of the entire dataset has a cognitive disability, and there may not be enough information to systematically identify this subset.  As a result. the random forest classifier simply memorizes and regurgitates many members of the training dataset that have the cognitive disability in order to minimize training loss. For most other members, the classifer simply resorts to predicting $+1$, including for most members of the test set restricted to the cognitive disability 
subgroup.

\textbf{Focusing on subgroups is important:} Note that the average Wasserstein distance between the prediction distributions on the entire training and test set is relatively small---it is roughly $0.073$. On the other hand, when we zoom into the subgroup containing individuals with a cognitive disability, the Wasserstein distance blows up to $1.561$, exposing a more significant privacy harm. This is why we consider subgroups in our framework, as aggregating over the entire dataset can result in privacy harms on subgroups being diminished to a point where they are not noticeable. 

\textbf{Privacy harms can accrue even without an explicit attack:} A common critique of the literature on privacy attacks is the question of who has the incentives to launch particular classes of attacks. Our results show that privacy harms can accrue even without such attacks being explicitly carried out. In our experiment, the adversary does not seek to infer membership in the training dataset or to reconstruct the data of individuals, but is attempting to learn a good classifier of a sensitive attribute. It turns out that the classifer memorizes information in the training dataset that can be leveraged to cause discrimination that affects members of the training set disproportionately, thereby precipitating a privacy harm. Hence, our framework systematically captures a set of privacy violations not directly addressed by the literature on privacy attacks.\ssnote{ Suggest our framework as an additional test in their benchmark/comment on how it relates to the benchmark. }


\section{Differentially Private Algorithms Enforce Wasserstein Demographic Coherence} \label{sec:acheivingdc}

In this section, we show that differentially private algorithms enforce Wasserstein demographic coherence. In fact, we prove the stronger statement \Cref{thm:max-info-implies-demographic-coherence} (in \Cref{sec:max-info-implies-demographic-coherence}) which says that algorithms with \emph{bounded max-information} are coherence enforcing (for all possible lenses $\rho$). Since \noakland{differentially private}\yoakland{DP} algorithms have bounded max-information \noakland{(Theorems~\ref{thm:pure-dp-implies-maxinfo} and~\ref{thm:approx-dp-implies-max-info})} \yoakland{(see full version)}, this also proves that all \noakland{differentially private}\yoakland{DP} algorithms are coherence enforcing regardless of the choice of lens $\rho$ (Theorems~\ref{thm:pure-dp-implies-coherence-enforcement}~and~\ref{thm:approx-dp-implies-coherence-enforcement}). 



Let $\cA : \univ^{n/2} \to \cY$ denote a data curator operating on datasets of size $n/2$. For parameters $\beta, \zeta > 0$, algorithm $\cA$ has $\beta$-approximate max-information at most $\gamma$ if for all sets of outcomes $T \subseteq (\univ \times \cY)$,
\[\Pr[(S, \cA(S)) \in T] \le e^{\gamma} \Pr[(S, \cA(S')) \in T] + \beta,\]
where $S, S'$ denote independent samples each of size $n/2$ from $\univ$. (See Definition~\ref{def:max-info-alg} for the full definition.) Intuitively, this condition guarantees that output of the algorithm $\cA$ does not reveal much more information about the input dataset $S$ than it would about an independently sampled dataset $S'$ drawn from the same distribution. The main technical result in this section shows that every algorithm with bounded max-information enforces demographic coherence with respect to the collection $\cC$ of all sufficiently large groups.

Recalling the experiment $\coherenceExp_{\cL\circ\alg,X,\cC,\selection}(\alpha)$ described in \Cref{fig:coherenceExp}, our goal is to show that for all algorithms $\cL:\cY\to\{\univ \to [-1,1]\}$, all lenses $\rho$, all datasets $X\in\cX^n$, and all sufficiently large $C \subseteq \univ$, that with high probability over a choice of split $X_a, X_b \setrandomly X$, and predictor $h\leftarrow \cL\circ\alg(X_a)$, the following holds:
$$\wdist(h(\restrict{X_a\lvert_C}{\rho}), h(\restrict{X_b\lvert_C}{\rho})) < \alpha.$$
Here, $\alpha$ is a parameter that we relate to $\gamma, \beta$ and the minimum subpopulation size in $\cC$. A conceptual challenge in analyzing this Wasserstein distance is that the predictor $h$ depends on the split $X_a, X_b$. If the draws being fed into $h$ could instead be modeled as independent samples, we could hope to use concentration inequalities to control this distance. Fortunately, this is where we can leverage bounded max-information of $\cA$. In \ifnum\oakland=0 ~\Cref{lem:max-info-decouples-Xa-from-ha} \else  the proof that bounded max information implies Wasserstein coherence enforcement, \fi we show that it is sufficient to reason about the case where the predictor $h$ is independent of the samples with which the prediction distributions are defined, that is,
\ifnum\oakland=0
\begin{align*}
	\wdist(h(\restrict{X_a\lvert_C}{\rho}), h(\restrict{X_b\lvert_C}{\rho}))
	\approx \wdist(h(\restrict{S\lvert_C}{\rho}), h(\restrict{\overline{S}\lvert_C}{\rho})),
\end{align*}
\else
\begin{multline*}
	\wdist(h(\restrict{X_a\lvert_C}{\rho}), h(\restrict{X_b\lvert_C}{\rho})) \\
	\approx \wdist(h(\restrict{S\lvert_C}{\rho}), h(\restrict{\overline{S}\lvert_C}{\rho})),
\end{multline*}
\fi
where $S$ denotes a fresh (independent of $X_a, X_b$) random subset of $X$.

Next, we show that, with high probability $\wdist(h(\restrict{S\lvert_C}{\rho}), h(\restrict{\overline{S}\lvert_C}{\rho})) { < \alpha}$ . For this, we use an equivalent formulation of the Wasserstein distance between \noakland{probability} distributions $\mu, \nu$ (over $[-1,1]$):
\[\wdist(\mu, \nu) = \int_{-1}^1  |\cdf_{\mu}(\ell) - \cdf_{\nu}(\ell)| \ d\ell,\] 
where $\cdf_\mu, \cdf_\nu$ denote the cumulative distribution functions of $\mu$ and $\nu$. The upshot of this formulation is that for each fixed value of the parameter $\ell$, the value of $\cdf_{h(\restrict{S\lvert_C}{\rho})}(\ell)$ concentrates tightly around its expectation, and is hence close to that of $\cdf_{h(\restrict{\overline{S}\lvert_C}{\rho})}(\ell)$ with high probability. We formalize this argument in \ifnum\oakland=0 \Cref{lem:unrelated-predictor}, \else the full version of the paper, \fi where we 
model these quantities as
hypergeometric random variables. We then argue that pointwise closeness of these CDFs implies closeness of the integral of their difference, and hence of the desired Wasserstein distance.

\ifnum\oakland=0
\ifnum\oakland=0
\subsection{Bounded Max-Information Algorithms are Coherence Enforcing}\label{sec:max-info-implies-demographic-coherence}
\else
\subsection{Bounded Max-Information Algorithms are Coherence Enforcing}\label{sec:max-info-implies-demographic-coherence}
\fi

\begin{definition}[Max-information of random variables \cite{DworkFHPRR15}] Let $X$ and $Y$ be jointly distributed random variables over the domain $(\cX,\cY)$. The \textit{$\beta$-approximate max-information} between $X$ and $Y$, denoted by $I_\infty^\beta(X;Y)$ is
	$$I_\infty^\beta(X;Y) = \ln \left(\underset{\substack{T\subseteq(\cX\times\cY)\\\Pr[(X,Y)\in T]>\beta}}{\sup}\frac{\Pr[(X,Y)\in T]-\beta}{\Pr[X\otimes Y \in T]}\right)$$
\end{definition}

\begin{definition}[Max-information of algorithms \cite{DworkFHPRR15}] \label{def:max-info-alg}
	\footnote{This definition, for sampling without replacement, is slightly different than the original one.}  Fix $n \in \mathbb{N}$, $\beta > 0$. Let $\univ$ be a finite data universe of size $m$. Let $S$ be a sample of size $n$ chosen without replacement from $\univ$. Let $\alg: \univ^{n} \to \cY$ be an algorithm.
	
	Then we define the max-information of the algorithm as follows:
	\begin{align}
		I^{\beta}_{\infty}(\alg, n) = I^{\beta}_{\infty}(S, \alg(S))
	\end{align}
\end{definition}

The main technical result we prove is the following.
 
\begin{theorem}\label{thm:max-info-implies-demographic-coherence}
    Let $n\in\N$,$\zeta>0$, $\beta \in (0,2)$, $\alpha \in (0,2]$. Let $\wdist(\cdot,\cdot)$ denote the Wasserstein-1 distance metric. Consider a collection $\cC$ of sub-populations $C \subseteq \univ$, a lens $\selection$, and an algorithm ${\alg:\univ^{n/2}\to \cY}$ with bounded max-information 
    $$I^{\beta/2|\cC|}_{\infty}(\alg,n/2) < \zeta.$$ 
    Then $\alg$ enforces $(\alpha,\beta)$-Wasserstein-coherence with respect to collection $\cC$, lens $\rho$, and size constraint $\gamma$, where
     \begin{align}
        \gamma = \max\Big\{ & \frac{33.2 \cdot (\zeta + \ln(16|\cC|/\beta))}{\alpha^2}, \frac{144\ln(6/\alpha)}{\alpha^2}, \nonumber \\
        &16.6 \cdot (\zeta + \ln(16|\cC|/\beta)), \frac{10.6}{\alpha}, 80\Big\}.
    \end{align}
    
\end{theorem}

\noindent
The intuition behind the result in \Cref{thm:max-info-implies-demographic-coherence} is that the output of an algorithm with bounded max-information does not contain too much specific information about the input dataset. This intuition is leveraged to prove \Cref{lemma:max-info-implies-demographic-coherence-int}, which is the main lemma underlying this result. \Cref{thm:max-info-implies-demographic-coherence} follows from this lemma by an appropriate setting of parameters. \ifnum\oakland=1 The proofs of both these results can be found in the full version of the paper. \fi

\begin{lemma}\label{lemma:max-info-implies-demographic-coherence-int}
    Let $\eta,\zeta>0$, $\alpha \in (0,1)$. Let $\wdist(\cdot,\cdot)$ represent the Wasserstein-1 distance metric. Consider a sub-population $C \subseteq \univ$, a lens $\selection$, and an algorithm ${\alg:\univ^{n/2}\to \cY}$ with bounded max-information 
    $$I^\eta_{\infty}(\alg,n/2) < \zeta.$$ 

    \noindent
    For all algorithms $\cL: \cY \to \{ \univ \to [-1,1] \}$, datasets $X\in\cX^n$, $\mu > 0$, as long as 
    \[|X\cap C|\ge \max\left\{\frac{4.15\cdot\ln(4/\mu)}{\alpha^2}, \frac{16/3}{\alpha}, 8.3\cdot\ln(4/\mu), 40\right\},\]   
    we have
\ifnum\oakland=1
   \begin{align*}
    \underset{X_a \leftarrow X, h \leftarrow \cL \circ \alg(X_a)}{\Pr}&\Big[\wdist(h(\restrict{X_a\lvert_C}{\rho})\\
    & , h(\restrict{X_b\lvert_C}{\rho})) > 2\alpha\Big] \\
    & \,\leq\,\, 2\mu(|X\cap C|+1)\cdot e^\zeta + \eta.
    \end{align*} 
\else
 \begin{align*}
    \underset{X_a \leftarrow X, h \leftarrow \cL \circ \alg(X_a)}{\Pr}&\left[\wdist(h(\restrict{X_a\lvert_C}{\rho}), h(\restrict{X_b\lvert_C}{\rho})) > 2\alpha\right]
    & \,\leq\,\, 2\mu(|X\cap C|+1)\cdot e^\zeta + \eta.
    \end{align*} 
\fi
    Here $X_a$ and $X_b$ denote a random split of the dataset $X$ as in the $\coherenceExp_{\cL\circ\alg,X,\cC^*,\selection}(\alpha)$ experiment in \Cref{fig:coherenceExp}.
\end{lemma}

\ifnum\oakland=0
\begin{proof}[Proof of \Cref{lemma:max-info-implies-demographic-coherence-int}]
Fix any arbitrary lens $\rho$. The proof proceeds in two claims. First, in \Cref{lem:max-info-decouples-Xa-from-ha}, we use the definition of max-information to replace $X_a,X_b$ with an independently chosen half-sample $S$ and its complement $\overline{S}=X\setminus S$.

\begin{claim}\label{lem:max-info-decouples-Xa-from-ha}
    Consider $\eta,\alpha,\zeta>0$ and a fixed dataset $X \in \univ^n$. Consider a data-curation algorithm ${\alg:\univ^{n/2}\to\cY}$ with bounded max-information, $I^{\eta}_{\infty,P}(\alg,n/2) \leq \zeta$, and an algorithm $\cL: \cY \to \{ \univ \to [-1,1] \}$ that uses the data report to create a predictor. Independently choose two random half samples $X_a, S\leftarrow X$, and let sets $X_b = X\setminus X_a, \overline{S} = X\setminus S$. Finally let $h \leftarrow \cL(X_a)$. Then, we have that
    \begin{align*}
    &\underset{X_a, h}{\Pr}[\wdist(h(\restrict{X_a\lvert_C}{\rho}), h(\restrict{X_b\lvert_C}{\rho})) > \alpha]\\
    \leq 
    &\underset{S,X_a,h}{\Pr}[\wdist(h(\restrict{S\lvert_C}{\rho}), h(\restrict{\overline{S}\lvert_C}{\rho})) > \alpha]\cdot e^{\zeta} + \eta
    \end{align*}
\end{claim}

\noindent Then, in \Cref{lem:unrelated-predictor}, we bound $\Pr[\wdist\left(g(\restrict{S\lvert_C}{\rho}),g\left(\restrict{\overline{S}\lvert_C}{\rho}\right)\right) > \alpha]$ for any confidence rated predictor $g:\univ \to [-1,1]$ that is produced independently of $S$.
\begin{claim}\label{lem:unrelated-predictor}
Let $\alpha \in (0, 1)$, let $S$ be a sample of size $n/2$ drawn uniformly without replacement from $X$, let $\overline{S} = X\setminus S$,
and let $g:\univ \to [-1,1]$ be any confidence rated predictor.
\medskip\noindent
For any $\mu > 0$, when $|X\cap C|\ge \max\{\frac{4.15}{\alpha^2}\ln(4/\mu), \frac{5.3}{\alpha}, 8.3\ln(4/\mu), 40\}$, we have that
\ifnum\oakland=0
$$\Pr\left[\wdist\left(g(\restrict{S\lvert_C}{\rho}),g(\restrict{\overline{S}\lvert_C}{\rho})\right) > \alpha\right] \leq 2(1 + |X \cap C|)\mu.$$ 
\else
\begin{multline*}
\Pr\left[\wdist\left(g(\restrict{S\lvert_C}{\rho}),g(\restrict{\overline{S}\lvert_C}{\rho})\right) > \alpha\right] \\
\leq 2(1 + |X \cap C|)\mu.
\end{multline*}
\fi
\end{claim}

\noindent Putting these claims together, we get that
\ifnum\oakland=0
\begin{equation*}
    \underset{X_a, h}{\Pr}\left[\wdist(h(\restrict{X_a\lvert_C}{\rho}), h(\restrict{X_b\lvert_C}{\rho})) > \alpha\right] \,\leq\,\, 2\mu(|X \cap C| +1)\cdot e^\zeta + \eta.
\end{equation*}
\else
\begin{multline*}
    \underset{X_a, h}{\Pr}\left[\wdist(h(\restrict{X_a\lvert_C}{\rho}), h(\restrict{X_b\lvert_C}{\rho})) > \alpha\right] \\
    \,\leq\,\, 2\mu(|X \cap C| +1)\cdot e^\zeta + \eta.
\end{multline*}
\fi

\ignore{
\begin{proof}[Proof of \cref{lem:applying-pure-dp-maxinfo-to-private-report}]
    Let $\alg^*:\cX^{n/2} \to \cY$ be an $\eps$-DP algorithm for $\eps \in (0,1]$, and let $\beta > 0$.
    Let $X \in \cX^n$, let $X_a$, $S$ be independent draws (sampled uniformly without replacement) of size $n/2$ from $X$. Note that, by the definition of max-information,  $$I_\infty^\beta(X_a;\alg(X_a)) = \dist_\infty^\beta\Big( \big(\,X_a,\alg(X_a)\,\big) || \big(\,X_a,\alg(S)\,\big)\Big)$$
    Since $\alg$ is differentially private, we know that there is bounded max-information between any random sample $X_a$ and the output of $\alg$ on the sample. In particular, applying \Cref{thm:pure-dp-implies-maxinfo} gives us the following
    $$\dist_\infty^\beta\Big( \big(\,X_a,\alg(X_a)\,\big) || \big(\,X_a,\alg(S)\,\big)\Big) \leq \eps^2n/4 + \eps\sqrt{n\ln(2/\beta)/4}$$
    Since $(X_a,\alg(S))$ is distributed exactly the same as $(S,\alg(X_a))$ then we have the following:
    $$\dist_\infty^\beta\Big( \big(\,X_a,\alg(X_a)\,\big) || \big(\,S,\alg(X_a)\,\big)\Big) \leq \eps^2n/4 + \eps\sqrt{n\ln(2/\beta)/4}$$
    By the definition of max-information this means
    $$\log\left(\underset{\{T\,|\,\, \Pr[(X_a,\alg^*(X_a)\in T] > \beta\}}{\sup}\frac{\Pr[(X_a,\alg^*(X_a)) \in T]-\beta}{\Pr[(S,\alg^*(X_a))) \in T]}\right) \leq \eps^2n/4 + \eps\sqrt{n\ln(2/\beta)/4}$$
    Which means that for all $\{T\,|\,\, \Pr[(X_a,\alg^*(X_a)\in T] > \beta\}$,
    $$\log\left(\frac{\Pr[(X_a,\alg^*(X_a)) \in T]-\beta}{\Pr[(S,\alg^*(X_a))) \in T]}\right) \leq \eps^2n/4 + \eps\sqrt{n\ln(2/\beta)/4}$$
    Let $C$ be any fixed subpopulation. Given $C$, we can post-process a pair $(X_a,\alg(X_a))$ to compute $(h(\restrict{X_a\lvert_C}{\selection}))$ and $(h(\restrict{X_b\lvert_C}{\selection}))$. This is because $h \leftarrow \cL(\alg(X_a))$ and $X_b = X\setminus X_a$. Similarly, we can post-process $(S,\alg(X_a))$ to compute $(h(S\lvert_C))$ and $(h(\overline{S}\lvert_C))$.

Let $T = \{ (S,h) \in (X^n,\cY)\,\mid\,\, \wdist(g(S\lvert_C), g(\overline{S}\lvert_C)) > \alpha^* \}$. Then if $\Pr[(X_a,\alg^*(X_a)\in T] > \beta$,
$$\log\left(\frac{\Pr[(X_a,\alg^*(X_a)) \in T]-\beta}{\Pr[(S,\alg^*(X_a))) \in T]}\right) \leq O(\eps^2n + n\sqrt{\delta/\eps})$$
Which means that if $\Pr[\wdist(h(\restrict{X_a\lvert_C}{\selection}), h(\restrict{X_b\lvert_C}{\selection})) > \alpha^* ] > \beta$,
\[
    \log\left(\frac{\Pr[\wdist(h(\restrict{X_a\lvert_C}{\selection}), h(\restrict{X_b\lvert_C}{\selection})) > \alpha^*] - \beta}{\Pr[\wdist(h(S\lvert_C), h(\overline{S}\lvert_C)) > \alpha^*]}\right) \leq \eps^2n/4 + \eps\sqrt{n\ln(2/\beta)/4}
\]
We can rearrange the above equation to get that 
\begin{align*}
    &\Pr[\wdist(h(\restrict{X_a\lvert_C}{\selection}), h(\restrict{X_b\lvert_C}{\selection})) > \alpha^*]\\
    \leq 
    &\Pr[\wdist(h(S\lvert_C), h(\overline{S}\lvert_C)) > \alpha^*]\cdot e^{\eps^2n/4 + \eps\sqrt{n\ln(2/\beta)/4}} + \beta
\end{align*}
\end{proof}
}

Now we proceed to prove Claims~\ref{lem:max-info-decouples-Xa-from-ha}~and~\ref{lem:unrelated-predictor}.

\begin{proof}[Proof of \cref{lem:max-info-decouples-Xa-from-ha}]
    First, note that since the algorithm $\cL$ postprocesses the report output by the data curator, by the fact that max-information is preserved under postprocessing, it inherits its max-information. Let $\alg^*$ be the combined algorithm $\cL \circ \alg$. Then by the definition of max-information, and since $(X_a,\alg^*(S))$ is distributed exactly the same as $(S,\alg^*(X_a))$,
    \ifnum\oakland=0
    $$I_\infty^\eta(X_a;\alg^*(X_a)) = \dist_\infty^\eta\Big( \big(\,X_a,\alg^*(X_a)\,\big) || \big(\,X_a,\alg^*(S)\,\big)\Big) = \dist_\infty^\eta\Big( \big(\,X_a,\alg^*(X_a)\,\big) || \big(\,S,\alg^*(X_a)\,\big)\Big).$$
    \fi
    we have that for all $T$ such that $\Pr[(X_a,\alg^*(X_a)\in T] > \eta$,
    $$\log\left(\frac{\Pr[(X_a,\alg^*(X_a)) \in T]-\eta}{\Pr[(S,\alg^*(X_a))) \in T]}\right)\leq \zeta.$$
   Given $C$, we can post-process a pair $(X_a,\alg^*(X_a))$ to compute $(h(\restrict{X_a\lvert_C}{\selection}))$ and $(h(\restrict{X_b\lvert_C}{\selection}))$. This is because $h \leftarrow \alg^*(X_a)$ and $X_b = X\setminus X_a$. Applying the same post-processing to $(S,\alg^*(X_a))$ yields $(h(\restrict{S\lvert_C}{\selection}))$ and $(h(\restrict{\overline{S}\lvert_C}{\selection}))$.
    
    \medskip\noindent
    Let $T = \{ (S,h) \,\mid\,\, \wdist(h(\restrict{S\lvert_C}{\selection}), h(\restrict{\overline{S}\lvert_C}{\selection})) > \alpha \}$.
    Then if $\Pr[(X_a,\alg^*(X_a))\in T] > \eta$,
    $$\log\left(\frac{\Pr[(X_a,\alg^*(X_a)) \in T]-\eta}{\Pr[(S,\alg^*(X_a))) \in T]}\right) \leq \zeta.$$
   This means that if $\Pr[\wdist(h(\restrict{X_a\lvert_C}{\selection}), h(\restrict{X_b\lvert_C}{\selection})) > \alpha] > \eta$,
    \[
    \log\left(\frac{\Pr[\wdist(h(\restrict{X_a\lvert_C}{\selection}), h(\restrict{X_b\lvert_C}{\selection})) > \alpha] - \eta}{\Pr[\wdist(h(\restrict{S\lvert_C}{\selection}), h(\restrict{\overline{S}\lvert_C}{\selection})) > \alpha]}\right) \leq \zeta.
    \]
We can rearrange the above equation to get that 
\begin{align*}
    &\Pr[\wdist(h(\restrict{X_a\lvert_C}{\selection}), h(\restrict{X_b\lvert_C}{\selection})) > \alpha]\\
    \leq 
    &\Pr[\wdist(h(\restrict{S\lvert_C}{\selection}), h(\restrict{\overline{S}\lvert_C}{\selection})) > \alpha]\cdot e^{\zeta} + \eta,
\end{align*}
as required.
\end{proof}

\begin{proof}[Proof of \cref{lem:unrelated-predictor}]
    
\noindent 
Let $\alpha, \mu \in (0, 1)$ (the statement holds trivially for $\mu > 1$), and suppose $|X\cap C|\ge \max\{\frac{4.15}{\alpha^2}\ln(4/\mu), \frac{5.3}{\alpha}, 8.3\ln(4/\mu), 40\}$. By the definition of the distance metric we have the following:
\ifnum\oakland=0
\begin{equation}\label{eq:wasserstein-integral}    \wdist\left(g(\restrict{S\lvert_C}{\selection}),g(\restrict{\overline{S}\lvert_C}{\selection})\right) = \int_{-1}^{1} \abs{\cdf_{g(\restrict{S\lvert_C}{\selection})}(g) - \cdf_{g(\restrict{\overline{S}\lvert_C}{\selection})}(g)} \,dg.
\end{equation}
\else
\begin{align}\label{eq:wasserstein-integral}
& \wdist\left(g(\restrict{S\lvert_C}{\selection}),g(\restrict{\overline{S}\lvert_C}{\selection})\right) \nonumber \\
   & \qquad = \int_{-1}^{1} \abs{\cdf_{g(\restrict{S\lvert_C}{\selection})}(\ell) - \cdf_{g(\restrict{\overline{S}\lvert_C}{\selection})}(\ell)} \,d \ell
\end{align}

\fi

To bound this value, we first prove the following for a fixed $y \in [-1,1]$.
\begin{equation}\label{eq:cdf-diffs-single-y}
    \Pr\left[|\cdf_{g(\restrict{S\lvert_C}{\selection})}(y) - \cdf_{g(\restrict{\overline{S}\lvert_C}{\selection})}(y)| \ge \alpha\right]\le \mu.
\end{equation}
Then, we observe that there are at most $|X \cap C|+1$ effectively different values of $y$ we need to consider with respect to any fixed $g$ and $C$. (For every realization of $S, \overline{S}$, $\cdf_{g(\pi_{\selection}(S|_C))}$ can only change for values of $y$ on which $\cdf_{g(\pi_{\selection}(X|_C))}$ changes. These values correspond to the partitioning of $[-1,1]$ into intervals induced by applying $g \circ \pi_{\selection}(\cdot)$ to the elements in $X \cap C$.) By union bounding over these $|X \cap C| + 1$ effectively different values of $y$, 
\Cref{eq:cdf-diffs-single-y} gives us the following.
\ifnum\oakland=0
\begin{equation}\label{eq:supremum-of-cdf-diffs}
\Pr\left[\sup_{y\in[-1,1]}\abs{\cdf_{g(\restrict{S\lvert_C}{\selection})}(y) - \cdf_{g(\restrict{\overline{S}\lvert_C}{\selection})}(y)} \ge \alpha\right]\le (1+|X\cap C|)\mu.
\end{equation}
\else
\begin{align}\label{eq:supremum-of-cdf-diffs}
\Pr&\left[\sup_{y\in[-1,1]}\abs{\cdf_{g(\restrict{S\lvert_C}{\selection})}(y) - \cdf_{g(\restrict{\overline{S}\lvert_C}{\selection})}(y)} \ge \alpha\right] \nonumber \\
& \qquad \le (1+|X\cap C|)\mu.
\end{align}
\fi

Finally, substituting the bound from \Cref{eq:supremum-of-cdf-diffs} in \Cref{eq:wasserstein-integral} proves the lemma.

To show \Cref{eq:cdf-diffs-single-y}, we reinterpret the sampling process in terms of hypergeometric distributions, defined as follows.

\begin{definition}[Hypergeometric distribution]\label{def:hyper}
	Fix $0 < a, s \leq b$. Consider a population of $b$ items of which $a$ items have a special property $P$. Consider $s$ items sampled without replacement from $b$. The distribution of the number of items in $s$ with property $P$ is called the hypergeometric distribution parameterized by $b,a,s$ (denoted by $H(b,a,s)$). 
\end{definition}

This will in turn allow us to invoke \Cref{thm:hypgeom}, a concentration result for the hypergeometric distribution.

\begin{theorem} [\cite{HushS05}]\label{thm:hypgeom}
	Let $K$ have a hypergeometric distribution $H(b, a, s)$. Then for every $\gamma \ge 2$,
	\ifnum\usenix=1
	\begin{multline*}
		\Pr\left[K > s\frac{a}{b} + \gamma \right] < e^{-2 c(\gamma^2 - 1)}  \text{ and }\\
		\Pr \left[K < s\frac{a}{b} - \gamma \right]  < e^{-2 c(\gamma^2 - 1)}
	\end{multline*}
	\else
	\begin{align*}
		\Pr[K > s\frac{a}{b} + \gamma] &< e^{-2 c(\gamma^2 - 1)} \\
		\Pr[K < s\frac{a}{b} - \gamma] &< e^{-2 c(\gamma^2 - 1)}
	\end{align*}
	\fi
	where
	\[c= \max\left\{\frac{1}{s+1} + \frac{1}{b-s+1}, \frac{1}{a+1} + \frac{1}{b-a+1}\right\}.\]
\end{theorem}

Consider an urn consisting of $n$ balls. Among those $n$ balls, $m$ are marked with a red stripe, representing membership in $C\cap X$. Among the $m$ red-striped balls, $t$ are further marked with a blue stripe, representing $x\in C\cap X$ such that $g(\restrict{x}{\selection}) \leq y$ for the value of $y$ being considered. Consider the experiment where we sample $n/2$ balls without replacement, and define the joint pair of random variables $(V, W)$ where $V$ counts the number of red-striped balls in the sample, (i.e., the number of sampled points that are in $X\cap C$) and $W$ counts the number of (red and) blue-striped balls in the sample, (i.e., the number of sampled points $x$ that are in $X\cap C$ and satisfy $g(\restrict{x}{\selection}) \leq y$. The random variables $W$ and $V$ follow hypergeometric distributions as follows:
\ifnum\usenix=1
		$$V \sim H(n, m, n/2) \text{ and }
		(W | V = v) \sim H(m, t, v).$$
\else
    \begin{align*}
		V &\sim H(n, m, n/2) \\
		(W | V = v) &\sim H(m, t, v).
	\end{align*}
\fi
Observe that the absolute value of the CDF difference we are trying to bound is equal to $\left|\frac{W}{V} - \frac{t - W}{m - V}\right|$ by definition.

Let $E_1$ be the event that the number of red-striped balls in the sample is close to its expected value (i.e., $|V - m/2| < m/4$). Then applying Theorem~\ref{thm:hypgeom} and using $m > 40$ and $m > 8.3\ln(4/\mu)$ we have that
\begin{align*}
    \Pr[\overline{E_1}] &< 2\exp\left(\frac{-2}{m+1} \cdot \Big( (m/4)^2 - 1\Big)\right)\\
    &\leq 2\exp\left(\frac{-2}{1.025m} \cdot \Big( 0.99(m/4)^2 \Big)\right)\\
    &= 2\exp\left(\frac{1.98}{16.4}m\right) < \mu/2.
\end{align*}
	
Now let us condition on a realization $V = v$. Given this, let $E_2$ be the event that the number of blue-striped balls in the sample is close to its expected value (i.e., $\left|W -\frac{tv}{m} \right| \leq \zeta$). Then applying Theorem~\ref{thm:hypgeom} for $\zeta \geq 2$ and $c$ as in the theorem:
    \[\Pr[ \overline{E_2} \bigmid V = v] < 2\exp\left(-2c\cdot \left( \zeta^2 - 1\right)\right).\]
Observe that 
\ifnum\oakland=0
    $$c = \max\left\{\frac{1}{v+1} + \frac{1}{m-v+1}, \frac{1}{t+1} + \frac{1}{m-t+1}\right\} \geq \frac{1}{t+1} + \frac{1}{m-t+1} \geq \frac{2}{\frac{m}{2}+1}.$$
\else
  \begin{align*}
  c & = \max\left\{\frac{1}{v+1} + \frac{1}{m-v+1}, \frac{1}{t+1} + \frac{1}{m-t+1}\right\} \\
  & \geq \frac{1}{t+1} + \frac{1}{m-t+1} \geq \frac{2}{\frac{m}{2}+1}.
  \end{align*}
\fi
Therefore,
\begin{equation}\label{eq:zeta}
    \Pr[\overline{E_2} \bigmid V = v] < 2\exp\left(\frac{-4}{\frac{m}{2}+1}\cdot \left( \zeta^2 - 1\right)\right).
\end{equation}

Assume that both events $E_1$ and $E_2$ hold. Then in this case we will argue that:
\begin{equation}\label{eq:alpha}
    \left|\frac{W}{V} - \frac{t - W}{m - V}\right| < \frac{m\zeta}{m^2/4 - \gamma^2} = \alpha.
\end{equation}
We will then substitute the derived value of $\zeta$ into Equation~\ref{eq:zeta} to show that $\Pr[ \overline{E_2} \bigmid V = v] < \mu/2$.

To this end, observe that if the number of blue-striped balls in the sample is within $\zeta$ of the expected value, $\frac{tV}{m}$, then the number of blue-striped balls in the sample is also within $\zeta$ of its own expected value, $\frac{t(m-V)}{m}$. This is because the balls can only be in one of these two sets. Therefore, when both events $E_1$ and $E_2$ hold, we have that:
\[\quad \Big|(t-W)-\frac{t(m-V)}{m} \Big| < \zeta.\]
Therefore,
\ifnum\oakland=0
    \begin{align*}
    \frac{W}{V}\, -\, &\frac{t-W}{m - V}\\
    &< \big(\E[W] + \zeta\big)\cdot\frac{1}{V} - \big(\E[t-W] - \zeta\big)\frac{1}{m - V}
    &&\left(\text{because }\abs{W-\frac{tV}{m}} < \zeta\right)\\
    &= \left(\frac{tV}{m} + \zeta\right)\cdot\frac{1}{V} - \left(\frac{t(m-V)}{m} - \zeta\right)\cdot\frac{1}{m - V}\\
    &=\frac{m\zeta}{V(m-V)}\\
    &< \frac{m\zeta}{m^2/4 - \gamma^2}
    &&\left(\text{because }|V - m/2| < \gamma\right).
    \end{align*}
\else
    \begin{align*}
    \frac{W}{V}\, -\, &\frac{t-W}{m - V} < \big(\E[W] + \zeta\big)\cdot\frac{1}{V} - \big(\E[t-W] - \zeta\big)\frac{1}{m - V}
    \\
    &= \left(\frac{tV}{m} + \zeta\right)\cdot\frac{1}{V} - \left(\frac{t(m-V)}{m} - \zeta\right)\cdot\frac{1}{m - V}\\
    &=\frac{m\zeta}{V(m-V)} < \frac{m\zeta}{m^2/4 - \gamma^2}.
    \end{align*}
    where the first inequality is because $\abs{W-\frac{tV}{m}} < \zeta$, and the final inequality is because $|V - m/2| < \gamma$.
\fi
 Similarly, 
    $$\frac{t-W}{m - V} - \frac{W}{V} < \frac{m\zeta}{m^2/4 - \gamma^2}.$$
    This gives us Equation~\ref{eq:alpha}. We can now set $\alpha = \frac{m\zeta}{m^2/4 - \gamma^2}$ to get that $\zeta = \frac{3m\alpha}{16}$. 
    Now, substituting this $\zeta$ value back into Equation~\ref{eq:zeta} and using $m > 40$, $m>16/3\alpha$, and $m > \frac{4.15}{\alpha^2}\ln(4/\mu)$ we have the following:
    
    \begin{align*}
        \Pr& [ \overline{E_2} \bigmid V = v] < \exp\left(\frac{-4}{\frac{m}{2}+1}\cdot \left( \zeta^2 - 1\right)\right)\\
        &= 2\exp\left(\frac{-4}{\frac{m}{2}+1}\cdot \left( \frac{9m^2\alpha^2}{16^2} - 1\right)\right)\\
        &\leq 2\exp\left(\frac{-8}{1.05m}\cdot \left( \frac{0.9\cdot9m^2\alpha^2}{16^2}\right)\right)\\
        &= 2\exp\left(\frac{8.1\alpha^2m}{33.6}\right) < \mu/2.
    \end{align*}

    Finally, we get the following:
    \[\Pr\left[\left|\frac{W}{V} - \frac{t - W}{m - V}\right| > \alpha\right]\leq \Pr[\overline{E_1}\lor\overline{E_2}]\leq \mu.\] 
\end{proof}

\noindent
With Claims~\ref{lem:max-info-decouples-Xa-from-ha}~and~\ref{lem:unrelated-predictor} now proved, this concludes the proof of \Cref{thm:max-info-implies-demographic-coherence}.
\end{proof}

\smallskip\noindent\textbf{Proof of \Cref{thm:max-info-implies-demographic-coherence}.}
In the remaining part of this section, we use Lemma~\ref{lemma:max-info-implies-demographic-coherence-int} to prove that any data curation algorithm $\alg$ with bounded max-information also enforces Wasserstein-coherence.

\begin{proof}[Proof Of ~\Cref{thm:max-info-implies-demographic-coherence}]

Fix $\eta > 0$. Fix any subpopulation $C \in \cC$. Consider any dataset $X$. Then for any algorithm $\cL:\cY\to\{\univ\to [-1,1]\}$, dataset $X \in \univ^n$, $\alpha^*\in(0,1)$, and $\mu > 0$, such that 
    \begin{align}\label{eq:condgrpsize}
        |X\cap C|\ge \max\left\{\frac{4.15\cdot\ln(4/\mu)}{{\alpha^*}^2}, \frac{5.3}{\alpha^*},8.3\cdot\ln(4/\mu), 40\right\},
    \end{align}
we have the following (by \Cref{lemma:max-info-implies-demographic-coherence-int})
 \ifnum\oakland=0   
    \begin{equation}\label{eq:probdemcoh}
        \Pr_{X_a,X_b \leftarrow X, h\leftarrow\cL\circ\alg(X_a)}\left[ \wdist\Big(h(\restrict{X_a|_C}{\selection}),\, h(\restrict{X_b|_C}{\selection})\Big) > 2\alpha^* \right] \leq  2(|X\cap C|+1)\mu\cdot e^{\zeta} +\eta
    \end{equation} 
\else
    \begin{align}\label{eq:probdemcoh}
        \Pr&_{X_a,X_b \leftarrow X, h\leftarrow\cL\circ\alg(X_a)} \Big[ d_{\mathcal{W}_1}\Big(h(\restrict{X_a|_C}{\selection}) \nonumber \\
        & ,\, h(\restrict{X_b|_C}{\selection})\Big) > \alpha \Big] \nonumber\\
        & \qquad \leq  2(|X\cap C|+1)\mu\cdot e^{\zeta} +\eta
    \end{align} 
\fi
where the choice of split $X_a,X_b \setrandomly X$ and the computed predictor $h\leftarrow \cL(X_a)$ are as in the $\coherenceExp_{\cL\circ\alg,X,\cC^*,\selection}(\alpha)$ experiment in \Cref{fig:coherenceExp}. 
    
We show, instead, that\footnote{Looking ahead, we set $\alpha^* = \alpha/2$ to get the final theorem.}
\ifnum\oakland=0
\begin{align*}
\gamma = \max\Big\{ & \frac{8.3 \cdot (\zeta + \ln(16|\cC|/\beta))}{{\alpha^*}^2}, \frac{36\ln((3/{\alpha^*})}{{\alpha^*}^2} \nonumber &, 16.6 \cdot (\zeta + \ln(16|\cC|/\beta)), \frac{5.3}{\alpha^*}, 80\Big\},
\end{align*}
\else
\begin{align*}
\gamma & = \max\Big\{ \frac{8.3 \cdot (\zeta + \ln(16|\cC|/\beta))}{{\alpha^*}^2}, \frac{36\ln((3/{\alpha^*})} {{\alpha^*}^2} \nonumber \\ 
&, 16.6 \cdot (\zeta + \ln(16|\cC|/\beta)), \frac{5.3}{\alpha^*}, 80\Big\},
\end{align*}
\fi
and all algorithms $\cL:\cY\to(\univ\to[-1,1])$, all datasets $X \in \univ^n$, the probability the following holds for all $C \in \cC$ (such that $|X|_C| \geq \gamma$,) is low:
 \[\wdist\left( h\big(\restrict{X_a|_C}{\selection}\big),\, h\big(\restrict{X_b|_C}{\selection}\big) \right) > 2 {\alpha^*},\]
where the split $X_a,X_b \setrandomly X$ and predictor $h\leftarrow \cL(X_a)$ are as in \Cref{fig:coherenceExp}.

\bigskip\noindent
To that end, we start by considering the fixed subpopulation $C$ and setting $\mu = \eta/\left(2(|X\cap C|+1)\cdot e^{\zeta} \right)$ (ensuring that the RHS of \Cref{eq:probdemcoh} is 2 $\eta$ ). 

    The main content in the proof will be arguing that the following lower bound on the size of $|X \cap C|$    implies the condition in \Cref{eq:condgrpsize}. We can then union bound over all sub-populations in $\cC$ to get the theorem.
    \begin{align}
        |X \cap C| \geq \max\Big\{ & \frac{8.3 \cdot (\zeta + \ln(16/\eta))}{{\alpha^*}^2}, \frac{36\ln(3/{\alpha^*})}{{\alpha^*}^2} \nonumber \\
        &, 16.6 \cdot (\zeta + \ln(16/\eta)), \frac{5.3}{\alpha^*}, 80\Big\}.
    \end{align}

    Substituting the value of $\mu$ back in \Cref{eq:condgrpsize}, the first term in the max corresponds to the condition
    \[
        |X\cap C|\ge \frac{4.15\cdot\ln(8 (|X\cap C|+1) e^{\zeta}/\eta)}{{\alpha^*}^2}
    \]
    which can also be written as:
    \[
        |X\cap C|\ge \frac{4.15\cdot\ln((|X\cap C|+1) 
 + g}{{\alpha^*}^2}
    \]
   where $g = 4.15 \cdot (\zeta + \ln(8/\eta))$. 

    Assume $ \frac{|X\cap C|}{2} \ge \frac{4.15\cdot\ln((|X\cap C|+1)}{{\alpha^*}^2}.$ Then, as long as $|X\cap C| \geq \frac{2g}{{\alpha^*}^2}$, the condition is satisfied. Now, using the fact that $|X \cap C| \geq 80$, we have that $\ln((|X\cap C|+1) \leq 1.01\ln((|X\cap C|+1)$, which implies that it suffices for
     $ |X\cap C| \ge \frac{9\cdot\ln((|X\cap C|)}{{\alpha^*}^2}.$ Consider the inequality  $ \frac{|X\cap C|}{\ln((|X\cap C|)} \ge  c$. Note that the left hand side is an increasing function of $|X \cap C|$. Let $|X\cap C| \geq 2c \ln c$. Then, we get that  $ \frac{|X\cap C|}{\ln((|X\cap C|)} \ge \frac{2c \ln c}{\ln(2c \ln c)}$, and some arithmetic shows that for $c \geq 9$ (which is true whenever $\alpha^* \leq 1$), the right hand side is indeed larger than $c$. Hence, it is additionally sufficient that  $ |X\cap C| \ge \frac{36\ln((3/{\alpha^*})}{{\alpha^*}^2}.$

     Similarly, to be larger than the third term in the max in~\Cref{eq:condgrpsize}, we need 
    \[
        |X\cap C|\ge 8.3 \cdot\ln(8 (|X\cap C|+1) e^{\zeta}/\eta)
    \]
    which can also be written as
    \[
        |X\cap C|\ge 8.3\cdot\ln((|X\cap C|+1) 
 + g
    \]
   where $g = 8.3 \cdot (\zeta + \ln(8/\eta))$. 

     Assume $ \frac{|X\cap C|}{2} \ge 8.3 \cdot\ln((|X\cap C|+1).$ Then, as long as $|X\cap C| \geq 2g$, the condition is satisfied. Now, using the fact that $|X \cap C| \geq 80$, we have that $\ln((|X\cap C|+1) \leq 1.01\ln((|X\cap C|+1)$, which implies that it suffices for
     $ |X\cap C| \ge 18\cdot\ln((|X\cap C|)$,which is true as long as $|X \cap C| \geq 80$, hence this case is taken care of. 

     Hence, for a single subpopulation $C$ we have argued that it is sufficient that
    \begin{align}
        |X \cap C| \geq \max\Big\{ & \frac{8.3 \cdot (\zeta + \ln(8/\eta))}{{\alpha^*}^2}, \frac{36\ln(3/{\alpha^*})}{{\alpha^*}^2} \nonumber \\
        &, 16.6 \cdot (\zeta + \ln(8/\eta)), \frac{5.3}{\alpha^*}, 80\Big\}.
    \end{align}

    Setting $\alpha^*=\alpha/2$, $\eta = \beta/2|\cC|$, and applying a union bound on \Cref{eq:probdemcoh} (over all subpopulations in the class) then gives the theorem.
\end{proof}

\fi

\fi

\ifnum\oakland=0
\subsection{Differentially Private Algorithms Enforce Demographic Coherence}\label{sec:dp-implies-coherence-enforcement}
\else
\subsection{Formal Statements Relating Differentially Privacy and Demographic Coherence}\label{sec:dp-implies-coherence-enforcement}
\fi

In this section, we state formal connections between differential privacy and demographic coherence. We do this by adapting known connections between differential privacy and max-information \ifnum\oakland=1 (see the full version of the paper) \fi    and
connecting max-information and demographic coherence. 

\begin{definition}[Differential Privacy \cite{DworkMNS16j}]
	Let $n\in\mathbb{N}$. A randomized algorithm $\alg:\mathcal{X}^n \to \cY$ is $(\epsilon, \delta)$-\emph{differentially private} if for all subsets $Y \subseteq \cY$ of the output space, and for all neighboring datasets $X, X' \in \mathcal{X}^n$ (i.e. $\|X - X' \|_0 \leq 1$), we have that
	\begin{align*}
		\Pr[\alg(X) \in Y] \leq e^{\eps} \Pr[\alg(X') \in Y] + \delta
	\end{align*}
\end{definition}

If $\delta = 0$, we refer to the algorithm as satisfying \noakland{pure differential privacy (pure-DP), whereas $\delta > 0$ corresponds to approximate DP (approx-DP).}
\yoakland{pure-DP, whereas $\delta > 0$ corresponds to approx-DP.}

\ifnum\oakland=0
The proofs of the following theorems connecting differential privacy and max-information can be found in Appendix~\ref{app:maxinfo}.
\begin{restatable}{theorem}{puredpmaxinfo}\label{thm:pure-dp-implies-maxinfo}
	\emph{(Pure-DP $\implies$ Bounded Max-Information)}  Fix $n \in \mathbb{N}$, $\eps > 0$ and let $\univ$ be a data universe of size at least $n$.
	Let $\alg:\cX^{n/2} \to \cY$ be an order-invariant $\eps$-DP algorithm. Then for any $\gamma > 0$,
	$$I^{\gamma}_{\infty}(\alg, n/2) \leq \eps^2n/4 + \eps\sqrt{n\ln(2/\gamma)/4}.$$
\end{restatable}

The following theorem is a generalized version of that in \cite{RogersRST16}. The proof follows theirs, with the following key distinctions: (1) it applies to sampling without replacement (2) It carefully tracks constants and 
(3) It maintains flexibility in setting parameters. We anticipate that this version of the result might be independently useful.

\begin{restatable}{theorem}{approxmaxinfogeneral}\label{thm:approx-dp-implies-max-info}
	\emph{(Approx-DP $\implies$ Bounded Max-Information, Generalised)} Let $\mathcal{A}: \cX^n \to \cY$ be an (order-invariant) $(\eps,\delta)$-differentially
	private algorithm for $\eps \in (0,1/2]$, $\delta \in \left (0,\eps \right )$.
	For $\hat{\delta} \in (0,\eps/15]$, $t>0$, and $\beta(t,\hat\delta) = e^{-t^2/2} + n\left(\frac{2\delta}{\hat{\delta}} + \frac{2\hat{\delta} + 2\delta}{1-e^{-3\eps}}\right)$ we have
	\ifnum\oakland=0
	\begin{align*}
		I^\beta_{\infty}(\cA,n)
		&\leq n\left( 347\hat{\delta} + 75 \left(\frac{\hat{\delta}}{\eps} \right)^2 + 24\frac{\hat{\delta}^2}{\eps}+ 240\eps^2\right)
		+ 6t\eps\sqrt{n}.
	\end{align*}
	\else
	\begin{align*}
		I^\beta_{\infty}(\cA,n)
		\leq & n\Big( 347\hat{\delta} + 75 \left(\frac{\hat{\delta}}{\eps} \right)^2 + 24\frac{\hat{\delta}^2}{\eps}+ 240\eps^2\Big) \\ &
		+ 6t\eps\sqrt{n}.
	\end{align*}

	\fi
	
\end{restatable}

\begin{restatable}{corollary}{approxmaxinfo}\label{cor:approx-dp-implies-max-info}
	\emph{(Approx-DP $\implies$ Bounded Max-Information, Specific)} Fix $n \in \mathbb{N}$, for $\eps \in (0,1/2]$, $\gamma \in (0,1]$, $\delta \in (0,\frac{\eps^2 \gamma^2}{(120n)^2}]$ and let $\univ$ be a data universe of size at least $n$.
	Let $\mathcal{A}: \cX^{n} \to \cY$ be an (order-invariant) $(\eps,\delta)$-differentially
	private algorithm. We have that
	\begin{align*}
		I^\gamma_{\infty}(\alg,n)
		&\leq 265\eps^2n + 12\eps\sqrt{n\ln(2/\gamma)}.
	\end{align*}
\end{restatable}
\fi

\ifnum\oakland=0
The proofs for pure differential privacy and approximate differential privacy follow a similar flavor. Firstly, we adapt known connections between (pure and approximate) differential privacy and max-information to the setting of sampling without replacement. 
Then, we use~\Cref{thm:max-info-implies-demographic-coherence}
(connecting bounded max-information to demographic coherence) to argue that differential privacy implies demographic coherence. 
\fi
\ifnum\oakland=1
We also include proofs of the following theorems connecting differential privacy and demographic coherence in the full version:
\fi

\ifnum\usenix=1 For the formal proofs, see~\Cref{sec:DPdemcoh}. \fi 

\begin{restatable}{theorem}{pureDPcoh}
    [Pure-DP Enforces Wasserstein Coherence]\label{thm:pure-dp-implies-coherence-enforcement}
    Fix any $\eps, \beta \in (0,1]$, $\alpha \in (0,2], n \in \mathbb{N}$.
     Let $\cC$ be a collection of subpopulations $C \in \univ^*$. Consider an order-invariant $\eps$-DP algorithm ${\alg:\univ^{n/2}\to \cY}$. Fix any lens $\rho$. Then, $\alg$ enforces $(\alpha,\beta)$-Wasserstein-coherence with respect to  collection $\cC$, lens $\rho$, and size constraint $\gamma$, where $\gamma$ is the maximum of the following terms
     \ifnum\oakland=0
     \begin{align}
        \gamma= \max\Big\{ & \frac{33.2 \cdot (\eps^2n/4 + \eps\sqrt{n\ln(4|\cC|/\beta)}/2 + \ln(16|\cC|/\beta))}{\alpha^2}, \frac{144\ln((6/\alpha)}{\alpha^2} \nonumber \\
        &, 16.6 \cdot (\eps^2n/4 + \eps\sqrt{n\ln(4|\cC|/\beta)}/2 + \ln(16|\cC|/\beta)), \frac{10.6}{\alpha}, 80\Big\}.
    \end{align}
    \else
   \begin{align}
        & \Big\{ \frac{33.2 \cdot (\eps^2n/4 + \eps\sqrt{n\ln(4|\cC|/\beta)}/2 + \ln(16|\cC|/\beta))}{\alpha^2}, \nonumber \\ & 144\cdot\ln(6/\alpha)/\alpha^2, 16.6 \cdot (\eps^2n/4 + \eps\sqrt{n\ln(4|\cC|/\beta)}/2 \nonumber \\
        & \qquad + \ln(16|\cC|/\beta)), 10.6/\alpha, 80\Big\}.
    \end{align}
    \fi 
\end{restatable}

\ifnum\oakland=0
\begin{proof}

    Fix $\beta > 0$. By \Cref{thm:pure-dp-implies-maxinfo} connecting differential privacy and max-information, we have that we have that,
    $$I^{\beta/2|\cC|}_{\infty}(\alg,n/2) \leq \eps^2n/4 + \eps\sqrt{n\ln(4|\cC|/\beta)/4}.$$

    Applying~\Cref{thm:max-info-implies-demographic-coherence} and substituting the bound on max-information then completes the proof.
\end{proof}
\fi

\begin{restatable}{theorem}{approxDPcoh}
[Approx-DP Enforces Wasserstein Coherence]\label{thm:approx-dp-implies-coherence-enforcement}
    Fix any $\beta \in (0,1]$, $\alpha \in (0,2)$, $n \in N$. Let $\eps \in (0,\frac{1}{2}]$, and $\delta \in (0, \frac{\eps^2 \beta^2}{(120n)^2 |\cC|^2}]$ 
    Let $\cC$ be a collection of subpopulations $C \in \univ^*$. Consider an order-invariant \footnote{Order-invariance can be relaxed by multiplying $\eps$ by $2$ in the $\gamma$ value, and dividing by $2$ in the range of $\delta$.} $(\eps,\delta)$-DP algorithm ${\alg:\univ^{n/2}\to \cY}$. Fix any lens $\rho$. Then, $\alg$ enforces $(\alpha,\beta)$-Wasserstein-coherence with respect to  collection $\cC$, lens $\rho$, and size constraint $\gamma$, where $\gamma$ is the maximum of the following terms
    \ifnum\oakland=0
     \begin{align}
        \gamma = \max\Big\{ & \frac{33.2 \cdot (133\eps^2n + 12\eps\sqrt{n\ln(4|\cC|/\beta)} + \ln(32|\cC|/\beta))}{\alpha^2}, \frac{144\ln(6/\alpha)}{\alpha^2} \nonumber \\
        &, 16.6 \cdot (133\eps^2n + 12\eps\sqrt{n\ln(4|\cC|/\beta)} + \ln(32|\cC|/\beta)), \frac{10.6}{\alpha}, 80\Big\}.
    \end{align}
    \else
\begin{align}
           & \Big\{ \frac{33.2 \cdot (133\eps^2n + 9\eps\sqrt{n\ln(4|\cC|/\beta)} + \ln(32|\cC|/\beta))}{\alpha^2}, \nonumber \\
        & 144\cdot\ln(6/\alpha)/\alpha^2, 16.6 \cdot (133\eps^2n + 9\eps\sqrt{n\ln(4|\cC|/\beta)} \nonumber \\ 
        & \qquad + \ln(32|\cC|/\beta)), 10.6/\alpha, 80\Big\}.
\end{align}
 \fi 
\end{restatable}

\ifnum\oakland=0
\begin{proof}
  Fix $\beta > 0$ and $\gamma = \frac{\beta}{2|\cC|}$. By \Cref{cor:approx-dp-implies-max-info} connecting differential privacy and max-information, we have that we have that, as long as $\delta \in (0, \frac{\eps^2 \beta^2}{(120n)^2 |\cC|^2}]$ 
  \begin{align*}
  I^{\beta/2|\cC|}_{\infty}(\alg,n/2) \leq \frac{265}{2} \eps^2n + 12\eps\sqrt{\frac{n}{2}\ln(4|\cC|/\beta)}. 
  \end{align*}

    Applying~\Cref{thm:max-info-implies-demographic-coherence} and substituting the bound on max-information then completes the proof.
    
\end{proof}
\fi

\subsection{Using Our Theorems for Setting DP Parameters}\label{sec:parametersetting}
Recall that one of our goals in designing Demographic Coherence was to support data curators as they go through the process of parameter selection, as described in \Cref{sec:intro}.  Namely, by providing a \emph{necessary} privacy condition that is more concrete, there is a hope that data curators can make parameter decisions in a contextualized way.  In this section, we leverage \Cref{thm:pure-dp-implies-coherence-enforcement} to illustrate how this might work for a data curator interested in using DP, using the American Community Survey (ACS) as a motivating example.  The ACS is an randomized annual survey that provides in-depth demographic and housing information on respondents, including demographics, employment, income, taxes, living conditions, insurance, etc\ldots The sample includes $\approx3.5$ million households each year, resulting in data about $\approx5$ million individuals included in the final data set. The collected data is crucial input for many, including  local governments, community organizers, researchers, and businesses, but the in-depth nature of the collected data comes with significant privacy concerns.  Moreover, there are ongoing conversations about how DP should be used with ACS data products \cite{Daily_2022}.

Using our framework requires a collaborative effort between socio-technical researchers, privacy experts, legal experts, and data curators to identify concrete privacy risks that may accrue from a data release.  For example, imagine that the data curators for the ACS are concerned that the data they release might be used by hiring employers to identify job candidates who have dependents (and therefore might be more likely to miss work to care for those dependents).  Directly asking candidates for this information may be illegal in some circumstances, but can also be predictive of the subset of job candidates the employer might prefer. The ACS gathers detailed information about household composition, which could leak information about the dependent responsibilities of candidates within the sample.

\begin{table}[t]
\centering
    \begin{tabular}{ccccc}
    \toprule
         $\alpha$ & $\beta$ & $|\mathcal{C}|$ & $\gamma$ & $\eps$  \\
         \midrule
         $0.1$ & $e^{-30}$ & 100 & $.05n$ & .003 \\
         $0.2$ & $e^{-30}$ & 100 & $.05n$ & .01 \\
         \midrule
         $0.3$ & $e^{-30}$ & 100 & $.05n$ & .02 \\
         $0.4$ & $e^{-30}$ & 100 & $.05n$ & .025 \\
                  \midrule
         $0.8$ & $e^{-30}$ & 100 & $.05n$ & 0.056 \\
         $0.9$ & $e^{-30}$ & 100 & $.05n$ & 0.065 \\
                  \midrule
         $1.3$ & $e^{-30}$ & 100 & $.05n$ & 0.09 \\
         $1.4$ & $e^{-30}$ & 100 & $.05n$ & 0.1 \\
         \bottomrule
    \end{tabular}
    \vspace{1em}
    \caption{Parameters examples that satisfy \Cref{thm:pure-dp-implies-coherence-enforcement} when $n = 5 \times 10^6$ (the ACS dataset size).}
   \label{tab:parameters}
   \end{table}

Having selected a concrete predictive harm, demographic coherence asks the data curator to perform parameter selection for this harm within the demographic coherence experiment.  Specifically, we need to select values for $\alpha$ and $\beta$.  
For example, the curator might decide that predictive algorithms should not ``implicitly'' reconstruct more than 10\% of the data subjects' dependents status. This can be seen as having 10\% of the survey population for whom the predictor would have predicted $-1$ \sstext{\sout{without access to the data release} based on statistical aggregate information}, but the predictor outputs $1$ instead.\footnote{Because demographic coherence is independent of accuracy, the model's baseline assumption could equivalently be $-1$.  Alternatively, a model might have a good estimate $(\pm p$ for $p\in[0,1]$) for the population-level  base-rate due to previous ACS survey results.}  If a predictive algorithm were able to conduct such an implicit reconstruction, it would result in a Wasserstein distance of at least $0.2$ between the prediction distributions on the data subjects and another similarly distributed set of individuals,\footnote{Formally, we consider a random dataset split.  Informally, however, it is reasonable to think about demographic coherence as comparing the performance of a predictor on individuals in the dataset and unobserved individuals.} as $0.1$ of the total mass of the probability distribution would be moved $2$ units 
Thus, a data curator concerned about such a harm may want to set $\alpha$ to $0.2$ (or less) to eliminate such a risk.  The curator can then set $\beta$ to be a small constant, like $e^{-30}$, much as $\delta$ is set to a small constant in approximate DP.
 
Finally, demographic coherence asks the data curator to fix a category class $\mathcal{C}$ (with associated minimum size of a category $\gamma$) for which this privacy guarantee would hold. These might be demographic communities against whom there is a long history of discrimination, \eg black women.  Alternatively, they might be tailored to the specific set of harms about with the data curator is concerned.  For example, men between the ages of 18 and 25 might be particularly unlikely to have dependents (\ie a good model might default to predictions of nearly $-1$ for this population) and therefore might have their job prospects particularly impacted by disclosure of their dependent status.  This later cohort constitutes about 5\% of the US population (and, therefore, should also constitute about 5\% of the ACS sample). The data curator team could collaboratively identify many such group of interest, attempting to be as comprehensive as possible. For the sake of this example, let us assume there are about 100 potential categories within $\mathcal{C}$, each constituting 5\% of the dataset ($\gamma \geq 250,000$).

Given these choices of $\alpha$, $\beta$, and $|\mathcal{C}|$ (with associated $\gamma$), we can leverage \Cref{thm:pure-dp-implies-coherence-enforcement} to identify a value of $\eps$ needed such that an $\eps$-DP algorithm ensures $(\alpha, \beta)$-demographically coherent predictions with respect to $\mathcal{C}$ and size constraint $\gamma$.  Plugging in the value above, we find that $\eps = 0.01$ is sufficient.  We note, however, that \Cref{thm:pure-dp-implies-coherence-enforcement} does not imply that $\eps = 0.01$ is \emph{necessary}, as our theorem is not optimal; future work is likely to show that a higher value of $\eps$ may also be sufficient.  We give other sets of satisfying parameters for \Cref{thm:pure-dp-implies-coherence-enforcement} in \cref{tab:parameters}.  We note that this value of $\eps$ also corresponds to a certain level of utility (for various data analysis tasks such as linear regressions, clustering, parameter estimation etc.). Hence our framework precipitates concrete tradeoffs between privacy concerns and desired utility guarantees for the intended applications, which can then be the subject of robust conversation in order to choose a value of $\eps$ that best balances these considerations.

The above illustration shows that the process of parameter selection can be made more concrete and tangible using demographic coherence.  Challenges, however, remain.  For example, there is nothing inherently good or bad about $10\%$, meaning that ambiguity will remain in the process. But, this is significantly more concrete than existing approaches in the literature.  

Moreover, the experiments in~\cref{sec:reallifevul} suggest natural experimental tools that can help data curators get a clearer sense for what any given value of $\alpha$ \emph{feels} like---they can choose and train natural learning algorithms on data releases, and plot the corresponding prediction distributions of the associated predictor as well as compute the Wasserstein distance. Observing the value of Wasserstein distance at which the prediction distribution seem disconcertingly different visually can help build intuition about setting $\alpha$.

\ifnum\oakland=0
\ifnum\oakland=0
\section{Non-Differentially Private Algorithms that Enforce Wasserstein Demographic Coherence }\label{sec:nondp}
\else
\subsection{Non-Differentially Private Algorithms That Enforce Wasserstein Demographic Coherence }\label{sec:nondp}

\fi

\gsk{For Oakland, I think we need to have some more intro text here explaining why this is an interesting thing to do. \sstext{I've made an attempt}}

In this section, we analyze an algorithm that is not differentially private ($\eps = \infty$), but satisfies excellent demographic coherence parameters, thereby showing that our framework is provably weaker than differential privacy. This also lends support to the encouraging possibility that it is possible to achieve a better privacy-utility tradeoff for many tasks with demographic coherence (in conjunction with other necessary conditions) than with differential privacy. 

We consider an algorithm that releases a single (shuffled) binary column of a dataset (with possibly many columns) in the clear.  Such a curator can be thought of as publishing the outcome of a vote without adding noise, thereby allowing for the accurate evaluation of the outcomes of such votes even when they are very close (which would be impossible with differential privacy). 

This algorithm is clearly not differentially private since the algorithm reproduces a shuffled version of one of its input columns exactly, but we will argue that it enforces Wasserstein demographic coherence with good parameters, owing to the fact that a single binary column with no additional context reveals essentially no information about the split (since the $1$s and $0$s in the column are partitioned roughly equally between the two halves of the split and since the information released is uncorrelated with other attributes due to shuffling). In particular, the data released does not have enough complexity to allow for a violation of demographic coherence.

Fix a universe $\univ$. Consider a data curator $\alg : (\univ \times \zo)^{n/2}\to\zo^{n/2}$ that (deterministically) publishes a binary column of its input in the clear.

We will start by considering the lens of the adversary at prediction time being just the column released, hence ensuring there is an absence of information correlated to the bit (\eg personal identifiers) that the adversary can use. 

\begin{theorem}\label{thm:nondpcoherence}
    Fix any $\beta \in (0,1]$, $\alpha \in (0,2], n \in \mathbb{N}$. Let $\univ$ be an arbitrary space, and let $\cC$ be a collection of subpopulations $C \in \{\univ \times \zo \}^*$. Consider algorithm $\alg: (\univ \times \zo)^{n/2}\to\zo^{n/2}$ that publishes one binary column of the input in the clear. Let the lens $\rho$ be this one column.
     
    Then, $\alg$ enforces $(\alpha,\beta)$-Wasserstein-coherence with respect to  collection $\cC$, lens $\rho$, and size constraint $
        \gamma= \frac{4\sqrt{n \log (|\mathcal{C}|/\beta)}}{\alpha}.$
    
\end{theorem}

\ifnum\oakland=0
\begin{proof}
 Fix any category $C \in \mathcal{C}$. Fix any dataset $X \in \zo^n$ (we will restrict to the one column that is being released since the others will not matter). Let $\beta' = \frac{\beta}{|\mathcal{C}|}$.

 We start by stating a concentration inequality we will use in the proof.

\begin{theorem}[Hoeffding's Inequality Without Replacement]\label{thm:hoeffworeplacement}
Let $R = (r_1,\dots,r_n) \in \{0,1\}^N$ be a finite population, and let $R_1,\dots,R_m$ be a random sample of size $m$ drawn without replacement from it. Let $\mu = \frac{1}{n} \sum_{i=1}^n R_i$. Then,
\begin{align*}
    \Pr\left[\frac{1}{m}\sum_{i=1}^m R_i - \mu \geq \alpha \right] \leq e^{-2m\alpha^2}.
\end{align*}

\end{theorem}

Now, consider a random split $X_a,X_b \leftarrow X$ of the dataset. Let $N_a$ be the number of ones in category $C$ in $X_a$, and $N_b$ be the number of ones in category $C$ in $X_b$ and $M_a$ and $M_b$  be the number of zeros in category $C$ in $X_a$ and $X_b$ respectively. Setting $m =n/2$, $\alpha = \sqrt{\frac{\log (1/\beta')}{n}}$, and using Theorem~\ref{thm:hoeffworeplacement} where we let $R_i = \mathrm{1}[X_i \in C, X_i = 1]$, and consider the dataset of $R_i$ values and apply the random split to this.
\begin{align*}
  \Pr\left[\sum_{i \in X_a} R_i - \frac{n}{2}\mu \geq \sqrt{n \log (1/\beta')} \right] \leq \beta'.
\end{align*}
 
Observe that this implies that $N_a$ is at most $\sqrt{n \log (1/\beta')}$ from half the number of ones in the category $C$ in the dataset, which implies that $|N_a - N_b| \leq 2\sqrt{n \log (1/\beta')}$ with probability at least $1-\beta'$. Also, note that $|N_a - N_b| = |M_a - M_b|$.

Now, consider any predictor $h$ on the elements of the dataset- which can be captured by the values $h(0)$ and $h(1)$. 

Then, with probability at least $1-\beta'$ over the split $X_a, X_b$, (and also using the fact that $|N_a - N_b |= |M_a - M_b|$),
\ifnum\oakland=0
\begin{align*}
\wass\Big(h(\restrict{X_a|_C}{\selection}),\, h(\restrict{X_b|_C}{\selection})\Big)  & = 
|h(1) - h(0)| \cdot \frac{|N_a - N_b|}{|C \cap X|} \leq 2\frac{|N_a - N_b|}{|C \cap X|} \leq  \frac{4\sqrt{n \log (1/\beta')}}{|C \cap X|}
\end{align*}
\else

\begin{align*}
\wass\Big(h(\restrict{X_a|_C}{\selection}),\, h(\restrict{X_b|_C}{\selection})\Big)  = \\
 |h(1) - h(0)| \cdot \frac{|N_a - N_b|}{|C \cap X|} \leq 2\frac{|N_a - N_b|}{|C \cap X|} \\
 \leq  \frac{4\sqrt{n \log (1/\beta')}}{|C \cap X|}
\end{align*}

\fi

Since this applies for all predictors, it applies in particular for any predictor produced by the algorithm $\mathcal{L}$.

Now, taking a union bound over all categories $C \in \mathcal{C}$, we get that with probability at least $1-\beta$ over the split $X_a, X_b$ and the randomness of the algorithm $\mathcal{L}$ generating the predictor $h$,
\begin{align*}
\wass\Big(h(\restrict{X_a|_C}{\selection}),\, h(\restrict{X_b|_C}{\selection})\Big) \leq  \frac{4\sqrt{n \log (|\mathcal{C}|/\beta)}}{|C \cap X|}
\end{align*}

Hence, we get that $\alg$ enforces $(\alpha, \beta)$-demographic coherence with respect to collection $\mathcal{C}$, lens $\rho$, and size constraint $\gamma$, where $\gamma =\frac{4\sqrt{n \log (|\mathcal{C}|/\beta)}}{\alpha}$.

\end{proof}
\fi

\ifnum\oakland=1
The proof of this theorem can be found in the full version of the paper.
\fi

Now, we will handle the more general case where the lens of the adversary is the entire input space- that is, the adversary sees all attributes of the input when making a prediction. 

The sanitization algorithm will involve shuffling the column before outputting it as the report (or equivalently outputting the number of $1$s in the column- as we will think about it in our analysis). Intuitively, the report should still not give the adversary much information about the split, since it only tells the adversary how many records with bit $1$ in a specific column there are in $X_a$, but not \textit{which} records they are.

However, the analysis in this case is significantly trickier and involves careful conditioning arguments and analysis of hypergeometric distributions. Consequently, the size constraint is slightly worse (by a logarithmic factor in the dataset size), and also has worse constants. \ifnum\oakland=1 The proof of this result can also be found in the full version of the paper. \fi

\begin{theorem}\label{thm:nondpcoherencefulllens}
    Fix any $\beta' \in (0,1]$, $\alpha \in (0,2]$, sufficiently large $n \in \mathbb{N}$. Let $\univ$ be an arbitrary space, and let $\cC$ be a collection of subpopulations $C \in \{\univ \times \zo \}^*$. Consider algorithm $\alg: (\univ \times \zo)^{n/2}\to \mathbb{N}$ that publishes the number of $1$s in the binary column in the input. Let the lens $\rho$ be the entire input space $\univ \times \{0,1\}$.
     
    Then, $\alg$ enforces $(\alpha,\beta)$-Wasserstein-coherence with respect to  collection $\cC$, lens $\rho$, and size constraint $
        \gamma= \frac{112\sqrt{n \log (10n|\cC|/\beta')}}{\alpha}.$
    
\end{theorem}

\ifnum\oakland=0
\begin{proof}
    Let the binary column whose number of $1$s is released in the clear be column $i$ of the dataset. Fix a subpopulation $C \in \cC$. Let the number of $1$s in this column in the population dataset be $L_1$ and the number of $0$s be $L_2$. Let $N_a$ be the number of ones sampled into $X_a$ and $N_b$ be the number of zeros sampled into $X_a$. By Hoeffding's inequality without replacement (Theorem~\ref{thm:hoeffworeplacement}), we have that with probability at least $1-\beta$, $|N_a - L_1/2| \leq \sqrt{n \log(1/\beta)}$ and $|N_b - L_2/2| \leq \sqrt{n \log(1/\beta)}$ (where we have also used that $N_a + N_b = n/2$). 

    Now, condition on a specific report $N_a$ that satisfies these inequalities $|N_a - L_1/2| \leq \sqrt{n \log(1/\beta)}$ and $|N_b - L_2/2| \leq \sqrt{n \log(1/\beta)}$. Conditioned on this, the split ($X_a, X_b$) can be thought of in the following way: $X_a$ is generated as follows: $N_a$ records with column $i$ value being one are sampled without replacement from the $L$ records with this property in the dataset and $N_b = n/2 - N_a$ records with column $i$ value being zero are sampled without replacement from $L_2$ records satisfying this property in the dataset. $X_b$ is then simply $X \setminus X_a$.

    Now, consider an adversary producing a predictor $h$ based on the report $N_a$. Condition on the coins of the adversary. Observe that this fixes the corresponding predictor $h$ - since all the randomness generating $h$ has been handled via conditioning (we have already conditioned on the report $N_a$ and have also now conditioned on the random coins of the algorithm). 

    We will now reason about the demographic coherence properties of $h$. 

\begin{align}
& \wdist\left(h(\restrict{X_a \lvert_C}{\selection}),h(\restrict{X_b \lvert_C}{\selection})\right) \nonumber \\
   & \qquad = \int_{-1}^{1} \abs{\cdf_{h(\restrict{X_a\lvert_C}{\selection})}(\ell) - \cdf_{h(\restrict{X_b \lvert_C}{\selection})}(\ell)} \,d \ell \label{eq:cdfdiff} \\
   & \qquad \leq \int_{-1}^{1} \Big[\abs{\cdf_{h(\restrict{X_a\lvert_C}{\selection})}(\ell) - \cdf_{h(\restrict{X \lvert_C}{\selection})}(\ell)} \nonumber \\
   & \qquad +\abs{\cdf_{h(\restrict{X_b\lvert_C}{\selection})}(\ell) - \cdf_{h(\restrict{X \lvert_C}{\selection})}(\ell)} \Big]  \,d \ell 
\end{align}

Fix a specific $\ell \in [-1,1]$. We will leverage the properties of the hypergeometric distribution in order to reason about the CDF difference in question, and then use this to bound the Wasserstein difference of interest.

We redefine the sampling process in a way that will allow us to use \Cref{thm:hypgeom}, a concentration result for the hypergeometric distribution (See Definition~\ref{def:hyper} for a definition of this distribution): Consider an urn consisting of $L_1$ balls representing records whose $i{th}$ column is $1$. Among those $L_1$ balls, $m_1$ are marked with a red stripe, representing membership in $C\cap X$. Among the $m_1$ red-striped balls, $t_1$ are further marked with a blue stripe, representing $x\in C\cap X$ such that $h(\restrict{x}{\selection}) \leq \ell$ for the value of $\ell$ being considered. Consider the experiment where we sample $N_a$ balls without replacement, and define the joint pair of random variables $(V_1, W_1)$ where $V_1$ counts the number of red-striped balls in the sample, (i.e., the number of sampled points that are in $X\cap C$) and $W_1$ counts the number of (red and) blue-striped balls in the sample, (i.e., the number of sampled points $x$ that are in $X\cap C$ and satisfy $h(\restrict{x}{\selection}) \leq y$. Similarly, consider an analogous case with another urn with $L_2$ balls (representing records with $i^{th}$ column equal to $0$), $m_2$ balls marked with a red stripe (indicating membership in subgroup), $t_2$ balls further marked with a blue stripe (indicating hypothesis value smaller than $\ell$), and $N_b$ balls sampled without replacement, with $(W_2, V_2)$ defined analogously.

The random variables $W_1, V_1, W_2, V_2$ follow hypergeometric distributions as follows:

    \begin{align*}
		V_1 &\sim H(L_1, m_1, N_a) \\
		(W_1 | V_1 = v_1) &\sim H(m_1, t_1, v_1). \\
        V_2 &\sim H(L_2, m_2, N_b) \\
		(W_2 | V_2 = v_2) &\sim H(m_2, t_2, v_2).
	\end{align*}

Observe that the absolute value of the CDF difference we are trying to bound is equal to $\left|\frac{W_1 + W_2}{V_1 + V_2} - \frac{t_1 + t_2}{m_1 + m_2}\right|$ by definition.

Let $E_1$ be the event that the number of red-striped balls in the sample is close to its expected value for Urn 1. (i.e., $|V_1 - m_1 N_a / L_1| < \sqrt{n \log(1/\beta)}$). Then applying Theorem~\ref{thm:hypgeom} and using the fact that $n$ is sufficiently large,
\begin{align*}
    \Pr[\overline{E_1}] &< 2\exp\left(\frac{-2}{N_a+1} \cdot \Big( n \log(1/\beta) - 1\Big)\right)\\
    &\leq \beta.
\end{align*}

Let $E_2$ be the same event for Urn 2. Hence, by the same argument, we get that with probability at least $1-\beta$, $|V_2 - m_2 N_b / L_2| < \sqrt{n \log(1/\beta)}$.
	
Now let us condition on a realization $V_1 = v_1$. Given this, let $G_1$ be the event that the number of blue-striped balls in the sample for Urn $1$ is close to its expected value (i.e., $\left|W_1 -\frac{t_1v_1}{m_1} \right| \leq \zeta$). Then applying Theorem~\ref{thm:hypgeom} for $\zeta \geq 2 = \sqrt{n \log 1/\beta}$ and $c$ as in the theorem, and using the fact that $n$ is sufficiently large:
    \[\Pr[ \overline{E_2} \bigmid V = v] < 2\exp\left(-2c\cdot \left( \zeta^2 - 1\right)\right).\]
Observe that 
  \begin{align*}
  c & = \max\left\{\frac{1}{v_1+1} + \frac{1}{m_1-v_1+1}, \frac{1}{t_1+1} + \frac{1}{m_1-t_1+1}\right\} \\
  & \geq \frac{1}{t_1+1} + \frac{1}{m_1-t_1+1} \geq \frac{2}{\frac{m_1}{2}+1}.
  \end{align*}

Therefore,
\begin{equation}\label{eq:zeta1}
    \Pr[\overline{G_1} \bigmid V = v] < 2\exp\left(\frac{-4}{\frac{m_1}{2}+1}\cdot \left( \zeta^2 - 1\right)\right) \leq \beta.
\end{equation}

Similarly, if we define $G_2$ similarly for Urn $2$ and condition on $V_2= v_2$, we get that 

\begin{equation}\label{eq:zeta2}
    \Pr[\overline{G_2} \bigmid V = v] < 2\exp\left(\frac{-4}{\frac{m_2}{2}+1}\cdot \left( \zeta^2 - 1\right)\right) \leq \beta.
\end{equation}

Now, substituting these values into the CDF difference, we can write that with probability at least $1-5\beta$,

\begin{align*}
& \frac{W_1 + W_2}{V_1 + V_2} \\
& \leq \frac{W_1 + W_2}{m_1 N_a / L_1 + m_2 N_b / L_2 -  2\sqrt{n \log 1/\beta}} \\
& \leq \frac{W_1 + W_2}{m_1 / 2 + m_2 / 2 - 4\sqrt{n \log(1/\beta)}}\\
& \leq  \frac{t_1 V_1 / m_1 + t_2 V_2 / m_2 + 2\sqrt{n \log1/\beta}}{m_1 / 2 + m_2 / 2 - 4\sqrt{n \log(1/\beta)}} \\
& \leq \frac{t_1 V_1 / m_1 + t_2 V_2 / m_2 + 2\sqrt{n \log1/\beta}}{m_1 / 2 + m_2 / 2 - 4\sqrt{n \log(1/\beta)}} \\
& \leq \frac{t_1 N_a / L_1 + t_2 N_b / L_2 + 4\sqrt{n \log1/\beta}}{m_1 / 2 + m_2 / 2 - 4\sqrt{n \log(1/\beta)}} \\
& \leq \frac{t_1 / 2 + t_2 / 2 + 6\sqrt{n \log1/\beta}}{m_1 / 2 + m_2 / 2 - 4\sqrt{n \log(1/\beta)}}
\end{align*}
Substituting in the CDF difference, we get that
\begin{align*}
  &  \frac{W_1 + W_2}{V_1 + V_2} - \frac{t_1 + t_2}{m_1 + m_2} \\
  & \leq \frac{t_1 + t_2 + 12\sqrt{n \log1/\beta}}{m_1 + m_2 - 8\sqrt{n \log(1/\beta)}} - \frac{t_1 + t_2}{m_1 + m_2} \\
  & \leq \frac{t_1 + t_2 + 12\sqrt{n \log1/\beta}}{m_1 + m_2 - 8\sqrt{n \log(1/\beta)}} - \frac{t_1 + t_2 -  8\sqrt{n \log(1/\beta)}}{m_1 + m_2 - 8\sqrt{n \log(1/\beta)}} \\
  & \leq \frac{20 \sqrt{n \log1/\beta}}{m_1 + m_2 - 8\sqrt{n \log(1/\beta)}} \leq \alpha/4
\end{align*}
when $|C \cap X| = m_1 + m_2 \geq \frac{112\sqrt{n \log(1/\beta)}}{\alpha}$.

Similarly, with probability at least $1-5\beta$,

\begin{align*}
& \frac{W_1 + W_2}{V_1 + V_2} \\
& \geq \frac{W_1 + W_2}{m_1 N_a / L_1 + m_2 N_b / L_2 +  2\sqrt{n \log 1/\beta}} \\
& \geq \frac{W_1 + W_2}{m_1 / 2 + m_2 / 2 + 4\sqrt{n \log(1/\beta)}}\\
& \geq  \frac{t_1 V_1 / m_1 + t_2 V_2 / m_2 - 2\sqrt{n \log1/\beta}}{m_1 / 2 + m_2 / 2 - 4\sqrt{n \log(1/\beta)}} \\
& \geq \frac{t_1 V_1 / m_1 + t_2 V_2 / m_2 - 2\sqrt{n \log1/\beta}}{m_1 / 2 + m_2 / 2 - 4\sqrt{n \log(1/\beta)}} \\
& \geq \frac{t_1 N_a / L_1 + t_2 N_b / L_2 - 4\sqrt{n \log1/\beta}}{m_1 / 2 + m_2 / 2 - 4\sqrt{n \log(1/\beta)}} \\
& \geq \frac{t_1 / 2 + t_2 / 2 - 6\sqrt{n \log1/\beta}}{m_1 / 2 + m_2 / 2 + 4\sqrt{n \log(1/\beta)}}
\end{align*}

Substituting in the CDF difference, we get that
\begin{align*}
  &  \frac{t_1 + t_2}{m_1 + m_2} - \frac{W_1 + W_2}{V_1 + V_2}  \\
  & \leq \frac{t_1 + t_2}{m_1 + m_2} - \frac{t_1 + t_2 - 12\sqrt{n \log1/\beta}}{m_1 + m_2 + 8\sqrt{n \log(1/\beta)}} \\
  & \leq \frac{t_1 + t_2 + 8\sqrt{n \log1/\beta}}{m_1 + m_2 + 8\sqrt{n \log(1/\beta)}} \\
  & - \frac{t_1 + t_2 -  12\sqrt{n \log(1/\beta)}}{m_1 + m_2 + 8\sqrt{n \log(1/\beta)}} \\
  & \leq \frac{20 \sqrt{n \log1/\beta}}{m_1 + m_2 - 8\sqrt{n \log(1/\beta)}} \leq \alpha/4
\end{align*}
when $|C \cap X| = m_1 + m_2 \geq \frac{112\sqrt{n \log(1/\beta)}}{\alpha}$.

Hence, we get that with probability at least $1-5\beta$, $\abs{\cdf_{h(\restrict{X_a\lvert_C}{\selection})}(\ell) - \cdf_{h(\restrict{X \lvert_C}{\selection})}(\ell)}\leq \frac{\alpha}{4}$. By a symmetric argument, with probability at least $1-5\beta$, $\abs{\cdf_{h(\restrict{X_b\lvert_C}{\selection})}(\ell) - \cdf_{h(\restrict{X \lvert_C}{\selection})}(\ell)} \leq \frac{\alpha}{4}$, giving that with probability at least $1-10\beta$, $\abs{\cdf_{h(\restrict{X_a\lvert_C}{\selection})}(\ell) - \cdf_{h(\restrict{X_b \lvert_C}{\selection})}(\ell)}\leq \frac{\alpha}{2}$. Now, observe that the CDF difference can change at at most $n$ points, since datasets $X_a, X_b$ each contain $n/2$ points. Hence, setting $\beta = \beta'/10n|\cC|$, and using a union bound, we get that with probability at least $1-\beta'/|\cC|$, for all $\ell \in [-1,1]$,  $\abs{\cdf_{h(\restrict{X_a\lvert_C}{\selection})}(\ell) - \cdf_{h(\restrict{X_b \lvert_C}{\selection})}(\ell)}\leq \frac{\alpha}{2}$. Hence, we get that with probability at least $1-\beta'/|\cC|$ over the randomness of the algorithm and the split, $\wdist\left(h(\restrict{X_a \lvert_C}{\selection}),h(\restrict{X_b \lvert_C}{\selection})\right) \leq \alpha$. Hence, by a union bound, we get that with probability at least $1-\beta'$ over the randomness of the algorithm and the split, for all categories $C \in \cC$, $\wdist\left(h(\restrict{X_a \lvert_C}{\selection}),h(\restrict{X_b \lvert_C}{\selection})\right) \leq \alpha$. Further substituting the value of $\beta$ into the restriction on $|C \cap X|$, we get that $|C \cap X| \geq \frac{112\sqrt{n \log(10 n |\cC|/\beta)}}{\alpha}$ is sufficient as the size constraint, completing the proof.
\end{proof}
\fi

Given our results, one might be tempted to conjecture that the above results can be extended to releasing multiple columns in the clear. However, the amount of information released about the split grows rapidly with the number of columns released in the clear, and so this intuition breaks down. In fact, in Section~\ref{sec:reallifevul}, we took a real dataset with many integer-valued columns and demonstrated a predictor that witnesses very large Wasserstein distance between the two predictive distributions restricted to a specific relevant subgroup, thereby indicating that releasing entire datasets in the clear is not coherence enforcing. When curators want to release more complex statistics than a single vote, our results in Section~\ref{sec:acheivingdc} show that differentially private algorithms give us an avenue to do this. We leave it as an interesting open question to investigate other ways to achieve demographic coherence.


 




\fi

\ifnum\oakland=1
\section{Ethical Considerations}

Our work is motivated by the desire to create ergonomic \emph{necessary} notions of privacy that simultaneously preserve the intuition about the harms against which they protect and are rigorous enough to inform algorithm design. At the same time, we are cognizant of \textit{privacy washing} that can be an unintentional side effect of formal work on privacy, where data products can leverage some formal privacy guarantees to convince the public of their commitment to privacy, while still causing significant privacy harms. Hence, we emphasize that our privacy notion is a \textit{necessary condition}, and should not be seen as one that covers all possible privacy risks, even if it offers the possibility of improved utility. In particular, we think other notions such as privacy narcissus singling-out \cite{CohenKMMNST24} offer a complementary perspective capturing different privacy harms and should be considered simultaneously with our notion to ensure robustness to different types of privacy risks.



        

    



\fi

\ifnum\usenix=1
\section{Conclusion}
In this work, we introduce demographic coherence, a new analytical framework for reasoning about the privacy of large data releases. This framework represents a \emph{necessary} condition for privacy, bridging the gap between stringent privacy conditions that are \emph{sufficient} for privacy preservation, and the intuitive understanding of the harms those conditions are meant to prevent. The framework is deliberately designed to be ergonomic in many different contexts; it captures an intuitive class of harms, while offering new analytical and experimental tools for privacy research. Additionally, we show that this framework can be instantiated, arguing that any differentially private algorithm enforces Wasserstein-demographic coherence. We leave it as an important open question to identify other ways of instantiating our definition.

\clearpage
\fi

\ifnum\oakland=0
\section{Acknowledgments}
We would like to thank Adam Smith for insightful comments and feedback that helped improve this work. We thank Christina Xu for her contributions to shaping the early phases of this project.
\fi

\bibliographystyle{plain}
\bibliography{citation-stuff/extra_citations,citation-stuff/abbrev3,citation-stuff/crypto,citation-stuff/references} 

\ifnum\oakland=0
\newpage
\fi

\appendix

\ifnum\oakland=1

\input{sections/bounded max info implies coherence}

\fi

\ifnum\oakland=0
\section{Comparison to Perfect Generalization}\label{sec:comparePG}

In this section we include a comparison between 
\emph{demographic coherence} (\Cref{def:coherence})
and the notion of \emph{perfect generalization} introduced by Cummings et al.~\cite{CummingsLNRW16}. (See also the work of Bassily and Freund~\cite{BassilyF16} that independently introduced a generalization thereof.) This notion was originally meant to capture generalization under post-processing but has since been shown to be closely related to other desirable properties  as well (\eg replicability~\cite{BunGHILPSS23}). Our framework shares conceptual similarities with this definition, but the technical details differ in important ways.

The following comparison uses the definition of  \emph{sample perfect generalization} (\Cref{def:sample-perfect-generalization}) from~\cite{BunGHILPSS23} which is roughly equivalent to the original definition from~\cite{CummingsLNRW16}. Intuitively, a mechanism $M$ running on i.i.d. samples from some distribution (sample) perfectly generalizes if the distribution of its output does not depend ``too much'' on specific realization of its sampled input. That is, its output distributions when run on two i.i.d samples from any distribution are indistinguishable.

\begin{definition}[Sample perfect generalization {\cite[{Def~3.4}]{BunGHILPSS23}} ]\label{def:sample-perfect-generalization}
An algorithm $\cA : \cX^m \to \cY$ is said to be $(\beta, \epsilon, \delta)$-sample perfectly generalizing if, for every distribution $\mathcal{D}$ over $\cX$, with probability at least $1 - \beta$ over the draw of two i.i.d. samples $X_a, X_b \sim \mathcal{D}^m$,
\[
\cA(X_a) \approx_{\epsilon, \delta} \cA(X_b),
\]
where $\approx_{\epsilon, \delta}$ denotes $\epsilon,\delta$ indistinguishability.
\end{definition}

\Cref{def:coherence} has several noticeable syntactic differences when compared to \Cref{def:sample-perfect-generalization}. First, demographic coherence is defined within a specific framework that explicitly lays out the entire data release pipeline, a design choice that intentionally lends itself to concrete intuition (and experimental evaluation) of data release. However, this still leaves open the possibility that the core statistical guarantee of demographic coherence is roughly equivalent to perfect generalization. In other words, it may still be the case that demographic coherence is simply a different way to describe the protections offered by perfect generalization; as we see below, this is not the case.

Second, while ``closeness'' in the definition of perfect generalization is required for distributions over the entire sets $X_a, X_b$, 
the ``closeness'' in the definition of demographic coherence is required for distributions over the sets $X_a|_C,X_b|_C$ for subpopulations $C \subseteq \univ$ from some collection $\cC$. For the sake of drawing a more direct comparison here, we collapse this difference by comparing sample perfect generalization to demographic coherence with $\cC=\{\univ\}$.

A third difference in the definitions is is the choice of sets $X_a,X_b$ that the comparison is made with respect to, \ie a random partition of a fixed dataset in demographic coherence vs. i.i.d. draws from a distribution in the case of perfect generalization. The choice of random partitioning in our framework is made to ensure concreteness and applicability in census-like settings but it is chosen intentionally to maintain both intuitive and quantitative similarities to i.i.d. sampling. Thus, we view this as more of a difference in interpretability and applicability of the definitions, rather than one about their underlying guarantees.

The main difference between the two definitions, thus, is in how how ``closeness'' is measured---as spelled out in \Cref{fig:coherence-generalization}. 

\begin{figure*}[ht!]
\begin{center}
\fcolorbox{black}{cyan!3}{
\small
\hbox{
\begin{minipage}{0.944\textwidth}
\vspace{0.3em}
An algorithm $\alg:\cX^{n/2}\to \cY$ is:
\begin{enumerate}
    \item \emph{$(\beta, \eps, \delta)$-sample perfectly generalizing if}

    $\forall$ distributions $\cD$ over $\cX$, with probability at least $1-\beta$ over $X_a,X_b \sim \cD^{n/2}$: 
    $$\alg(X_a) \approx_{\eps,\delta} \alg(X_b)$$

    \item \emph{$(\alpha,\beta)$-coherence enforcing if}

    $\forall$ datasets $X \in \cX^n$, learners $\cL:\cY \to (\cX \to [-1,1])$, with probability at least $1-\beta$ over the random split $X_a \cup X_b = X$ and the coins of $\alg$, $\cL$:
    $$\dist_{W}(h_a(X_a), h_a(X_b)) \leq \alpha$$
    where $h_a \leftarrow \cL\circ\cA(X_a)$ is a confidence rated predictor, and $h_a(X_i)$ is the distribution induced by randomly choosing $x\sim X_i$ and computing $h_a(x)$.
\end{enumerate}
\vspace{0.4em}
\end{minipage}
}
}
\end{center}
\caption{Comparing the definition of sample perfect generalization to a simplified definition of demographic coherence.}
\label{fig:coherence-generalization}
\end{figure*}

Perfect generalization asks that w.h.p. over independent samples, $\alg$ produces indistinguishable distributions over reports $R_a \leftarrow \alg(X_a)$ and $R_b \leftarrow \alg(X_b)$. Meanwhile, coherence enforcement asks that w.h.p. $\cL\circ\alg(X_a)$ produces a confidence rated predictor $h_a:\cX \to [-1,1]$ which has ``similar'' predictions on $X_a$ and $X_b$ (a property enforced by $\alg$). That is, the comparison in perfect generalization is on the behavior of the algorithm $\cA$ itself, while the comparison in demographic coherence is on the likely behavior of a realized hypothesis $h_a$ that is produced only over the report $R_a \leftarrow \alg(X_a)$.

Since coherence enforcement limits the set of algorithms against which indistinguishability applies, one might expect that perfectly generalizing algorithms also enforce demographic coherence, and indeed, \cref{thm:max-info-implies-demographic-coherence} proves this to be true. However, the converse need not be true---implying that demographic coherence is a relaxation of perfect generalization. 
In particular, the example below shows a set $X$ such that no confidence rated predictor $h:\cX \to [-1,1]$ violates this property. In this case, all data curators are vacuously coherence enforcing.
 
Consider a data curator $\alg:\zo^{n/2}\to\zo^{n/2}$ that (deterministically) publishes its input in the clear. This is clearly not perfectly generalizing as the distribution of the report $X_a$ is a point mass that (for reasonable choices of $X$, with high probability) is distinct from the distribution of the report $X_b$.

Meanwhile, considering a dataset $X\in\zo^n$ and a random split $X_a,X_b \leftarrow X$ of the dataset, there are two possible predictors $h:\zo \to [-1,1]$ that witnesses the highest possible Wasserstein distance when run on $X_a$ vs. $X_b$. That is, without loss of generality $h$ is either $h(0)=-1, h(1) = 1$ or $h(0) = 1, h(1) = -1$. In either case, $h$ cannot be improved even by seeing $X_a$ in the clear. So, any data curation algorithm in this scenario, is coherence enforcing since the data itself doesn't have enough complexity to allow for a violation of demographic coherence. Note that the absence of information correlated with the bits contained in the dataset $X$ (\eg time, location, computer system) is crucial to this example.

\section{Differential Privacy implies Bounded Max-Information: Sampling without Replacement }\label{app:maxinfo}

Prior work shows that for datasets sampled i.i.d., differentially private algorithms have bounded max-information \cite{DworkFHPRR15,RogersRST16}. In this appendix we prove \Cref{thm:pure-dp-implies-maxinfo} and \Cref{thm:approx-dp-implies-max-info}, which are analogs of those theorems for sampling without replacement.

\subsection{Preliminaries}\label{sec:prelims-appendix}

First we state \Cref{thm:mcdiarmids_without-our-version}, a version of McDiarmid's inequality that applies to the case of sampling without replacement. This result follows from Lemma~2 in \cite{Tolstikhin17}. This result in used in the proof of \Cref{thm:pure-dp-implies-maxinfo}, which says that pure-DP algorithms have bounded max-information (even in the case of sampling without replacement.)

The following definition of order-invariant algorithms appears as a technical assumption in some of our theorem statements.\footnote{Since it is an assumption in the version of McDiarmid's Inequality for sampling without replacement (\Cref{thm:mcdiarmids_without-our-version}) that we use.} This is a minimal assumption because any non-order-invariant algorithm can be made order-invariant by simply pre-processing the dataset with a sorting or shuffling operation.


\begin{definition}\label{def:data-order-invariant-function}
    A function $f:\cX^m \to \cY$, is called \emph{order invariant} if for all $X \in \cX^m$, the value of the function $f(X)$ does not depend on the order of the elements of $X$. 
\end{definition}

\begin{theorem}[McDiarmid's for sampling without replacement \cite{Tolstikhin17}]\label{thm:mcdiarmids_without-our-version}
    Let $f:\cX^n \to \cY$ be an an order invariant function with global sensitivity $\Delta > 0$. Let $\cX$ be a data universe of size $n$, let $S$ be a sample of size $m$ chosen without replacement from $\cX$. Then for $t \geq 0$,
    \ifnum\usenix=0
        $$\Pr_S [f(S)-\E[f(S)] \geq t] \leq \exp\left(-\frac{2t^2}{m\Delta^2}\cdot\left(\frac{n-1/2}{n-m}\right)\cdot\left(1-\frac{1}{2\max(m,n-m)}\right)\right)$$
    \else
    \begin{align*}
    & \Pr_S [f(S)-\E[f(S)] \geq t] \\
    & \qquad \leq \exp\left(-\frac{2t^2}{m\Delta^2}\cdot\left(\frac{n-1/2}{n-m}\right)\cdot\left(1-\frac{1}{2\max(m,n-m)}\right)\right)
    \end{align*}
    
    \fi    

    In particular, for $m = n/2$ and $n \geq 3$,
        $$\Pr[f(S)-\E[f(S)] \geq t] \leq \exp\left(-\frac{4t^2}{n\Delta^2}\right)$$
\end{theorem}

Next we state some lemmas that are used in the proof of \Cref{thm:approx-dp-implies-max-info} and \Cref{cor:approx-dp-implies-max-info} which say that approximate-DP algorithms have bounded max-information (even in the case of sampling without replacement.)

\begin{definition}[Point-wise indistinguishibility \cite{KasiviswanathanS14}]
Two random variables $A,B$ are point-wise $(\eps, \delta)$-indistinguishable if with probability at least $1-\delta$ over $ a \sim P(A)$:
$$
e^{-\epsilon} \Prob{}{B = a} \leq \Prob{}{A = a} \leq e^{\epsilon} \Prob{}{B=a}.
$$
\end{definition}

\begin{lemma}[Indistinguishability implies Pointwise Indistinguishability \cite{KasiviswanathanS14}] \label{lem:prelims}
 Let $A, B$ be two random variables. If $A \approx_{\eps, \delta} B$ then $A$ and $B$ are pointwise $\left(2\eps, \frac{2\delta}{1-e^{-\eps}} \right) $-indistinguishable.
\end{lemma}

\begin{lemma}[Conditioning Lemma \cite{KasiviswanathanS14}]\label{lem:conditioning}
Suppose that $(A,B) \approx_{\eps, \delta} (A',B')$.  Then for every $\hat \delta>0$, the following holds:
$$
\Prob{t \sim P(B) }{ A|_{B = t} \approx_{3\epsilon, \hat\delta} A'|_{B'=t} } \geq 1-\frac{2\delta}{\hat\delta} - \frac{2\delta}{1-e^{-\eps}}.
$$
\end{lemma}

\begin{theorem}[Azuma's Inequality]
Let $C_1, \cdots, C_n$ be a sequence of random variables such that for every $i \in [n]$, we have
$$
\Prob{}{|C_i| \leq \alpha} = 1
$$
and for every fixed prefix $\mathbf{C}_1^{i-1} = \mathbf{c}_1^{i-1}$, we have
$$
\Ex{}{C_i|\mathbf{c}_1^{i-1}} \leq \gamma,
$$
then for all $t\geq 0$, we have
$$
\Prob{}{\sum_{i=1}^n C_i > n \gamma + t \sqrt{n} \alpha} \leq e^{-t^2/2}.
$$
\label{thm:azuma}
\end{theorem}


\subsection{Pure-DP $\implies$ Bounded Max-Information}

In this appendix we state \Cref{thm:pure-dp-implies-maxinfo}, which is an analog of Theorem~ from \cite{DworkFHPRR15}. The proof of this theorem works exactly as in \cite{DworkFHPRR15}, except replacing the application of McDiarmid's Lemma with a version of McDiarmid's for sampling without replacement (\cref{thm:mcdiarmids_without-our-version}) which we state in \Cref{sec:prelims-appendix}.

\puredpmaxinfo

\subsection{$(\eps,\delta)$-DP $\implies$ Bounded Max-Information}

In this appendix we prove \Cref{thm:approx-dp-implies-max-info}, which is an analog of Theorem~1 from Rogers et al.\cite{RogersRST16}. In fact, the following proof is an adaptation of the proof of them from \cite{RogersRST16} with the following differences: (1) we compute all the constants exactly and avoid using asymptotic notation, and (2) we keep the tunable parameters in the final version of the theorem to obtain the most flexible result that we can. Finally, we set the parameters in \Cref{thm:approx-dp-implies-max-info} to get \Cref{cor:approx-dp-implies-max-info}, which is used in the proof of \Cref{thm:approx-dp-implies-coherence-enforcement} (3) We attempt to make the proof more readable by simplifying the notation. In particular, we do not define ``good sets'' $\cE,\cF$.

\approxmaxinfogeneral

\approxmaxinfo

\newcommand{\uglyterm}{342\frac{\hat{\delta}}{\eps} + 112\frac{\hat{\delta}^2}{\eps^2} + 25\frac{\hat{\delta}^2}{\eps} + 240\eps^2}

We will sometimes abbreviate
conditional probabilities of the form $\Prob{}{\bX=\bbx\mid\cA = a}$ as
$\Prob{}{\bX=\bbx\mid a}$ when the random variables are clear from
context.  Further, for any $\bbx \in \cX^n$ and $ a \in \cY$, we define

\begin{equation}
Z_i(a,\bbx_{[i]}) \defeq 
\log\dfrac{\Prob{}{X_i = x_i \mid  a, \bbx_{[i-1]}}}{\Prob{}{X_i = x_i\mid  \bbx_{[i-1]}}}. \label{eqn:Z_i}
\end{equation}

\begin{align*}
Z(a,\bbx) & \defeq 
\log \left( \dfrac
{\Prob{\bbx}{\cA(x)=a, \bX=\bbx}}
{\Prob{}{\cA = a} \cdot \Prob{}{\bX=\bbx}}
\right) \\
 & = \sum\limits_{i=1}^n Z_i(a,\bbx_{[i]}) \numberthis \label{eqn:ln_sum}
\end{align*}

If we can bound $Z(a,\bbx)$ with high probability over $(a,\bbx) \sim p(\cA(\bX),\bX)$, then we can bound the approximate max-information by using the following lemma:

\begin{lemma}[{\cite[Lemma 18]{DworkFHPRR15}}]
\ifnum\usenix=0
$$\Prob{}{\log \left( \dfrac
{\Prob{\bbx}{\cA(x)=a, \bX=\bbx}}
{\Prob{}{\cA = a} \cdot \Prob{}{\bX=\bbx}}
\right) \geq k} \leq \beta \,\,\implies\, I_\infty^\beta(\cA(\bX);\bX) \leq k.$$
\else
\begin{align*}
\Pr&[\log \left( \dfrac
{\Pr[\cA(x)=a, \bX=\bbx]}
{\Pr[\cA = a] \cdot \Pr[\bX=\bbx]}
\right) \geq k] \\
& \leq \beta \,\,\implies\, I_\infty^\beta(\cA(\bX);\bX) \leq k.
\end{align*}
\fi
\label{lem:boundmaxinfo}
\end{lemma}
To bound $Z(a,\bbx)$ with high probability over $(a,\bbx) \sim p(\cA(\bX),\bX)$ we will apply Azuma's inequality (\Cref{thm:azuma}) to the sum of the $Z_i(a,\bbx_{[i]})$'s. For this we must first argue that each $Z_i(a,\bbx_{[i]})$ term is bounded with high probability:
\begin{claim}\label{claim:zibound}
Let $\hat\delta>0$, and $\delta'' \defeq \tfrac{2\hat\delta}{1-e^{-3\eps}}$.
If $\cA$ is $(\eps,\delta)$-differentially private and, $\bX \in \univ^n$ is sampled without replacement from a finite universe $\univ$, then for each $i \in [n]$, and each prefix $\bbx_{[i-1]}\in \cX^{i-1}$ and answer $a$, we have:
    $$
\Prob{x_i \in \mathsf{Supp}(X_i |_{\bbx_{[i-1]}})}
{\log\dfrac{\Prob{}{X_i = x_i \mid a, \bbx_{[i-1]}}}{\Prob{}{X_i = x_i\mid \bbx_{[i-1]}}} \leq 6 \eps } \geq 1 - \delta''
$$
\label{claim:Fi}
\end{claim}

\begin{proof}
    Whenever $X_i|_{\bbx_{[i-1]}}$ and $X_i|_{a,\bbx_{[i-1]}}$ are $\left(3 \epsilon, \hat{\delta} \right)$-indistinguishable, \Cref{lem:prelims} tells us that $X_i|_{\bbx_{[i-1]}}$ and $X_i|_{a,\bbx_{[i-1]}}$ are point-wise $\left(6 \epsilon, \delta'' \right)$-indistinguishable. \ie given that $X_i|_{\bbx_{[i-1]}}$ and $X_i|_{a,\bbx_{[i-1]}}$ are $\left(3 \epsilon, \hat{\delta} \right)$-indistinguishable, we have that 
$$
\Prob{x_i \in \mathsf{Supp}(X_i |_{\bbx_{[i-1]}})}
{\log\dfrac{\Prob{}{X_i = x_i \mid a, \bbx_{[i-1]}}}{\Prob{}{X_i = x_i\mid \bbx_{[i-1]}}} \leq 6 \eps } \geq 1 - \delta''
$$
\end{proof}

\begin{claim}
Let $\hat\delta>0$, $\delta' \defeq
\tfrac{2\delta}{\hat\delta} + \tfrac{2\delta}{1-e^{-\eps}}$, and $\delta'' \defeq \tfrac{2\hat\delta}{1-e^{-3\eps}}$.
If $\cA$ is $(\eps,\delta)$-differentially private and, $\bX \in \univ^n$ is sampled without replacement from a finite universe $\univ$ , then for each $i \in [n]$, and each prefix $\bbx_{[i-1]}\in \cX^{i-1}$ we have:
    $$
    \Prob{\substack{x_i \in \mathsf{Supp}(X_i |_{\bbx_{[i-1]}})\\
    a \sim \cA|_{\bbx_{[i-1]}}}}
    {\log\dfrac{\Prob{}{X_i = x_i \mid  a, \bbx_{[i-1]}}}{\Prob{}{X_i = x_i\mid  \bbx_{[i-1]}}} \leq 6 \eps } \geq 1- \delta' - \delta''
    $$
\label{claim:Gi}
\end{claim}

\begin{proof}
For this proof, we use \Cref{claim:Fi} and then show for each $i \in [n]$, and prefix $\bbx_{[i-1]}\in \cX^{i-1}$,
$$
   \Prob{a \sim p\left( \cA|_{\bbx_{[i-1]}}\right) }{X_i|_{\bbx_{[i-1]}} \approx_{3\eps,\hat{\delta}} X_i|_{a, \bbx_{[i-1]}}} \geq 1-\delta'. 
$$
We use the differential privacy guarantee on $\cA$ to show that $(\cA,X_i)|_{\bbx_{[i-1]}} \approx_{\eps,\delta} \cA|_{\bbx_{[i-1]}} \otimes X_i|_{\bbx_{[i-1]}}$. The above equation then follows directly from the conditioning lemma (\Cref{lem:conditioning}).

\medskip
To show that $(\cA,X_i)|_{\bbx_{[i-1]}} \approx_{\eps,\delta} \cA|_{\bbx_{[i-1]}} \otimes X_i|_{\bbx_{[i-1]}}$, we will show the following:
$$\prob{(\cA(\bX),X_i )  \in\cO  \mid \bbx_{[i-1]}} \leq e^{\eps}\cdot\prob{\cA(\bX) \otimes X_i   \in \cO  \mid \bbx_{[i-1]}} + \delta.$$ 
($\prob{\cA(\bX) \otimes X_i   \in \cO  \mid \bbx_{[i-1]}}  \leq  e^{\eps}\cdot\prob{(\cA(\bX),X_i) \in \cO \mid \bbx_{[i-1]}} + \delta$ follows from a very similar argument.)

Fix any set $\cO \subseteq \cY \times \cX$ and prefix $\bbx_{[i-1]} \in \cX^{i-1}$. From the differential privacy of $\cA$, and the order-invariance of the algorithm, we get the following:
\ifnum\usenix=0
\begin{align*}
& \prob{(\cA(\bX),X_i )  \in\cO  \mid \bbx_{[i-1]}}\\
&= \sum_{x_i \in \mathsf{Supp}(X_i |_{\bbx_{[i-1]}})}\prob{X_i = x_i \mid \bbx_{[i-1]}} \prob{(\cA(\bX),x_i ) \in \cO  \mid \bbx_{[i-1]}, x_i} \\
&\leq \sum_{x_i \in \mathsf{Supp}(X_i |_{\bbx_{[i-1]}})}\prob{X_i = x_i\mid \bbx_{[i-1]}}\left( e^{\eps} \prob{(\cA(\bX),x_i ) \in \cO \mid \bbx_{[i-1]}, t_i} + \delta \right)
\qquad \forall t_i \in \mathsf{Supp}(X_i |_{\bbx_{[i-1]}})\\
&= \sum_{x_i,t_i \in \mathsf{Supp}(X_i |_{\bbx_{[i-1]}})}\prob{X_i = t_i \mid \bbx_{[i-1]}}\prob{X_i = x_i\mid \bbx_{[i-1]}}\left( e^{\eps} \prob{(\cA(\bX),x_i ) \in \cO \mid \bbx_{[i-1]}, t_i} + \delta \right)
\\
&= \sum_{x_i \in \mathsf{Supp}(X_i |_{\bbx_{[i-1]}})}\prob{X_i = x_i\mid \bbx_{[i-1]}}\left( e^{\eps} \prob{(\cA(\bX),x_i ) \in \cO \mid \bbx_{[i-1]}} + \delta \right)\\
&\leq e^{\eps}\left(\sum_{x_i \in \mathsf{Supp}(X_i |_{\bbx_{[i-1]}})} \prob{X_i = x_i\mid \bbx_{[i-1]}}\prob{\cA(\bX),X_i ) \in \cO \mid \bbx_{[i-1]}}\right) + \delta\\
&= e^{\eps}\prob{\cA(\bX) \otimes X_i   \in \cO  \mid \bbx_{[i-1]}} + \delta
\end{align*}
\else
\begin{align*}
& \prob{(\cA(\bX),X_i )  \in\cO  \mid \bbx_{[i-1]}}\\
&= \sum_{x_i \in \mathsf{Supp}(X_i |_{\bbx_{[i-1]}})}\prob{X_i = x_i \mid \bbx_{[i-1]}} \prob{(\cA(\bX),x_i ) \in \cO  \mid \bbx_{[i-1]}, x_i} \\
&\leq \sum_{x_i \in \mathsf{Supp}(X_i |_{\bbx_{[i-1]}})}\prob{X_i = x_i\mid \bbx_{[i-1]}}\\
& \qquad \qquad \cdot \left( e^{\eps} \prob{(\cA(\bX),x_i ) \in \cO \mid \bbx_{[i-1]}, t_i} + \delta \right)
\quad \forall t_i \in \mathsf{Supp}(X_i |_{\bbx_{[i-1]}})\\
&= \sum_{x_i,t_i \in \mathsf{Supp}(X_i |_{\bbx_{[i-1]}})}\prob{X_i = t_i \mid \bbx_{[i-1]}}\prob{X_i = x_i\mid \bbx_{[i-1]}}\\
& \qquad \qquad \cdot \left( e^{\eps} \prob{(\cA(\bX),x_i ) \in \cO \mid \bbx_{[i-1]}, t_i} + \delta \right)
\\
&= \sum_{x_i \in \mathsf{Supp}(X_i |_{\bbx_{[i-1]}})}\prob{X_i = x_i\mid \bbx_{[i-1]}}\\
& \qquad \qquad \left( e^{\eps} \prob{(\cA(\bX),x_i ) \in \cO \mid \bbx_{[i-1]}} + \delta \right)\\
&\leq \delta + e^{\eps}\sum_{x_i \in \mathsf{Supp}(X_i |_{\bbx_{[i-1]}})} \prob{X_i = x_i\mid \bbx_{[i-1]}}\prob{\cA(\bX),X_i ) \in \cO \mid \bbx_{[i-1]}} \\
&= e^{\eps}\prob{\cA(\bX) \otimes X_i   \in \cO  \mid \bbx_{[i-1]}} + \delta
\end{align*}
\fi
Where, the first inequality is true by the following coupling argument where we separately consider the cases where $t_i \notin \bbx_{[i-1]}$ and those where $t_i \in \bbx_{[i-1]}$. First, whenever $t_i \notin \bbx_{[i-1]}$, The following is true by differential privacy:
$$\prob{(\cA(\bX),x_i ) \in \cO  \mid \bbx_{[i-1]}, x_i} \leq e^{\eps} \prob{(\cA(\bX),x_i ) \in \cO \mid \bbx_{[i-1]}, t_i} + \delta$$
Otherwise, if  $t_i \in \bbx_{[i-1]}$, then $\bbx_{[i-1]}~= ~\bbx_{[j-1]},t_i,\bbx_{[j+1\,:\,i-1]}$. In these cases, couple $\bbx_{[i-1]},x_i$ with $\bbx_{[j-1]},x_i,\bbx_{[j+1\,:\,i-1]},t_i$. Notice that the following is true because of the order invariance of the algorithm:\footnote{If the algorithm were not order invariant we could instead double the parameters to get $\prob{(\cA(\bX),x_i ) \in \cO  \mid \bbx_{[i-1]}, x_i} \leq e^{2\eps} \prob{(\cA(\bX),x_i ) \in \cO \mid \bbx_{[i-1]}, t_i} + 2\delta$ by invoking differential privacy twice in the latter case.}
$$\prob{(\cA(\bX),x_i ) \in \cO  \mid \bbx_{[j-1]}, t_i, \bbx_{[j+1\,:\,i-1]}, x_i} = \prob{(\cA(\bX),x_i ) \in \cO \mid \bbx_{[j-1]}, x_i, \bbx_{[j+1\,:\,i-1]}, t_i}$$
\end{proof}

Having shown a high probability bound on the terms $Z_i$, our next step is to bound their expectation so that we can continue towards our goal of applying Azuma's inequality. 

We will use the following shorthand notation for conditional expectation:
\begin{align*}& \Ex{}{Z_i(\cA,\bX_{[i]})\mid a,\bbx_{[i-1]}, \abs{Z_i}\leq 6\eps} \\ & \quad 
\stackrel{def}{=} \Ex{}{Z_i(\cA,\bX_{[i]})\mid \cA=a,\bX_{[i-1]}=\bbx_{[i-1]}, \abs{Z_i(\cA,\bX_{[i]})}\leq 6\eps}, 
\end{align*}

\begin{lemma}\label{lem:exp_Z}
Let $\cA$ be $(\eps,\delta)$-differentially private and, $\bX \in \univ^n$ be sampled without replacement from a finite universe $\univ$ . Let $\eps \in (0,1/2]$ and $\hat{\delta} \in \left (0,\eps/15 \right ]$,
\ifnum\usenix=0
\begin{equation*}
{X_i|_{\bbx_{[i-1]}} \approx_{3\eps,\hat{\delta}} X_i|_{a, \bbx_{[i-1]}}}\quad \implies  \qquad \Ex{}{Z_i(\cA,\bX_{[i]})\mid a,\bbx_{[i-1]}, \abs{Z_i}\leq 6\eps} = O(\eps^2 + \hat\delta).
\end{equation*}
\else
\begin{align*}
X_i|_{\bbx_{[i-1]}} & \approx_{3\eps,\hat{\delta}}  X_i|_{a, \bbx_{[i-1]}} \\
& => \Ex{}{Z_i(\cA,\bX_{[i]})\mid a,\bbx_{[i-1]}, \abs{Z_i}\leq 6\eps} = O(\eps^2 + \hat\delta).
\end{align*}
\fi

More precisely, $\Ex{}{Z_i(\cA,\bX_{[i]})\mid a,\bbx_{[i-1]}, \abs{Z_i}\leq 6\eps} \leq \nu (\hat{\delta})$, where $\nu (\hat{\delta})$ is defined in \eqref{eqn: nu}.
\end{lemma}

\begin{proof}
Let $S \defeq \{x_i\mid a,\bbx_{[i-1]},\abs{Z_i}<6\eps\}$. Given an outcome and prefix $(a,\bbx_{[i-1]})$ such that ${X_i|_{\bbx_{[i-1]}} \approx_{3\eps,\hat{\delta}} X_i|_{a, \bbx_{[i-1]}}}$, we have the following by definition:
\begin{align*}
    &\Ex{}{Z_i(\cA,\bX_{[i]})\mid a,\bbx_{[i-1]}, \abs{Z_i}\leq 6\eps}\\ 
    &=  \sum_{x_i\in S} \prob{X_i = x_i\mid a,\bbx_{[i-1]}, \abs{Z_i}\leq 6\eps} Z_i(a,x_{[i]})\\
    &=  \sum_{x_i\in S} \prob{X_i = x_i\mid a,\bbx_{[i-1]}, \abs{Z_i}\leq 6\eps} \log\left( \tfrac{\prob{X_i = x_i \mid a,\bbx_{[i-1]}}}{\prob{X_i = x_i \mid \bbx_{[i-1]}}} \right)
\end{align*}
\begin{claim}
\ifnum\usenix=0
    $$ \sum_{x_i\in S} \prob{X_i = x_i\mid \bbx_{[i-1]}} \log\left( \tfrac{\prob{X_i = x_i \mid a,\bbx_{[i-1]}}}{\prob{X_i = x_i\mid\bbx_{[i-1]}}} \right) \leq \log\left(\tfrac{1 - \prob{X_i \notin S\mid a,\bbx_{[i-1]}}}{1- \prob{X_i \notin S\mid \bbx_{[i-1]}}} \right)$$
\else
\begin{multline*}
    \sum_{x_i\in S} \prob{X_i = x_i\mid \bbx_{[i-1]}} \log\left( \tfrac{\prob{X_i = x_i \mid a,\bbx_{[i-1]}}}{\prob{X_i = x_i\mid\bbx_{[i-1]}}} \right) \\
    \leq \log\left(\tfrac{1 - \prob{X_i \notin S\mid a,\bbx_{[i-1]}}}{1- \prob{X_i \notin S\mid \bbx_{[i-1]}}} \right)
\end{multline*}
\fi
\end{claim}
\begin{proof}
\begin{align*}
& \sum_{x_i\in S} \prob{X_i = x_i\mid a,\bbx_{[i-1]}} \log\left( \tfrac{\prob{X_i = x_i \mid a,\bbx_{[i-1]}}}{\prob{X_i = x_i\mid\bbx_{[i-1]}}} \right)\\ 
&=\prob{X_i \in S\mid \bbx_{[i-1]}}  \sum_{x_i\in S}\tfrac{\prob{X_i = x_i\mid \bbx_{[i-1]}}}{\prob{X_i \in S\mid \bbx_{[i-1]}}}  \log\left( \tfrac{\prob{X_i = x_i \mid a,\bbx_{[i-1]}}}{\prob{X_i = x_i\mid\bbx_{[i-1]}}} \right)\\
&\leq  \sum_{x_i\in S}\tfrac{\prob{X_i = x_i\mid \bbx_{[i-1]}}}{\prob{X_i \in S\mid \bbx_{[i-1]}}}  \log\left( \tfrac{\prob{X_i = x_i \mid a,\bbx_{[i-1]}}}{\prob{X_i = x_i\mid\bbx_{[i-1]}}} \right)\\
&\leq \log\left( \sum_{x_i\in S}\tfrac{\prob{X_i = x_i \mid a,\bbx_{[i-1]}}}{\prob{X_i \in S\mid \bbx_{[i-1]}}}\right)\\
&\leq \log\left(\tfrac{\prob{X_i \in S \mid a,\bbx_{[i-1]}}}{\prob{X_i \in S\mid \bbx_{[i-1]}}}\right) = \log\left(\tfrac{1 - \prob{X_i \notin S \mid a,\bbx_{[i-1]}}}{1- \prob{X_i \notin S\mid \bbx_{[i-1]}}}\right)\\
\end{align*}
The first inequality follows form the fact that all probabilities are less than one. The second inequality follows from noticing that $ \sum_{x_i\in S}\frac{\prob{X_i = x_i\mid \bbx_{[i-1]}}}{\prob{X_i \in S\mid \bbx_{[i-1]}}} = 1$ and applying Jensen's inequality. 
\end{proof}
Let $\prob{X_i \notin  \{x_i\mid a,\bbx_{[i-1],\abs{Z_i}<6\eps}\} \mid a,\bbx_{[i-1]}} = \prob{X_i \notin S \mid a,\bbx_{[i-1]}} \defeq q$. Note that, because ${X_i|_{\bbx_{[i-1]}} \approx_{3\eps,\hat{\delta}} X_i|_{a, \bbx_{[i-1]}}}$, we have for $\hat\delta>0$:
$$
\prob{X_i \notin S\mid \bbx_{[i-1]}} \leq e^{3\eps} \prob{X_i \notin S \mid a,\bbx_{[i-1]}} + \hat\delta = e^{3\eps}q + \hat\delta
$$
Note that $q \leq \delta''$ by \Cref{claim:Fi}. Now, we can bound the following:
\begin{align*}
\sum_{x_i \in  S} & \prob{X_i = x_i \mid \bbx_{[i-1]}}  \log\left(\tfrac{\prob{X_i = x_i \mid a,\bbx_{[i-1]}}}{\prob{X_i = x_i \mid \bbx_{[i-1]}}} \right)\\
&\leq \log\left(\tfrac{1 - \prob{X_i \notin S \mid a,\bbx_{[i-1]}}}{1- \prob{X_i \notin S\mid \bbx_{[i-1]}}}\right)\\
&\leq \log(1-q) - \log(1- (e^{3\eps}q + \hat{\delta})  ) \\
& \leq \log(e) \cdot(- q+ e^{3\eps}q + \hat{\delta} + 2(e^{3\eps}q + \hat{\delta})^2)\\
&= \log(e) \cdot ( (e^{3\eps}-1)q + \hat{\delta} + 2(e^{3\eps}q + \hat{\delta})^2 )\\
&\defeq \tau(\hat\delta)
\end{align*}
where the second inequality follows by using the inequality $(-x - 2x^2)\log(e) \leq \log(1-x) \leq -x\log(e)$ for $0< x \leq 1/2$, and as $(e^{3\eps}q + \hat{\delta}) \leq 1/2$ for $\eps$ and $\hat\delta$ bounded as in the lemma statement.

We use the results above to to upper bound the expectation we wanted:
\ifnum\usenix=0
{\footnotesize
\allowdisplaybreaks[2]
\begin{align*}
& \Ex{}{Z_i(\cA,\bX_{[i]})\mid a,\bbx_{[i-1]}, \abs{Z_i}\leq 6\eps} \\
& \leq \sum_{x_i \in  S} \prob{X_i = x_i\mid a,\bbx_{[i-1]},\abs{Z_i}\leq 6\eps} \log\left( \tfrac{\prob{X_i = x_i \mid a,\bbx_{[i-1]}}}{\prob{X_i = x_i  \mid \bbx_{[i-1]}}} \right)  \\
& \qquad - \sum_{x_i \in  S} \prob{X_i = x_i  \mid \bbx_{[i-1]}} \log\left(\tfrac{\prob{X_i = x_i \mid a,\bbx_{[i-1]}}}{\prob{X_i = x_i  \mid \bbx_{[i-1]}}} \right) + \tau(\hat\delta) \\
& = \sum_{x_i \in  S} \left(\prob{X_i = x_i\mid a,\bbx_{[i-1]},\abs{Z_i}\leq 6\eps} - \prob{X_i = x_i  \mid \bbx_{[i-1]}}\right) \log\left( \tfrac{\prob{X_i = x_i \mid  a,\bbx_{[i-1]}}}{\prob{X_i = x_i  \mid \bbx_{[i-1]}}} \right)+ \tau(\hat\delta) \\
&\leq_{|Z_i| \leq 6\eps } 6\eps \sum_{x_i \in  S} \abs{\prob{X_i = x_i\mid a,\bbx_{[i-1]},\abs{Z_i}\leq 6\eps} - \prob{X_i = x_i  \mid \bbx_{[i-1]}}}+ \tau(\hat\delta) \\
& \leq_{\text{Def of $S$,  Claim~\ref{claim:zibound}}} 6\eps \sum_{x_i \in  S} \prob{X_i = x_i  \mid \bbx_{[i-1]}} \max\Bigg\{ \dfrac{e^{6\eps} }{\prob{|Z_i| \leq 6\eps \mid a,\bbx_{[i-1]}}} - 1, 1 - \dfrac{e^{-6\eps} }{\prob{|Z_i| \leq 6\eps \mid a,\bbx_{[i-1]}}}  \Bigg\} + \tau(\hat\delta) \\ 
& \leq_{\text{Claim~\ref{claim:zibound}}} 6\eps \left(\tfrac{e^{6\eps}}{1- \tfrac{2\hat\delta}{1-e^{-3\eps}}} - 1\right) + \tau(\hat\delta)\\
&\leq_{\text{Substituting for $\tau(\hat\delta)$}} 6\eps\left(e^{6\epsilon} \left( 1+ \tfrac{4\hat\delta}{1-e^{-3\epsilon} } \right) - 1 \right) + \log(e) \cdot ( (e^{3\eps}-1)q + \hat{\delta} + 2(e^{3\eps}q + \hat{\delta})^2 ) \\
& \leq_{\text{Upper bound for $q$}} 
6\eps\left(e^{6\epsilon} \left( 1+ \tfrac{4\hat\delta}{1-e^{-3\epsilon} } \right) - 1 \right) \\
& \qquad + \log(e) \cdot \left( (e^{3\eps}-1)\frac{2 \hat{\delta}}{1-e^{-3\epsilon}} + \hat{\delta} + 2 \hat{\delta}^2 + 8\frac{\hat{\delta}^2 e^{6 \eps}}{(1-e^{-3\eps})^2} + 8\frac{\hat{\delta}^2 e^{3 \eps}}{1-e^{-3\eps}} \right) \\
& =_{ b = \frac{\hat{\delta}}{1-e^{-3\eps}} } 
6\eps\left(e^{6\epsilon} \left( 1+ 4b \right) - 1 \right) + \log(e) \cdot \left( b \left( 2e^{3\eps} - 2  + 8\frac{\hat{\delta} e^{6 \eps}}{(1-e^{-3\eps})} + 8\hat{\delta} e^{3 \eps}\right) + \hat{\delta} + 2\hat{\delta}^2 \right) \\
& = 
b \left( 24\eps e^{6\eps} + 2e^{3\eps} - 2  + 8\frac{\hat{\delta} e^{6 \eps}}{(1-e^{-3\eps})} + 8\hat{\delta} e^{3 \eps}\right) + \hat{\delta} + 2\hat{\delta}^2 + 6\eps(e^{6\eps} - 1) \\
& = \frac{\hat{\delta}}{1-e^{-3\eps}}  \left( 2e^{3 \eps} (4e^{3\eps}(3\eps + \frac{\hat{\delta}}{1-e^{-3\eps}} )) + 4\hat{\delta} + 1) - 2 \right) + \hat{\delta} + 2\hat{\delta}^2 + 6\eps(e^{6\eps} - 1) \\
& \leq_{{e^{-3\eps}} \leq 1-1.5\eps \text{ for } \eps \in [0,0.5]}
2 \frac{\hat{\delta}}{1.5\eps} e^{3 \eps} \left (4e^{3\eps}(3\eps + \frac{\hat{\delta}}{1.5\eps})) + 4\hat{\delta} + 1\right) + \hat{\delta} \left( \frac{-2}{1.5\eps} + 2\hat{\delta} + 1 \right) + 6\eps(e^{6\eps} - 1) \\
& \leq
8\frac{e^{6\eps} \hat{\delta}}{1.5\eps} \left (3\eps + \frac{\hat{\delta}}{1.5\eps} \right) + 2 \frac{\hat{\delta}}{1.5\eps} e^{3 \eps} \left( 4\hat{\delta} + 1\right) +  \hat{\delta} \left( \frac{-2}{1.5\eps} + 2\hat{\delta} + 1 \right) + 6\eps(e^{6\eps} - 1) \\ 
& \leq_{{e^{3\eps}} \leq 1+7\eps, {e^{6\eps}} \leq 1+40\eps \text{ for } \eps \in [0,0.5]} 8\frac{(1+40\eps) \hat{\delta}}{1.5\eps} \left (3\eps + \frac{\hat{\delta}}{1.5\eps} \right) + 2 \frac{(1+7\eps)\hat{\delta}}{1.5\eps} \left( 4\hat{\delta} + 1\right) \\
& \qquad +   \hat{\delta} \left( \frac{-2}{1.5\eps} + 2\hat{\delta} + 1 \right) + 6\eps(40\eps) \\
& \leq \frac{(8+320\eps) \hat{\delta}}{1.5\eps} \left (3\eps + \frac{\hat{\delta}}{1.5\eps} \right) + \frac{(2+14\eps)}{1.5\eps} \left( 4\hat{\delta}^2 + \hat{\delta}\right) - \frac{2\hat{\delta}}{1.5\eps} + 2\hat{\delta}^2 + \hat{\delta} + 240\eps^2\\
& \leq_{\eps < 0.5}  \frac{168\hat{\delta}}{1.5\eps}(3\eps + \frac{\hat{\delta}}{1.5\eps} ) + \frac{8}{1.5}\frac{\hat{\delta}^2}{\eps} + \frac{56}{1.5}\hat{\delta}^2 + \frac{14}{1.5}\hat{\delta} + 2\hat{\delta}^2 + \hat{\delta} + 240\eps^2 \\
& \leq_{\eps \leq 0.5} 347\hat{\delta} + 75 \left(\frac{\hat{\delta}}{\eps} \right)^2 + 24\frac{\hat{\delta}^2}{\eps}+ 240\eps^2\\
& \defeq \nu(\hat\delta) \numberthis \label{eqn: nu}
\end{align*} 
}
\else

\begin{align*}
& \Ex{}{Z_i(\cA,\bX_{[i]})\mid a,\bbx_{[i-1]}, \abs{Z_i}\leq 6\eps} \\
& \leq \sum_{x_i \in  S} \prob{X_i = x_i\mid a,\bbx_{[i-1]},\abs{Z_i}\leq 6\eps} \log\left( \tfrac{\prob{X_i = x_i \mid a,\bbx_{[i-1]}}}{\prob{X_i = x_i  \mid \bbx_{[i-1]}}} \right)  \\
& \qquad - \sum_{x_i \in  S} \prob{X_i = x_i  \mid \bbx_{[i-1]}} \log\left(\tfrac{\prob{X_i = x_i \mid a,\bbx_{[i-1]}}}{\prob{X_i = x_i  \mid \bbx_{[i-1]}}} \right) + \tau(\hat\delta) \\
& = \sum_{x_i \in  S} \left(\prob{X_i = x_i\mid a,\bbx_{[i-1]},\abs{Z_i}\leq 6\eps} - \prob{X_i = x_i  \mid \bbx_{[i-1]}}\right) \\
& \qquad \cdot \log\left( \tfrac{\prob{X_i = x_i \mid  a,\bbx_{[i-1]}}}{\prob{X_i = x_i  \mid \bbx_{[i-1]}}} \right)+ \tau(\hat\delta) \\
&\leq_{|Z_i| \leq 6\eps } 6\eps \sum_{x_i \in  S} \abs{\prob{X_i = x_i\mid a,\bbx_{[i-1]},\abs{Z_i}\leq 6\eps} \\
& \qquad \qquad \qquad - \prob{X_i = x_i  \mid \bbx_{[i-1]}}}+ \tau(\hat\delta) \\
& \leq_{\text{Definition of $S$, Proof of Claim~\ref{claim:zibound}}} 6\eps \sum_{x_i \in  S} \prob{X_i = x_i  \mid \bbx_{[i-1]}} \\
& \qquad \cdot \max\Bigg\{ \dfrac{e^{6\eps} }{\prob{|Z_i| \leq 6\eps \mid a,\bbx_{[i-1]}}} - 1 \\
& \qquad, 1 - \dfrac{e^{-6\eps} }{\prob{|Z_i| \leq 6\eps \mid a,\bbx_{[i-1]}}}  \Bigg\} + \tau(\hat\delta) \\ 
& \leq_{\text{Claim~\ref{claim:zibound}}} 6\eps \left(\tfrac{e^{6\eps}}{1- \tfrac{2\hat\delta}{1-e^{-3\eps}}} - 1\right) + \tau(\hat\delta)\\
&\leq_{\text{Substituting for $\tau(\hat\delta)$}} 6\eps\left(e^{6\epsilon} \left( 1+ \tfrac{4\hat\delta}{1-e^{-3\epsilon} } \right) - 1 \right) \\
& \qquad + \log(e) \cdot ( (e^{3\eps}-1)q + \hat{\delta} + 2(e^{3\eps}q + \hat{\delta})^2 ) \\
& \leq_{\text{Upper bound for $q$}} 
6\eps\left(e^{6\epsilon} \left( 1+ \tfrac{4\hat\delta}{1-e^{-3\epsilon} } \right) - 1 \right) \\
& \qquad + \log(e) \cdot \Bigg( (e^{3\eps}-1)\frac{2 \hat{\delta}}{1-e^{-3\epsilon}} + \hat{\delta} + 2 \hat{\delta}^2 \\
& \qquad + 8\frac{\hat{\delta}^2 e^{6 \eps}}{(1-e^{-3\eps})^2} + 8\frac{\hat{\delta}^2 e^{3 \eps}}{1-e^{-3\eps}} \Bigg) \\
& =_{ b = \frac{\hat{\delta}}{1-e^{-3\eps}} } 
6\eps\left(e^{6\epsilon} \left( 1+ 4b \right) - 1 \right) + \log(e) \cdot \\
& \qquad \qquad \left( b \left( 2e^{3\eps} - 2  + 8\frac{\hat{\delta} e^{6 \eps}}{(1-e^{-3\eps})} + 8\hat{\delta} e^{3 \eps}\right) + \hat{\delta} + 2\hat{\delta}^2 \right) \\
& = 
b \left( 24\eps e^{6\eps} + 2e^{3\eps} - 2  + 8\frac{\hat{\delta} e^{6 \eps}}{(1-e^{-3\eps})} + 8\hat{\delta} e^{3 \eps}\right) + \hat{\delta} \\
& \qquad +  2\hat{\delta}^2 + 6\eps(e^{6\eps} - 1) \\
& = \frac{\hat{\delta}}{1-e^{-3\eps}}  \left( 2e^{3 \eps} (4e^{3\eps}(3\eps + \frac{\hat{\delta}}{1-e^{-3\eps}} )) + 4\hat{\delta} + 1) - 2 \right) \\
& \qquad + \hat{\delta} + 2\hat{\delta}^2 + 6\eps(e^{6\eps} - 1) \\
& \leq_{{e^{-3\eps}} \leq 1-1.5\eps \text{ for } \eps \in [0,0.5]}
2 \frac{\hat{\delta}}{1.5\eps} e^{3 \eps} \left (4e^{3\eps}(3\eps + \frac{\hat{\delta}}{1.5\eps})) + 4\hat{\delta} + 1\right) \\
& \qquad + \hat{\delta} \left( \frac{-2}{1.5\eps} + 2\hat{\delta} + 1 \right) + 6\eps(e^{6\eps} - 1) 
\end{align*}
\begin{align*}
\leq & 8\frac{e^{6\eps} \hat{\delta}}{1.5\eps} \left (3\eps + \frac{\hat{\delta}}{1.5\eps} \right) + 2 \frac{\hat{\delta}}{1.5\eps} e^{3 \eps} \left( 4\hat{\delta} + 1\right) \\
& \qquad +  \hat{\delta} \left( \frac{-2}{1.5\eps} + 2\hat{\delta} + 1 \right) + 6\eps(e^{6\eps} - 1) \\ 
& \leq_{{e^{3\eps}} \leq 1+7\eps, {e^{6\eps}} \leq 1+40\eps \text{ for } \eps \in [0,0.5]} 8\frac{(1+40\eps) \hat{\delta}}{1.5\eps} \left (3\eps + \frac{\hat{\delta}}{1.5\eps} \right) \\
& \qquad + 2 \frac{(1+7\eps)\hat{\delta}}{1.5\eps} \left( 4\hat{\delta} + 1\right) + \hat{\delta} \left( \frac{-2}{1.5\eps} + 2\hat{\delta} + 1 \right) + 6\eps(40\eps) \\
& \leq \frac{(8+320\eps) \hat{\delta}}{1.5\eps} \left (3\eps + \frac{\hat{\delta}}{1.5\eps} \right) + \frac{(2+14\eps)}{1.5\eps} \left( 4\hat{\delta}^2 + \hat{\delta}\right) \\
& \qquad - \frac{2\hat{\delta}}{1.5\eps} + 2\hat{\delta}^2 + \hat{\delta} + 240\eps^2 \\
& \leq_{\eps < 0.5}  \frac{168\hat{\delta}}{1.5\eps}(3\eps + \frac{\hat{\delta}}{1.5\eps} ) + \frac{8}{1.5}\frac{\hat{\delta}^2}{\eps} + \frac{56}{1.5}\hat{\delta}^2 \\
& + \frac{14}{1.5}\hat{\delta} + 2\hat{\delta}^2 + \hat{\delta} + 240\eps^2 \\
& \leq_{\eps \leq 0.5} 347\hat{\delta} + 75 \left(\frac{\hat{\delta}}{\eps} \right)^2 + 24\frac{\hat{\delta}^2}{\eps}+ 240\eps^2
\\
& \defeq \nu(\hat\delta) \numberthis \label{eqn: nu}
\end{align*}
\fi
\end{proof}

Finally, we need to apply Azuma's inequality (stated in \Cref{thm:azuma}) to a set of variables that are bounded with probability $1$, not just with high probability. Towards this end, we now define (1) the sets $\cG_i(\hat\delta)$ and $\cG_{\leq i}(\hat\delta)$ of ``good'' tuples of outcomes and databases, and (2) a variable $T_i$ that will match $Z_i$ for ``good events'', and will be zero otherwise---and hence, is always bounded:
\begin{align}
\cG_i(\hat\delta) = \left\{(a,\bbx_{[i]}) \,\,\Big|\quad  | Z_i(a,\bbx_{[i]}) | \leq 6 \eps \quad \& \quad X_i\mid_{\bbx_{[i-1]}} \approx_{3\eps,\hat{\delta}} X_i|_{a, \bbx_{[i-1]}}   \right\},
\label{eq:G}\\
\cG_{\leq i}(\hat\delta)   = \left\{(a,\bbx_{[i]}) : (a,x_1) \in \cG_1(\hat\delta), \cdots, (a,\bbx_{[i]}) \in \cG_i(\hat\delta) \right\} \label{eq:Gvect}
\end{align}
\begin{equation}
T_i(a,\bbx_{[i]}) = \begin{cases}
             Z_i(a,\bbx_{[i]})  & \text{if }   
             (a,\bbx_{[i]}) \in \cG_{\leq i}(\hat\delta) \\
             0  & \text{otherwise }
       \end{cases}
\label{eq:Ti}
\end{equation}

Note that the variables $T_i$ indeed satisfy the requirements of Azuma's inequality. The first condition, $\Prob{}{\abs{T_i(\cA,\bX_{[i]})} \leq 6 \eps} = 1$ holds by definition, and the second holds because of Lemma~\ref{lem:exp_Z}.

We are now ready to prove our main theorem.
\begin{proof}[Proof of \Cref{thm:approx-dp-implies-max-info}]
For any constant $\nu$, we have:
\ifnum\usenix=0
\begin{align*}
& \Prob{}{\sum\limits_{i=1}^n Z_i(\cA,\bX_{[i]})  > n\nu + 6t\eps\sqrt{n}} \\
& \leq \Prob{}{\sum\limits_{i=1}^n Z_i(\cA,\bX_{[i]}) > n\nu + 6t\eps\sqrt{n} \cap (\cA,\bX) \in \cG_{\leq n}(\hat\delta)}
+\Prob{}{(\cA,\bX) \notin \cG_{\leq n}(\hat\delta)} \\
& =  \Prob{}{\sum\limits_{i=1}^n T_i(\cA,\bX_{[i]}) > n\nu + 6t\eps\sqrt{n} \cap (\cA,\bX) \in  \cG_{\leq n}(\hat\delta) } 
+ \Prob{}{(\cA,\bX) \notin  \cG_{\leq n}(\hat\delta)}
\end{align*}
\else
\begin{align*}
& \Prob{}{\sum\limits_{i=1}^n Z_i(\cA,\bX_{[i]})  > n\nu + 6t\eps\sqrt{n}} \\
& \leq \Prob{}{\sum\limits_{i=1}^n Z_i(\cA,\bX_{[i]}) > n\nu + 6t\eps\sqrt{n} \cap (\cA,\bX) \in \cG_{\leq n}(\hat\delta)}
\\ 
&\quad
+\Prob{}{(\cA,\bX) \notin \cG_{\leq n}(\hat\delta)} \\
& =  \Prob{}{\sum\limits_{i=1}^n T_i(\cA,\bX_{[i]}) > n\nu + 6t\eps\sqrt{n} \cap (\cA,\bX) \in  \cG_{\leq n}(\hat\delta) } 
\\ 
&\quad
+ \Prob{}{(\cA,\bX) \notin  \cG_{\leq n}(\hat\delta)}
\end{align*}
\fi
We then substitute $\nu$ by $\nu(\hat\delta)$ as defined in \Cref{eqn: nu}, and apply a union bound on $\prob{(\cA,\bX) \notin  \cG_{\leq n}(\hat\delta)}$ using \Cref{claim:Gi} to get
\ifnum\usenix=0
\begin{align*}
 \Prob{}{\sum\limits_{i=1}^n Z_i(\cA,\bX_{[i]})  > n\nu(\hat\delta) + 6t\eps\sqrt{n}}
 &\leq \Prob{}{\sum\limits_{i=1}^n T_i(\cA,\bX_{[i]}) > n\nu(\hat\delta) + 6t\eps\sqrt{n}  } + n(\delta' + \delta'')\\
&\leq e^{-t^2/2} + n(\delta' + \delta'')
\end{align*}
\else
\begin{align*}
&\quad
 \Prob{}{\sum\limits_{i=1}^n Z_i(\cA,\bX_{[i]})  > n\nu(\hat\delta) + 6t\eps\sqrt{n}}
\\ 
 &\leq \Prob{}{\sum\limits_{i=1}^n T_i(\cA,\bX_{[i]}) > n\nu(\hat\delta) + 6t\eps\sqrt{n}  } + n(\delta' + \delta'')\\
&\leq e^{-t^2/2} + n(\delta' + \delta'')
\end{align*}
\fi
where the two inequalities follow from \Cref{claim:Gi} and \Cref{thm:azuma}, respectively.  Therefore,
\ifnum\usenix=0
\begin{align*}
\Prob{ }{Z(\cA(\bX),\bX) > n\nu(\hat\delta) + 6t\eps\sqrt{n}} 
\leq e^{-t^2/2} + n(\delta' + \delta'')  \quad\stackrel{def}{=} \beta(t,\hat\delta)
\end{align*}
\else
\begin{align*}
&\quad
\Prob{ }{Z(\cA(\bX),\bX) > n\nu(\hat\delta) + 6t\eps\sqrt{n}} 
\\ 
&
\leq e^{-t^2/2} + n(\delta' + \delta'')  \quad\stackrel{def}{=} \beta(t,\hat\delta)
\end{align*}
\fi
From \Cref{lem:boundmaxinfo}, we have
$I^{\beta(t,\hat\delta)}_\infty(\bX;\cA(\bX)) \leq n\nu(\hat\delta)
+ 6t\eps\sqrt{n}.$  

\end{proof}

\noindent
We now prove the \Cref{cor:approx-dp-implies-max-info}, which we use in \Cref{sec:dp-implies-coherence-enforcement} to prove \Cref{thm:approx-dp-implies-coherence-enforcement}.

\begin{proof}[Proof of~\Cref{cor:approx-dp-implies-max-info}]
    Setting $t = \sqrt{2 \ln(2/\gamma)}$, and $\hat{\delta} = \frac{\sqrt{\eps \delta}}{15}$, in Theorem~\ref{thm:approx-dp-implies-max-info}, we get that $\beta(t,\hat{\delta}) \leq \gamma/2 + 30n\sqrt{\delta/\eps} + n\frac{2\sqrt{\eps \delta}+2\delta}{1.5\eps}$, where we've used that $1-e^{-3\eps} \geq 1.5\eps$ for $\eps \in (0,1/2]$. We note that for $\delta \leq \frac{\eps^2\gamma^2}{(120n)^2}$, we have that $n\frac{2\sqrt{\eps \delta}+2\delta}{1.5\eps} \leq \frac{\gamma}{2}$, ensuring that $\beta(t,\hat{\delta}) \leq \gamma$. Also, the same bound on $\delta$ ensures that $n\left( 347\hat{\delta} + 75 \left(\frac{\hat{\delta}}{\eps} \right)^2 + 24\frac{\hat{\delta}^2}{\eps}+ 240\eps^2\right) \leq   265\eps^2n$. Substituting for $t$ directly in the max-information bound then completes the proof. 
\end{proof}

\fi

\end{document}